\documentclass[a4paper, english, cleveref, autoref, thm-restate, nolineno]{socg-lipics-v2021}
\hideLIPIcs

\newcommand{\OO}{O}
\newcommand{\oo}{o}

\newcommand{\cyc}{C}
\newcommand{\cycA}{C}
\newcommand{\cycB}{C'}

\newcommand{\ntd}[1]{\operatorname{npd}(#1)}

\newcommand{\bag}[1]{X_{#1}}

\newcommand{\aspace}[1]{Y_{#1}}

\newcommand{\tdspace}[1]{\operatorname{pd}(Y_{#1})}

\newcommand{\DCP}{\textsc{Directed Cycle Packing }} 
\newcommand{\CSR}{\textsc{Connected Subsurface Recognition }}
\newcommand{\CSRdot}{\textsc{Connected Subsurface Recognition}}
\newcommand{\SR}{\textsc{Subsurface Recognition }}
\newcommand{\SRdot}{\textsc{Subsurface Recognition}}
\newcommand{\SoGSR}{\textsc{Sum-of-Genus Subsurface Recognition }}
\newcommand{\SoGSRdot}{\textsc{Sum-of-Genus Subsurface Recognition}}
\newcommand{\SP}{\textsc{Subsurface Packing }}

\newcommand{\D}{\mathbf{D}} 
\newcommand{\V}{\mathbf{V}} 
\newcommand{\R}{\mathbb{R}} 
\newcommand{\boundary}{\partial}
\newcommand{\restr}[2]{\left.{#1}\right|_{#2}} 
\newcommand{\overbar}[1]{\overline{\vphantom{V}#1}} 

\DeclareMathOperator{\im}{im}     
\DeclareMathOperator{\cl}{cl}     
\DeclareMathOperator{\lk}{lk}     
\DeclareMathOperator{\poly}{poly} 
\newcommand{\str}[2]{\text{st}_{#1}\,#2} 
\newtheorem{problem}{Problem}
\newtheorem{hypothesis}{Hypothesis}

\usepackage{float}    
\usepackage{enumitem} 
\usepackage[dvipsnames, table]{xcolor} 
\title{ETH-tight algorithms for finding surfaces in simplicial complexes of bounded treewidth}

\author{Mitchell Black}{School of Electrical Engineering and Computer Science, Oregon State University, USA}{blackmit@oregonstate.edu}{}{}

\author{Nello Blaser}{Department of Informatics, University of Bergen, Norway }{nello.blaser@uib.no}{https://orcid.org/0000-0001-9489-1657}{}

\author{Amir Nayyeri}{School of Electrical Engineering and Computer Science, Oregon State University, USA}{nayyeria@eecs.oregonstate.edu }{}{}

\author{Erlend Raa Vågset}{Department of Informatics, University of Bergen, Norway }{erlend.vagset@uib.no}{}{}

\authorrunning{M. Black and N. Blaser and A. Nayyeri and E.\,R. Vågset} 

\Copyright{Mitchell Black and Nello Blaser and Amir Nayyeri and Erlend Raa Vågset} 



\relatedversion{} 

\begin{document}

\maketitle

\begin{abstract}
    Given a simplicial complex with \(n\) simplices, we consider the \CSR (c-SR) problem of finding a subcomplex that is homeomorphic to a given connected surface with a fixed boundary. We also study the related \SoGSR (SoG) problem, where we instead search for a surface whose boundary, number of connected components, and total genus are given. 
    For both of these problems, we give parameterized algorithms with respect to the treewidth \(k\) of the Hasse diagram that run in \(2^{\OO(k \log k)}n^{\OO(1)}\) time. For the SoG problem, we also prove that our algorithm is optimal assuming the exponential-time hypothesis. In fact, we prove the stronger result that our algorithm is ETH-tight even without restriction on the total genus. 
\end{abstract}


\section{Introduction}
Simplicial complexes are a generalization of graphs that give a discrete representation of higher-dimensional spaces. A natural and interesting class of such spaces are manifolds. A $d$-manifold is a space that is ``locally $d$-dimensional'', meaning each point has a neighborhood homeomorphic to $\R^{d}$. In particular, circles are 1-manifolds and spheres are 2-manifolds. Manifolds are important in both mathematics and computer science. For example, triangular meshes in computer graphics are typically 2-manifolds, and the manifold hypothesis in machine learning is the assumption that real-world data often lie on low-dimensional submanifolds of high-dimensional spaces. 
\par 
Since manifolds are so important, it is natural to ask if a given simplicial complex is a manifold, or whether two manifolds are homeomorphic. There are fascinating complexity results on these problems.  While both recognizing and classifying a $2$-manifold have polynomial algorithms, this problem becomes much harder for arbitrary $d$-manifolds. Deciding whether two manifolds are homeomorphic is undecidable for $d\geq 4$ \cite{markov_insolubility}. Deciding whether or not a simplicial complex is homeomorphic to the $d$-sphere is undecidable for $d\geq 5$ (see \cite{chernavsky_unrecognizability}), which implies deciding whether or not a simplicial complex is an $n$-manifold is undecidable for $d\geq 6$.

\par
We consider several variants of the problem of finding subcomplexes homeomorphic to 2-manifolds, or \textit{surfaces}, in simplicial complexes. While there are polynomial time algorithms for deciding if a simplicial complex is homeomorphic to a surface or deciding the homeomorphism class of a surface, it is a hard problem deciding whether or not a simplicial complex contains a surface as a subcomplex. In particular, Ivanov proved that it is NP-Hard to decide if a simplicial complex contains a 2-sphere \cite{ivanov-hardness}, and Burton et al.~proved that finding a 2-sphere is W[1]-hard when parameterized by solution size \cite{burton-finding-two-sphere}.  The complexity of this problem is analogous to the graph isomorphism problem. While there is a quasipolynomial algorithm to determine if two graphs are isomorphic~\cite{babai_isomorphism}, it is NP-Hard to determine if one graph contains a subgraph isomorphic to another graph~\cite{cook_complexity}.
\par 
As this problem is NP-Hard, it is natural to ask whether there is any class of simplicial complexes for which polynomial time algorithms exist. In this paper, we consider the parameterized complexity of this problem and related problems with respect to the treewidth of the Hasse diagram. A tree decomposition of the Hasse diagram defines a recursively nested series of subcomplexes of $K$ that we can use to incrementally build our surfaces. We also give tight lower bounds for a subset of our algorithms based on the Exponential Time Hypothesis.

\subsection{Subsurface Recognition Problems}

\begin{figure}[!ht]
    \centering
    \includegraphics[width = 0.9\textwidth]{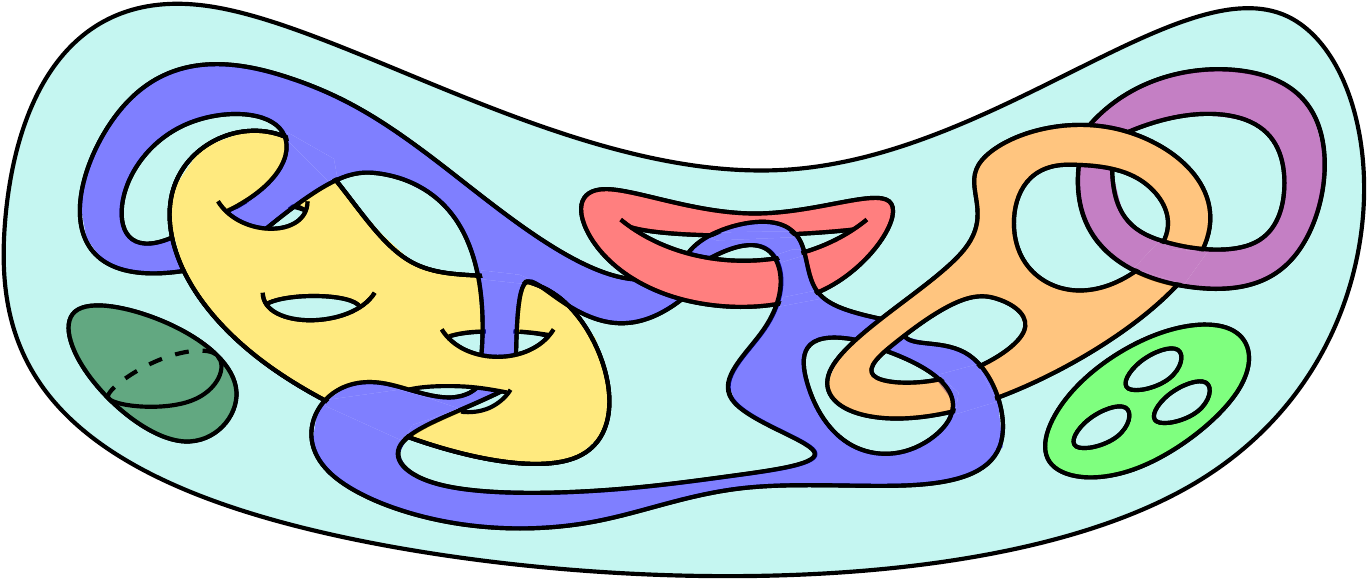}
    \caption{A solution to an instance of the \SR problem where we have found an orientable surface consisting of seven connected components with genus {\color[rgb]{0.01, 0.5, 0.1}\textbf{0}}, {\color{red}\textbf{1}}, {\color{purple}\textbf{1}}, {\color{orange}\textbf{2}}, {\color{blue}\textbf{3}}, {\color{green}\textbf{3}} and {\color[rgb]{0.9, 0.8, 0.15}\textbf{4}} respectively.}
    \label{fig:introduction}
\end{figure}

We consider several variants of the following generic problem: given a 2-dimensional simplicial complex $K$ and a 1-dimensional subcomplex $B\subset K$, does $K$ contain a subcomplex homeomorphic to a surface with boundary $B$? Note that this includes finding surfaces without boundary, as we can set $B=\emptyset$.
\par

The Subsurface Recognition (SR) problem places the most restrictions on the manifold we are looking for. In this problem, we are asked to find a subcomplex of $K$ homeomorphic to a given (possibly disconnected) surface $X$. \cref{fig:introduction} shows an example of SR. 

\begin{problem}\textsc{The Subsurface Recognition} (SR) problem: \\
Input: A simplicial complex $K$, a subcomplex $B\subset K$, and a surface $X$.\\
Question: Does $K$ contain a subcomplex homeomorphic to $X$ with boundary $B$?
\end{problem}

Although there is no known FPT algorithm for SR, several variants of SR with looser requirements admit FPT algorithms. One special case of SR requires the surface $X$ to be connected. This variant is called the \CSR (c-SR) problem. The extra requirement of connectivity allows us to find an FPT algorithm.

\begin{problem}The \CSR (c-SR)  problem: \\
Input: A simplicial complex $K$, a subcomplex $B\subset K$, and a connected surface $X$.\\
Question: Does $K$ contain a subcomplex homeomorphic to $X$ with boundary $B$?
\end{problem}

We can also ask for a surface of a certain genus and orientability in $K$, which is a slightly weaker criterion than finding a surface up to homeomorphism. For a disconnected surface, we define its \textit{\textbf{total genus}} to be the sum of the genus of each of its connected components\footnote{If any connected component of a surface is non-orientable, we will add twice the genus of any orientable components.}. While a connected surface is characterized up to homeomorphism by its genus and orientability, this is not true for disconnected surfaces. As an example, consider a surface $X$ that is a genus 2 surface and a surface $Y$ that is the disjoint union of two tori. The two surfaces both have total genus 2, but they are not homeomorphic.

\begin{problem}The \SoGSR (SoG) problem: \\
Input: A simplicial complex $K$, a subcomplex $B\subset K$, and integers $g$ and $c$.\\
Question: Does $K$ contain a surface $X$ of total genus $g$ with $c$ connected components and with boundary $B$? 
\end{problem}

The \SP problem asks to find \textit{any} set of $c$ disjoint surfaces. In particular, no restriction is placed on the genus or orientability of these surfaces.

\begin{problem}The \SP (SP) problem: \\
Input: A simplicial complex $K$, a subcomplex $B$, and an integer $c$.\\
Question: Does $K$ contain a surface $X$ with $c$ connected components and boundary $B$?
\end{problem}

\subsection{Our Results}

\begin{table}[h!]
\centering
\begin{tabular}{|c||c|c|c|c|}
\hline
 Problem &  SR & c-SR & SoG & SP \\ \hline
 Upper & $2^{\OO(n)}$ &  
 \cellcolor{yellow!25} $\mathbf{2^{\OO(k\log k)}n^{\OO(1)}}$ &  
 \cellcolor{yellow!25} $\mathbf{2^{\OO(k\log k)}n^{\OO(1)}}$ &  
 \cellcolor{yellow!25} $\mathbf{2^{\OO(k\log k)}n^{\OO(1)}}$  \\ \hline
 Lower & 
 \cellcolor{yellow!25} $\mathbf{2^{\oo(k\log k)}n^{\OO(1)}}$ & 
 NP-Hard \cite{ivanov-hardness} & 
 \cellcolor{yellow!25} $\mathbf{2^{\oo(k\log k)}n^{\OO(1)}}$ & 
 \cellcolor{yellow!25} $\mathbf{2^{\oo(k\log k)}n^{\OO(1)}}$ \\ \hline
\end{tabular}
\caption{Upper and ETH lower bounds for times to solve the different problems considered in this manuscript. Here \(n\) is the number of simplices and \(k\) is the treewidth of the Hasse diagram. The results of this paper are highlighted.}
\label{tab:bounds}
\end{table}

We consider the parameterized complexity of the above problems with respect to the treewidth \(k\) of the Hasse diagram. \cref{tab:bounds} summarizes the known upper and lower bounds. We give FPT algorithms for c-SR, SoG, and SP, and ETH-based lower bounds for SR, SP, and SoG. In fact, we show that these lower bounds are true even when $k$ is the pathwidth of the Hasse diagram. 
The algorithms for SoG and SP are ETH-tight. 

\subsection{Related Work}

\paragraph*{Tree Decompositions and Simplicial Complexes} Tree decompositions have seen much success as an algorithmic tool on graphs. Often, graphs having tree decompositions of bounded-width admit polynomial-time solutions to otherwise hard problems. A highlight of the algorithmic application of tree decompositions is Courcelle's Theorem \cite{courcelle_msol}, which states that any problem that can be stated in monadic second order logic can be solved in linear time on graphs with bounded treewidth.
We recommend \cite[Chapter 7]{Cygan_2015} for an introduction to the algorithmic use of tree decompositions.
\par
While tree decompositions have long been successful for algorithms on graphs, they have only recently seen attention for algorithms on simplicial complexes. Existing algorithms use tree decompositions of a variety of graphs associated with a simplicial complex. The most commonly used graph is the dual graph of combinatorial $d$-manifolds \cite{bagchi_tightness, burton_courcelle, Burton_morse, burton_taut}. Other graphs that have been used are level $d$ of the Hasse diagram \cite{Burton_morse, blaser_hl, vaagset2021mbc}, the adjacency graph of the $d$-simplices \cite{blaser_hl}, and the 1-skeleton \cite{bagchi_tightness}. Our algorithm uses a tree decomposition of the entire Hasse diagram. As far as we know, we are the first to consider tree decompositions of the full Hasse diagram. The condition on vertex links that makes a simplicial complex a surface is dependent on the incidence of vertices and triangles (see Section \ref{sec:background_surfaces}), so considering only one level of the Hasse diagram would likely not be sufficient for our problem.

\paragraph*{Normal Surface Theory} 

Normal surface theory is the study of which surfaces exist as submanifolds of a given 3-manifold. Many algorithms on 3-manifolds, like those for unknot recognition~\cite{haken_normal} and 3-sphere recognition~\cite{rubinstein_sphere, thompson_sphere}, use normal surface theory. While normal surface theory appears to be similar to our problems, the distinction is that the surfaces in normal surface theory are not subcomplexes of the 3-manifold and can instead intersect 3-simplices in the manifold. Accordingly, the techniques in normal surface theory are quite different from the algorithms we present in this paper. 

\section{Background}

\subsection{Simplicial Complexes and Directed Graphs}

A \textit{\textbf{simplicial complex}} is a set $K$ such that (1) each element $\sigma\in K$ is a finite set and (2) for each $\sigma\in K$, if $\tau\subset\sigma$, then $\tau\in K$. An element $\sigma\in K$ is a \textit{\textbf{simplex}}. A simplex $\sigma$ is a \textit{\textbf{face}} of a simplex $\tau$ if $\sigma\subset\tau$. Likewise, $\tau$ is a \textit{\textbf{coface}} of $\sigma$. The simplices $\sigma$ and $\tau$ are \textit{\textbf{incident}}. Two simplices $\sigma_1$ and $\sigma_2$ are \textit{\textbf{adjacent}} if they are both the face or coface of a simplex $\tau$.
\par
A simplex $\sigma$ with $|\sigma|=d+1$ is a \textit{\textbf{d-simplex}}. The set of all $d$-simplices in $K$ is denoted $K_d$. The \textit{\textbf{dimension}} of a simplicial complex is the largest integer $d$ such that $K$ contains a $d$-simplex. A $d$-dimensional simplicial complex $K$ is \textit{\textbf{pure}} if each simplex in $K$ is a face of $d$-simplex. We call a 0-simplex a \textit{\textbf{vertex}}, a 1-simplex an \textit{\textbf{edge}}, and a 2-simplex a \textit{\textbf{triangle}}. 
\par 
The \textit{\textbf{Hasse diagram}} of $K$ is a graph $H$ with vertex set $K$ and edges between each $d$-simplex $\sigma\in K$ and each $(d-1)$-dimensional face of $\sigma$ for all $d>0$. 
\par 
Let $\Sigma\subset K$. The \textit{\textbf{closure}} of $\Sigma$ is $\cl\Sigma:=\{\tau\subset\sigma\mid\sigma\in\Sigma\}$. Note that the closure of $\Sigma$ is a simplicial complex, even if $\Sigma$ is not. Note also that the closure $\cl\Sigma$ is defined only by the set $\Sigma$ and not the complex $K$. The \textit{\textbf{star}} of $\Sigma$ is $\str{K}{\Sigma}:=\{\sigma\in K\mid \exists\,\tau\in \Sigma\text{ such that }\tau\subset\sigma\}$. The \textit{\textbf{link}} of a simplex $\sigma$ is $\lk_{K}{\sigma}=\cl\str{K}{\sigma}-\str{K}{\cl\sigma}$. Alternatively, the link $\lk_{K}{\sigma}$ is all simplices in $\cl\str{K}{\sigma}$ that do not intersect $\sigma$. Note that for any simplex $\tau\in\lk_{K}{\sigma}$ that $\sigma$ and $\tau$ are incident to a common coface in $\str{K}{\sigma}$.
\par
A \textit{\textbf{simple path}} is a 1-dimensional simplicial complex $P=\{\{v_1\},\{v_1,v_2\},\{v_2\},\ldots,\{v_l\}\}$ such that the vertices $\{v_i\}$ are distinct. The vertices $\{v_1\},\{v_l\}$ are the \textit{\textbf{endpoints}} of $P$. We will denote a simple cycle as a tuple $P=(v_1,\ldots,v_l)$ as the edges are implied by the vertices. A \textit{\textbf{simple cycle}} is a simple path, with the exception that the endpoints $v_1=v_l$. We denote a simple cycle with an overline, e.g. $(\overbar{v_1,\ldots,v_l})$.
\par
A directed graph $D$ consists of a set of vertices and a set of directed edges, i.e. ordered pairs of vertices $(u,v) := uv$ so that $uv \neq vu$. A \textit{\textbf{directed simple cycle}} $\cyc$ in $D$ (not to be confused with a simple cycle) is a sequence of directed edges $(v_1v_2, v_2v_3,\dots, v_{l}v_1)$ where all the vertices $v_i$ are distinct. We say that $\cyc$ has the vertex set $\{v_1,\dots, v_{l}\}$. Two cycles, $\cycA$ and $\cycB$, are said to be \textit{\textbf{vertex disjoint}} if their vertex sets are disjoint. A family of cycles is said to be vertex disjoint if they are pairwise vertex disjoint.
\begin{figure}[!ht]
    \centering
    \begin{subfigure}{0.35\textwidth}
        \centering
        \includegraphics[height=1in]{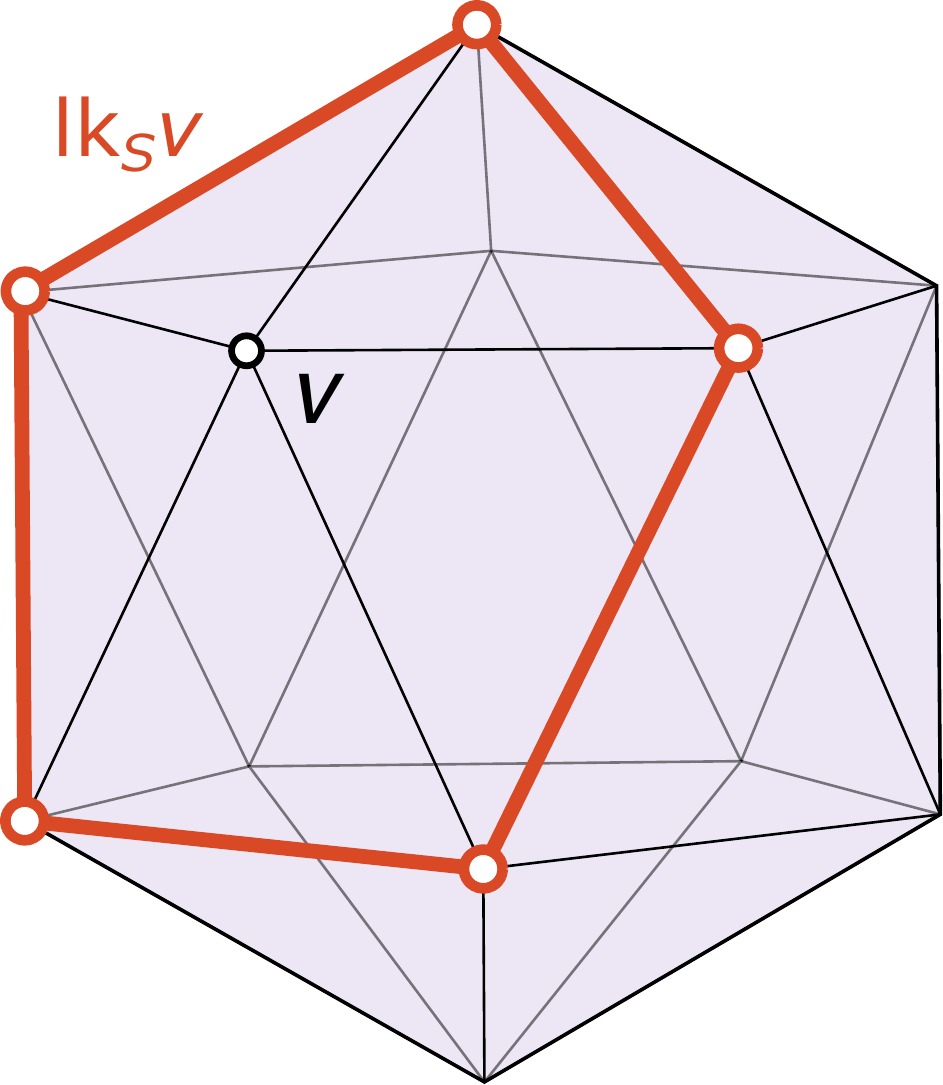}
    \end{subfigure}
    \begin{subfigure}{0.35\textwidth}
        \centering
        \includegraphics[height=1in]{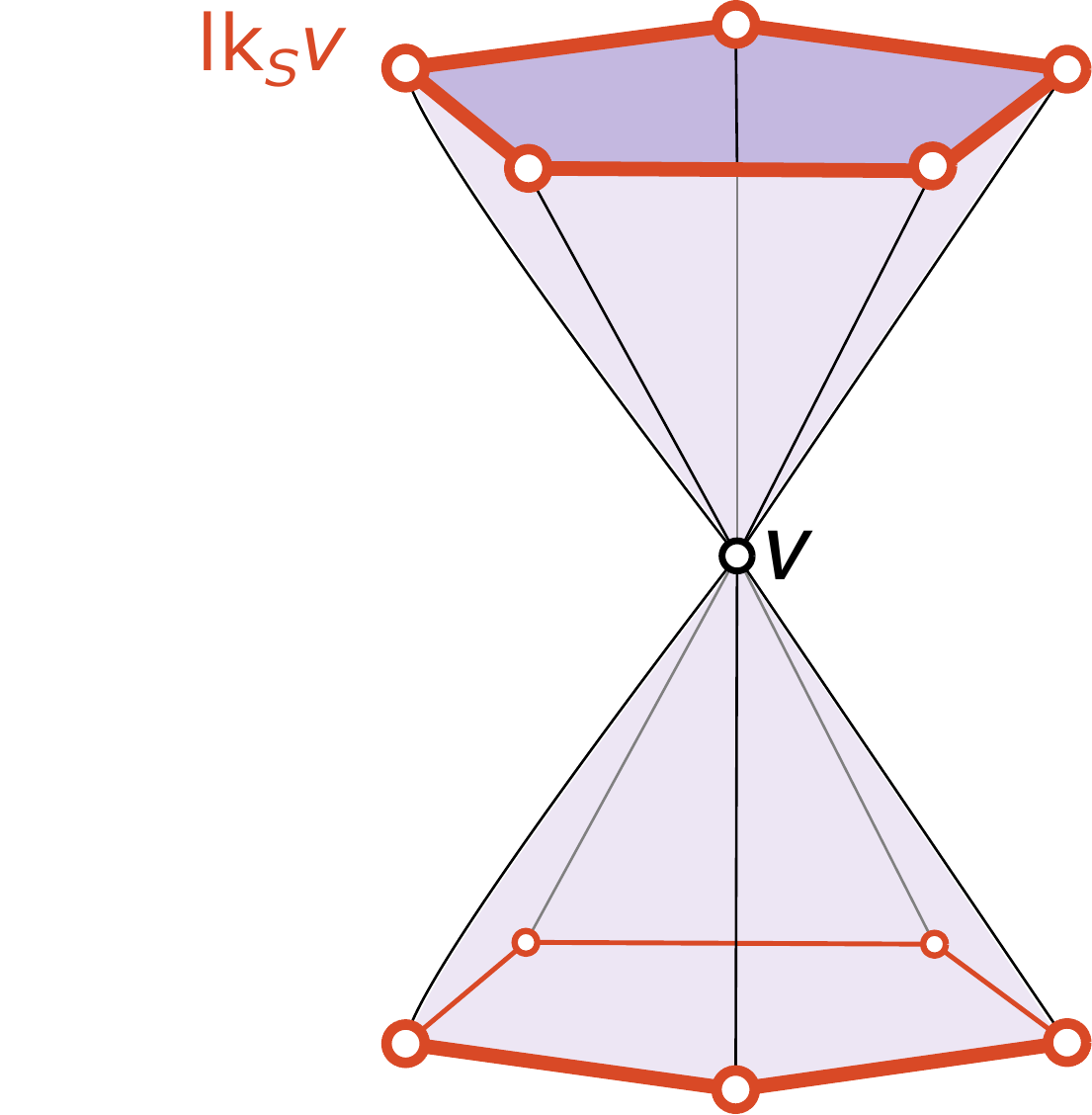}
    \end{subfigure}
    \caption{Left: A combinatorial surface. The vertex $v$ is an interior vertex. Right: A vertex $v$ with link that is neither a simple path or cycle. We conclude that $S$ is not a combinatorial surface. The point $v$ has no neighborhood homeomorphic to the plane or half-plane, so $S$ is not ``locally 2-dimensional'' at $v$.}
    \label{fig:links}
\end{figure}

\subsection{Surfaces}
\label{sec:background_surfaces}

Informally, a \textit{\textbf{surface with boundary}} is a compact topological space where each point has a neighborhood homeomorphic to the plane or the half plane, and the \textit{\textbf{boundary}} of the surface is all points with a neighborhood homeomorphic to the half plane. Intuitively, a surface is ``locally 2-dimensional''.
\par
Any connected surface with boundary can be constructed by adding handles, crosscaps, and boundary components to a sphere. A \textit{\textbf{handle}} is constructed by removing two disjoint disks from a surface and identifying the boundaries of the removed disks. A \textit{\textbf{crosscap}} is constructed by taking the disjoint union of the surface and the real projective plane, removing a disk from each, and identifying the boundaries of the removed disks. A \textit{\textbf{boundary component}} is constructed by removing a disk from a surface. A surface is \textit{\textbf{non-orientable}} if it has a crosscap and \textit{\textbf{orientable}} otherwise. The \textbf{\textit{genus}} of an orientable surface is the number of handles on the surface, and the genus of a non-orientable surface is the number of crosscaps plus twice the number of handles.
\par 
In this paper, we are only concerned with surfaces that are also simplicial complexes, which we call combinatorial surfaces. A \textit{\textbf{combinatorial surface with boundary}} is a pure 2-dimensional simplicial complex $S$ such that the link of each vertex is a simple path or a simple cycle. The condition on the link of the vertices is the combinatorial way of saying that a combinatorial surface is ``locally 2-dimensional''.  A vertex $v\in S$ such that $\lk_{S}{v}$ is a simple path is a \textit{\textbf{boundary vertex}}. A vertex $v\in S$ such that $\lk_{S}{v}$ is a simple cycle is an \textit{\textbf{interior vertex}}. Figure \ref{fig:links} shows examples of an interior vertex and a vertex that is neither an interior or boundary vertex. It follows from the condition on the links of the vertices that each edge $e\in S$ has link $\lk_{S}{e}$ that is either one or two vertices. An edge $e\in S$ such that $\lk_{S}{e}$ is a single vertex is a \textit{\textbf{boundary edge}}. An edge $e\in S$ such that $\lk_{S}{e}$ is two vertices is an \textit{\textbf{interior edge}}. A triangle $t\in S$ has empty link $\lk_{S}{t}=\emptyset$ as $S$ is a 2-dimensional simplicial complex. We denote the set of boundary vertices and boundary edges $\boundary S$. The boundary $\boundary S$ is a collection of simple cycles.
\par

\subsection{Tree Decompositions}

Let $G=(V,E)$ be a graph. A \textit{\textbf{tree decomposition}} of $G$ is a tuple $(T,X)$, where $T=(I,F)$ is a tree with nodes $I$ and edges $F$, and $X=\{X_t\subset V\,|\,t\in I\}$ such that (1) $\cup_{t\in I}X_t=V$, (2) for any $\{v_1,v_2\}\in E$, $\{v_1,v_2\}\subset X_t$ for some $t\in I$, and (3) for any $v\in V$, the subtree of $T$ induced by the nodes $\{t\in I\mid v\in X_t\}$ is connected. A set $X_t$ is the \textit{\textbf{bag}} of $T$. The \textit{\textbf{width}} of $(T,X)$ is ${\max}_{t\in I}|X_t|-1$. The \textit{\textbf{treewidth}} of a graph $G$ is the minimum width of any tree decomposition of $G$. Computing the treewidth of a graph is NP-hard \cite{arnborg_treewidth_hardness}, but there are algorithms to compute tree decompositions that are within a constant factor of the treewidth, e.g. \cite{bodlaender2013ock}. 
\par
Tree decompositions are used to perform dynamic programs on graphs, and a certain type of tree decomposition, called a nice tree decomposition, makes defining dynamic programs easier. A \textit{\textbf{nice tree decomposition}} is a tree decomposition with a specified root $r\in I$ such that (1) $X_r=\emptyset$, (2) $X_l=\emptyset$ for all leaves $l\in I$, and  (3) all non-leaf nodes are either an introduce node, a forget node, or a join node, which are defined as follows. An \textit{\textbf{introduce node}} is a node $t\in I$ with exactly one child $t'$, and for some $w\in V$, $w\notin X_{t'}$ and $X_t=X_{t'}\cup\{w\}$. We say $t$ \textit{\textbf{introduces}} $w$. A \textit{\textbf{forget node}} is a node $t\in I$ with exactly one child $t'$, and for some $w\in V$, $w\notin X_t$ and $X_t\cup\{w\}=X_{t'}$. We say $t$ \textit{\textbf{forgets}} $w$. A \textit{\textbf{join node}} is a node $t\in I$ with exactly two children $t'$ and $t''$ where $X_{t}=X_{t'}=X_{t''}$. The following lemma proves that we can convert any tree decomposition to a nice tree decomposition without increasing width.
\begin{lemma}[Lemma 7.4 of \cite{Cygan_2015}] 
\label{lem:kn-nodes}
Given a tree decomposition $(T=(I,F),X)$ of width $k$ of a graph $G=(V,E)$, a nice tree decomposition of width $k$ with $\OO(kn)$ nodes can be computed in $\OO(k^2\max\{|V|,|I|\})$ time.
\end{lemma}

\begin{figure}[!ht]
    \centering
    \begin{subfigure}{0.3\textwidth}
        \centering
        \includegraphics[height=1in]{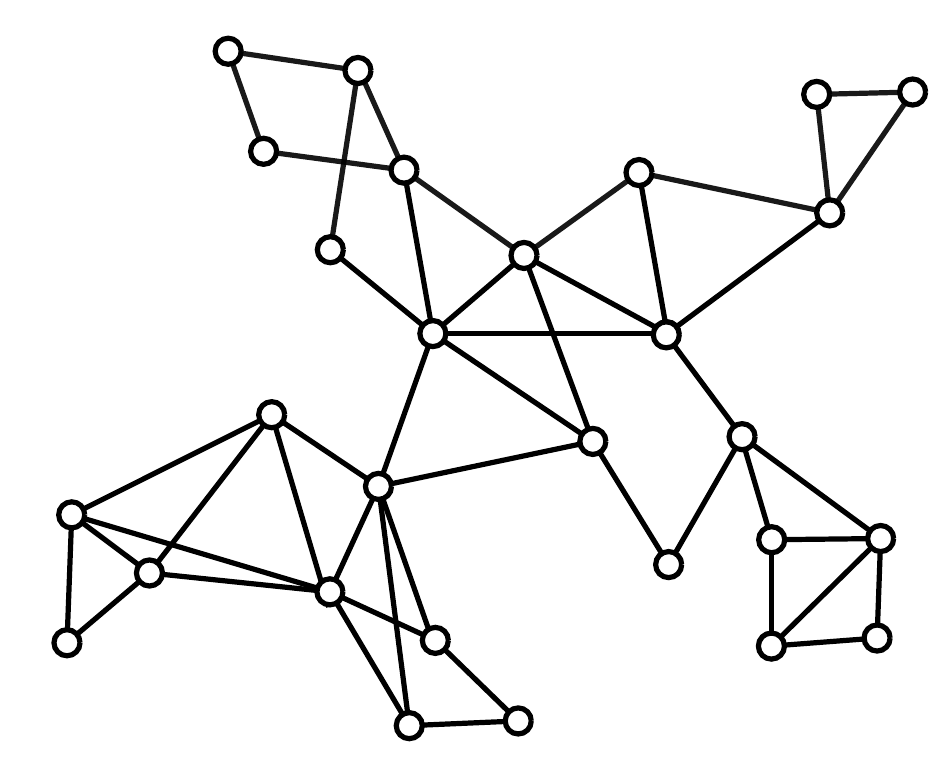}
    \end{subfigure}
    \begin{subfigure}{0.3\textwidth}
        \centering
        \includegraphics[height=1in]{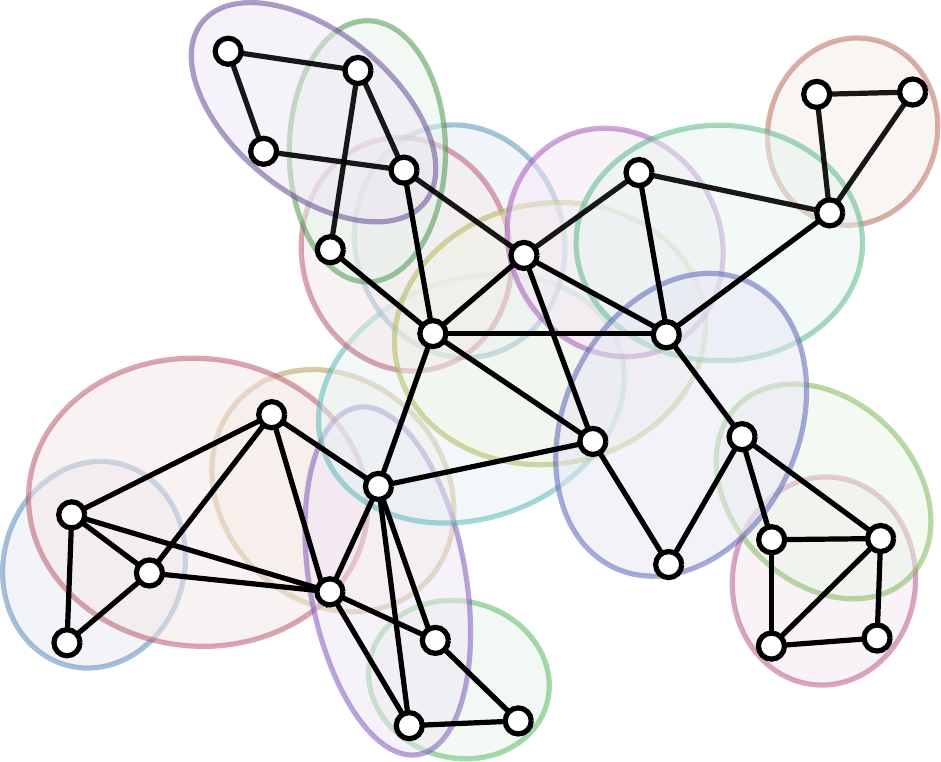}
    \end{subfigure}
    \begin{subfigure}{0.3\textwidth}
        \centering
        \includegraphics[height=1in]{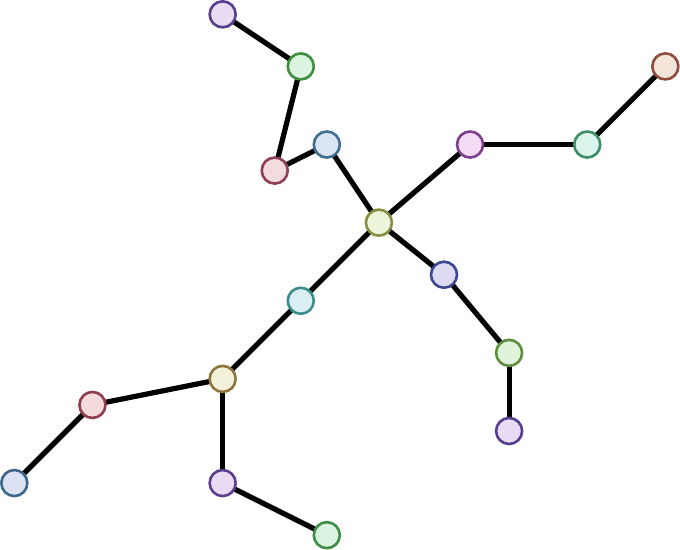}
    \end{subfigure}
    \caption{Left: A graph. Right and Center: A (not nice) tree decomposition of the graph of width 3. Each node of the tree corresponds to a subset of the vertices of the graph.}
    \label{fig:td}
\end{figure} 
 
A \textit{\textbf{path decomposition}} is a special kind of tree decomposition $(T,X)$ where $T$ is a path. A \textit{\textbf{nice path decomposition}} is a tree decomposition without join nodes, i.e. where every node is either an introduce node or a forget node. The \textit{\textbf{pathwidth}} of a graph $G$ is the smallest width of any path decomposition of $G$. As any path decomposition is also a tree decomposition, the treewidth of $G$ is at most the pathwidth of $G$.

\subsection{The Exponential Time Hypothesis}

When a new algorithm is discovered it is natural to ask if it is possible to improve it. To prove that the algorithm was sub-optimal it is enough to find a new and better algorithm. On the other hand, if the algorithm is actually the best possible, then the situation becomes more complicated. Although there are optimality results for a few problems in P,\footnote{One such example is sorting, which we know can at best be done in $\Omega(n\log n)$ time.} none are known for algorithms solving NP-complete problems. Such a result would imply $\text{P}\neq\text{NP}$, which remains famously unproven. 

This theoretical barrier does not make the question of optimality less relevant. No one wants to spend years searching for improvements to an algorithm that cannot be improved! 
For instance, the algorithms in this paper need $2^{\OO(k\log k)}n^{\OO(1)}$ time, which may prompt the question ``Why were you unable to deliver a $2^{\OO(k)}n$ time solution?''.   


A pragmatic and popular response to these kinds of questions is to prove that you have optimality under the Exponential Time Hypothesis (ETH). The ETH is a conjecture stating that there is no sub-exponential algorithm for 3-SAT. More precisely, let $n$ be the number of variables in a given instance of 3-SAT.

\begin{hypothesis}[ETH] 
    3-SAT cannot be solved in time $2^{\oo(n)}$.
\end{hypothesis}

Similar to NP-hardness, an ETH-lower bound is a way of connecting the hardness of a new and often poorly understood problem to problems we already have a good understanding of. The idea is to show that an improvement on the runtime of the currently best algorithm for a new problem would disprove the ETH. Although the ETH remains unproven, 
the continued absence of any algorithm for 3-SAT fast enough to disprove the ETH is itself strong empirical evidence in support of the hypothesis.

\section{Algorithms}
\label{sec: algorithms}

We first present a high-level overview of our algorithm in Section \ref{sec:algorithm_overview}. The remainder of the section then explains the ideas presented in the overview in greater detail.

\subsection{Overview of the Algorithms}
\label{sec:algorithm_overview}

Our algorithms are all dynamic programs on a tree decomposition $(T,X)$ of the Hasse diagram of a simplicial complex $K$. For each node $t\in T$, starting at the leaves of $T$ and moving towards the root, we compute a set of candidate solutions to our problem, where a candidate solution is a subcomplex of $K$ that might be a subcomplex of a solution to our problem. We recursively use candidate solutions at the children of $t$ to build the candidate solutions at $t$. At the end of the algorithm, candidate solutions at the root of $t$ will be solutions to our problem. In this section, we explore how a candidate solution to our problem is defined, and how we can effectively store representations of these candidate solutions so that our final algorithm is FPT.
\par 
Certain nice tree decompositions\footnote{Certain here means \textit{closed}, which is a type of tree decomposition of the Hasse diagram we define in Section \ref{sec:closed_td}. In particular, the set $K_t$ as defined above is a simplicial complex in a closed tree decomposition, which would not true for general tree decompositions of the Hasse diagram.}
$(T,X)$ of the Hasse diagram of a simplicial complex $K$ define a recursively-nested set of subcomplexes of $K$. Recall that each bag of the tree decomposition is a set of simplices of $K$. For each node $t\in T$, the subcomplex $K_t\subset K$ is the union of the bags of each descendant of $t$ minus the triangles in the bag of $t$. These subcomplexes have the property that if $t'$ is a child of $t$, then $K_{t'}\subset K_{t}$.
\par 
We use this set of subcomplexes to recursively build solutions to our problems. Our algorithm computes a set of \textit{\textbf{candidate solutions}} at each node $t$. The exact definition of candidate solution is given in Section \ref{sec:candidate_solutions}, but intuitively, a candidate solution at a node $t$ is a subcomplex of $K_t$ that could be a subcomplex of a combinatorial surface in $K$. In particular, the link of each vertex in a candidate solution must be a subset of a simple path or simple cycle. Our definition of candidate solution works recursively: if $\Sigma$ is a candidate solution at $t$, then for each child $t'$ of $t$, the complex $\Sigma\cap K_{t'}$ is a candidate solution at $t'$. Our algorithm uses this fact to find candidate solutions at $t$. Specifically, our algorithm attempts to build candidate solutions at $t$ by growing candidate solutions at $t'$. 
\par 
The main challenge with this approach is storing candidate solutions. There can be an exponential number of candidate solutions at a given node $t$, so we cannot simply store all candidate solutions. Generally, dynamic programs on tree decompositions work by storing some local representation of candidate solutions at $t$, where a local representation is a description of a candidate solution only in terms of vertices and edges in the bag $X_t$. Two candidate solutions with the same local representation are typically interchangeable in the sense that one candidate solution can be extended to a complete solution if and only if the other can too. The number of these local representations at $t$ is typically a function of the size of $X_t$, which allows for FPT algorithms parameterized by the treewidth.
\par
The local representation of candidate solutions for our problems should have several properties. First, they should represent a candidate solution using only simplices in $X_t$. Second, they should retain enough information that we can verify that a subcomplex is a candidate solution, i.e. it could be extended to a surface in $K$. In particular, we should be able to deduce information about the links of simplices in $X_t$ from the local representation. The first and second properties are at odds, as even if a simplex $\sigma$ is contained in $X_t$, the link of $\sigma$ need not be contained in $X_t$. Finally, we should be able to deduce the homeomorphism class of a candidate solution from the local representation. Again, this property is at odds with the first property, as topological properties like the genus and orientability of a surface are global, not local, properties of a surface. One of our contributions is introducing a data structure to store local representations of candidate solution with each of these properties called the \textit{\textbf{annotated cell complex}}. 
\par 
A (non-annotated) \textit{\textbf{cell complex}} is an algebraic representation of a surface that was originally introduced by Ahlfors and Sario \cite{ahlfors_sario} to prove the Classification Theorem of Compact Surfaces. Intuitively, a cell complex is a collection of disks, called \textit{\textbf{faces}}, joined by shared edges in their boundaries. The faces in a cell complex differ from triangles in a simplicial complex as the faces in a cell complex can have more than three edges in their boundary. A definition of cell complex and a discussion of their properties can be found in Section \ref{sec:cell_complexes}.
\par 
The advantage of using cell complexes rather than simplicial complexes to store surfaces is that there is a simple equivalence relation that partitions cell complexes into homeomorphism classes. This is of obvious benefit as the surface $S$ we are looking for may be specified by its homeomorphism class, but there is a secondary benefit. We define a set of \textit{\textbf{equivalence-preserving moves}}, operations on cell complexes that preserve their homeomorphism class. We use these moves to compress the local representation of each candidate solution we keep during our algorithm. The most important benefit that these moves provide is the ability to merge two faces that share an edge.
\par 
To see why merging faces is helpful, suppose that we have a candidate solution $\Sigma$ at a node $t$ that is represented as a cell complex. We would like to store a local representation of $\Sigma$ using only edges in $X_t$. There would then be a bounded number of local representations of candidate solutions at a node $t$, as there are a bounded number of edges in $X_t$. To this end, each time we forget an edge $e$, we would like to merge the two faces incident to $e$ into a single face. See Figure~\ref{fig:remove_e}, left panel.
\begin{figure}
    \centering
    \begin{subfigure}{0.50\textwidth}
        \centering
        \includegraphics[height=1in]{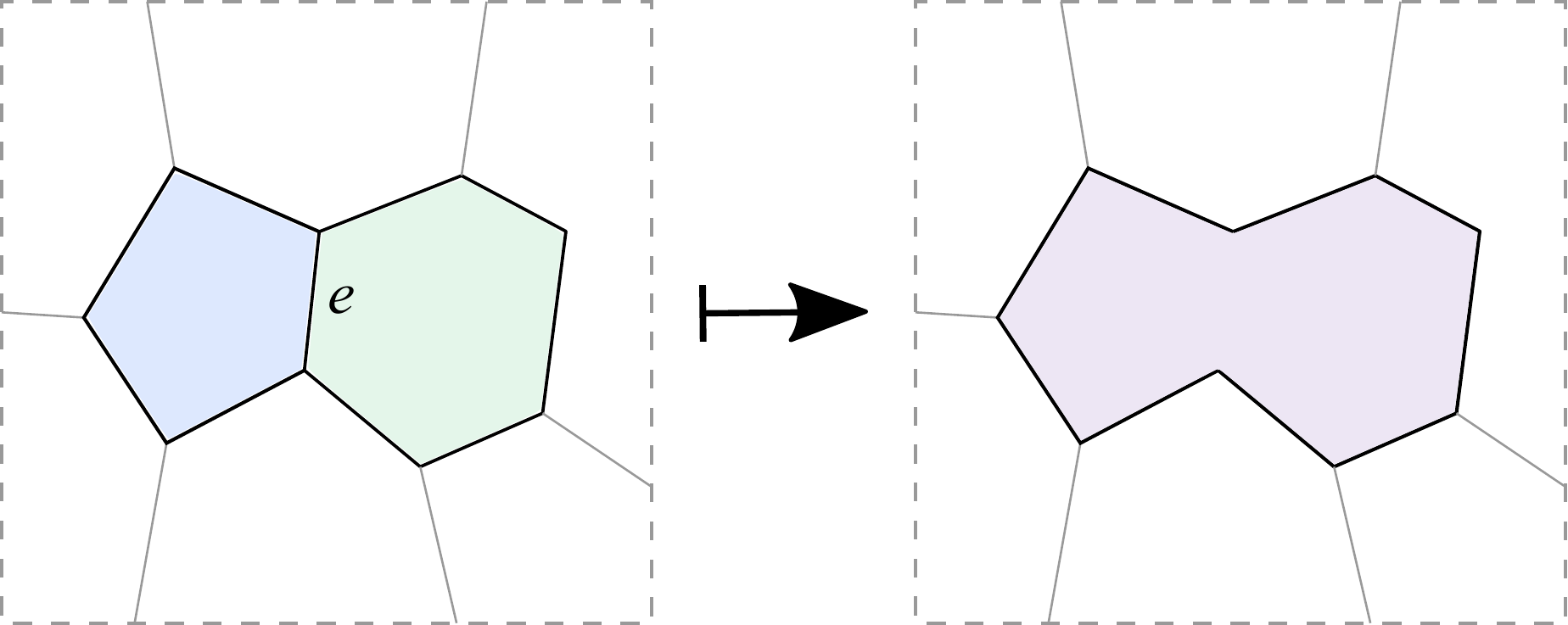}
    \end{subfigure}
    \begin{subfigure}{0.40\textwidth}
        \centering
        \includegraphics[height=1in]{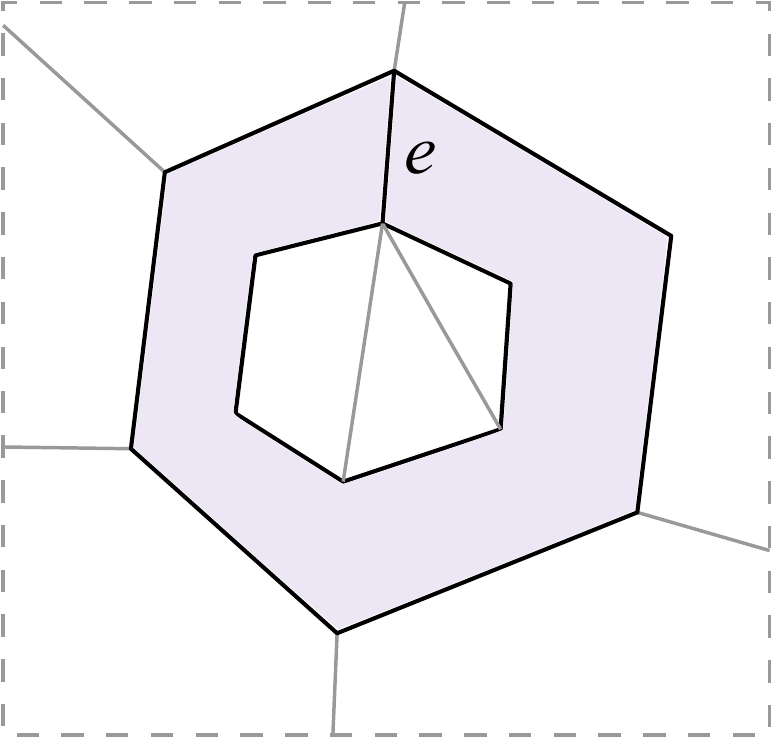}
    \end{subfigure}
    \caption{Left: The edge $e$ is removed by merging the two incident faces. Right: The edge $e$ appears twice on the boundary of the same face, so $e$ cannot be removed by merging incident faces as this would make the interior of the face an annulus. We use annotated cell complexes to remove $e$.}
    \label{fig:remove_e}
\end{figure}
\par 
The idea of merging faces when we forget $e$ works unless $e$ is incident to the same face twice; the right panel of Figure \ref{fig:remove_e} gives an example. After merging some faces, it is possible that a face may have two edges on its boundary identified. If two edges on the boundary of the same face are identified, then we can no longer remove these edges by merging their incident faces, as then the interior of this face would no longer be a disk.
\par 
We therefore modify the definition of cell complex to allow for a more general type of face. Our first change is to allow a face to be a disk with multiple boundary components like in Figure \ref{fig:remove_e}, but we need to go a step further. Topological features like handles, crosscaps, and boundaries in cell complexes are the result of a single face having edges on its boundary identified in certain ways; thus, we need a way of removing the edges that constitute these topological features. An \textit{\textbf{annotated cell complex}} annotates each face with the number of topological features like handles, crosscaps, and boundaries on this face, rather than storing these features explicitly with edges. In effect, an annotated cell complex is a representation of a surface where the interior of a face is allowed to be any compact connected surface.

\subsection{Cell Complexes}
\label{sec:cell_complexes}

 A cell complex is an algebraic representation of a surface. Cell complexes are more flexible than combinatorial surface as they provide a set of algebraic rules that will allow us to store a compressed representation of candidate solutions in our complex. Intuitively, these moves allow us to merge two faces together into a single face and perform other transformations to simplify our candidate solutions.   
\par
In this section, we introduce cell complexes and their basic properties. In Section \ref{sec:cell_definition}, we define cell complexes.  In Section \ref{sec:cell_surf}, we explore the relationship between cell complexes and combinatorial surfaces. In Section \ref{sec:surface_classification}, we see how cell complexes provide a concise way of classifying connected, compact surfaces up to homeomorphism. In Section \ref{sec:moves}, we give a set of equivalence-preserving rules for cell complexes and define a generalization of the cell complex, the \textit{\textbf{annotated cell complex}}, which allows us to further compress cell complexes as is needed for our algorithm.

\subsubsection{Definition}
\label{sec:cell_definition}

For a set $X$, let $X^{-1}=\{x^{-1}\mid x\in X\}$. We will let $(x^{-1})^{-1}=x$. For the time being, it is fine to treat $^{-1}$ as meaningless notation.
\par 
A \textit{\textbf{cell complex}} is a tuple $C=(F,E,B)$ where $F$ and $E$ are finite sets and $B:F\sqcup F^{-1}\to (E\sqcup E^{-1})^{*}$ is a map that assigns each element $A\in F\sqcup F^{-1}$ a cyclically ordered sequence $B(A)=(\overbar{a_1\ldots a_m})$ where each $a_i\in E\sqcup E^{-1}$, such that (1) $B(A^{-1})=(\overbar{a_m^{-1}\ldots a_1^{-1}})$ and (2) each element $e\in E\sqcup E^{-1}$ appears either once or twice in some sequence $B(A)$.\footnotemark 
\footnotetext{This definition of cell complex was introduced by Ahlfors and Sario \cite{ahlfors_sario} to give a proof of the Classification Theorem for Compact Surfaces, although many proofs of this theorem use an algebraic description of surfaces similar to cell complexes. We recommend Chapter 6 of the book by Gallier and Xu \cite{gallier_xu} for a modern treatment of cell complexes.}
\par 
Elements of $F$ are called \textbf{\textit{faces}} and elements of $F\sqcup F^{-1}$ are called \textbf{\textit{oriented faces}}. Elements of $E$ are called \textbf{\textit{edges}} and elements of $E\sqcup E^{-1}$ are called \textbf{\textit{oriented edges}}. The sequence $B(A)$ is the \textbf{\textit{boundary}} of $A$.
\par
It is informative to visualize cell complexes as a collection of disks $F$ identified along shared edges in their boundaries. A face $A\in F$ is a disk with edges on its boundary $B(A)$. The element $A^{-1}$ is the same disk as $A$ but with opposite orientation, and the boundary $B(A^{-1})$ is the boundary of $A$ traversed in the opposite direction as $B(A)$. For example, $B(A)$ might be the order of the edges on the boundary when traversing the boundary clockwise and $B(A^{-1})$ would be the order of the edges when traversing the boundary counterclockwise. 
\begin{figure}[H]
    \centering
    \includegraphics[height=0.75in]{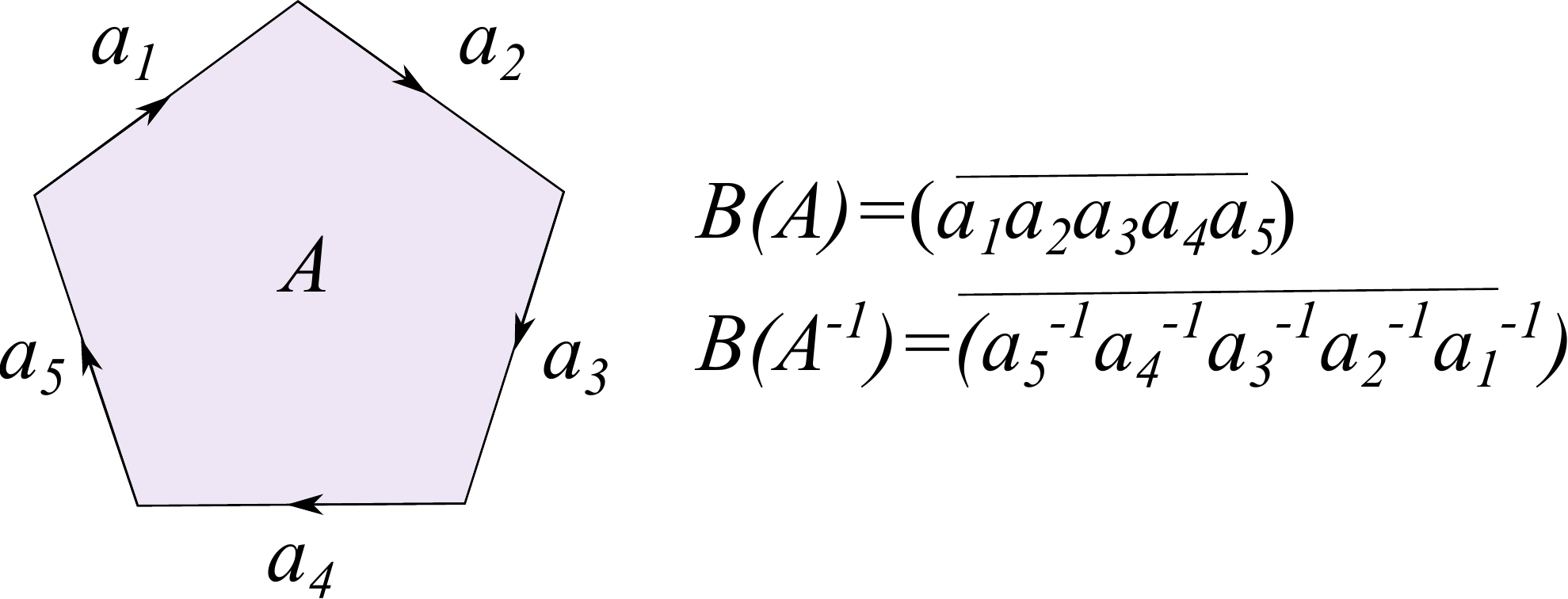}
    \caption{The oriented faces $A$ and $A^{-1}$ represent the two ways (clockwise and counterclockwise) of traversing the boundary of $A$.}
    \label{fig:a_a_inverse}
\end{figure}
 The oriented edges $a$ and $a^{-1}$ are the two oppositely-directed edges defined by the edge $a\in E$. If an oriented edge $a$ appears twice in the boundaries of faces, then $a$ on one face is glued to $a^{-1}$ on the other face. If we repeat this gluing for all oriented edges, then we can visualize our collection of disks $F$ as part of a single (possibly disconnected) surface. We store both the oriented faces $A$ and $A^{-1}$ to ensure that both the oriented edges $a$ and $a^{-1}$ appear in the boundary of a face.
\begin{figure}[!ht]
    \centering
        \begin{subfigure}{0.35\textwidth}
            \centering
            \includegraphics[height=1in]{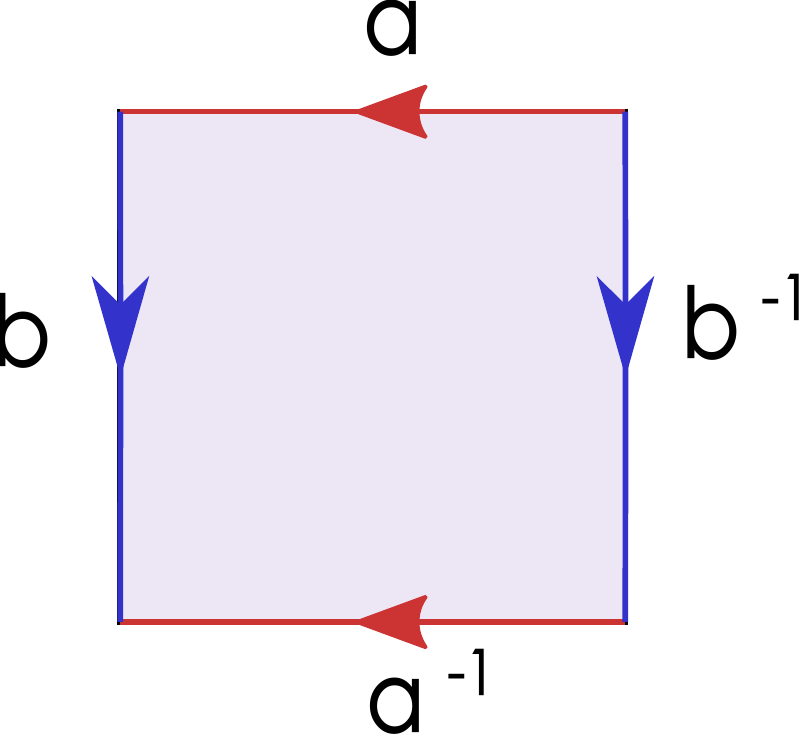}
        \end{subfigure}
        \begin{subfigure}{0.35\textwidth}
            \centering
            \includegraphics[height=1in]{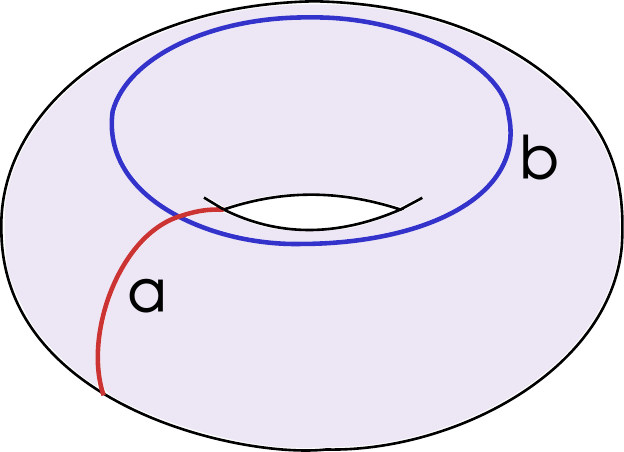} 
        \end{subfigure}
    \caption{If an edge $a$ appears on the boundary of two faces (or twice on the boundary of a face), we visualize this as $a$ on one face being glued to $a^{-1}$ on the other face. Here we see how these edge identifications can be used to define the torus.}
    \label{fig:cell_edge}
\end{figure}

A cell complex is \textit{\textbf{connected}} if there is no partition $F_1\sqcup F_2=F$ and $E_1\sqcup E_2= E$ such that $(E_1,F_1,\restr{B}{F_1\sqcup F_1^{-1}})$ and $(E_2,F_2,\restr{B}{F_2\sqcup F_2^{-1}})$ are also cell complexes. If a cell complex $C$ is disconnected, the \textit{\textbf{connected components}} of $C$ are  the cell complexes $(E_i,F_i,\restr{B}{F_i\sqcup F_i^{-1}})$ for $1\leq i\leq k$ where (1) $E$ and $F$ are partitioned by $E = E_1\sqcup\cdots\sqcup E_k$ and $F = F_1\sqcup\cdots\sqcup F_k$  and (2)  each cell complex $(E_i,F_i,\restr{B}{F_i\sqcup F_i^{-1}})$ is connected.
\par
Let $a\in E\sqcup E^{-1}$. A \textit{\textbf{successor}} of $a$ is an edge $b$ such that $ab$ is a substring of a boundary $B(A)$ for some $A\in F\sqcup F^{-1}$. As each oriented edge appears in at most two boundaries, then $a$ has at most two successors. If an oriented edge appears in two boundaries, we say it has a \textit{\textbf{pair of successors}}; otherwise, if an oriented edge appears in a single boundary, it has a \textit{\textbf{single successor}}. A \textit{\textbf{sequence of successors}} is a sequence of edges $(a_1\ldots a_k)$ such that $a_{i-1}^{-1}$ and $a_{i+1}^{-1}$ are pairs of successors of $a_{i}$ for $2\leq i\leq k-1$, $a_2^{-1}$ is the single successor of $a_1$, and $a_{k-1}^{-1}$ is the single successor of $a_{k}$. If there is a face $(\overbar{a})$ in a cell complex, then this face and its inverse $(\overbar{a^{-1}})$ define the sequence of successors $(aa^{-1})$ as faces are cyclic sequences of edges. A \textit{\textbf{cyclic sequence of successors}} is a cyclically ordered sequence of edges $(\overbar{a_1\ldots a_k})$ such that $a_{i-1}^{-1}$ and $a_{i+1}^{-1}$ are a pair of successors for $a_{i}$, with indices taken modulo $k$. We distinguish cyclic sequences of successors from sequences of successors with an overline. If $aa^{-1}$ is the substring of a boundary, then $(\overbar{a})$ is a cyclic sequence of successors. Sequence of successors describe sets of edges that all enter a common vertex and are analogous to rotation systems in surface graphs. Sequence of successors define vertices in cell complexes when the edges do not have their own notion of vertex; see \cite{gallier_xu}.

\begin{figure}[H]
    \centering
    \includegraphics[height=0.9in]{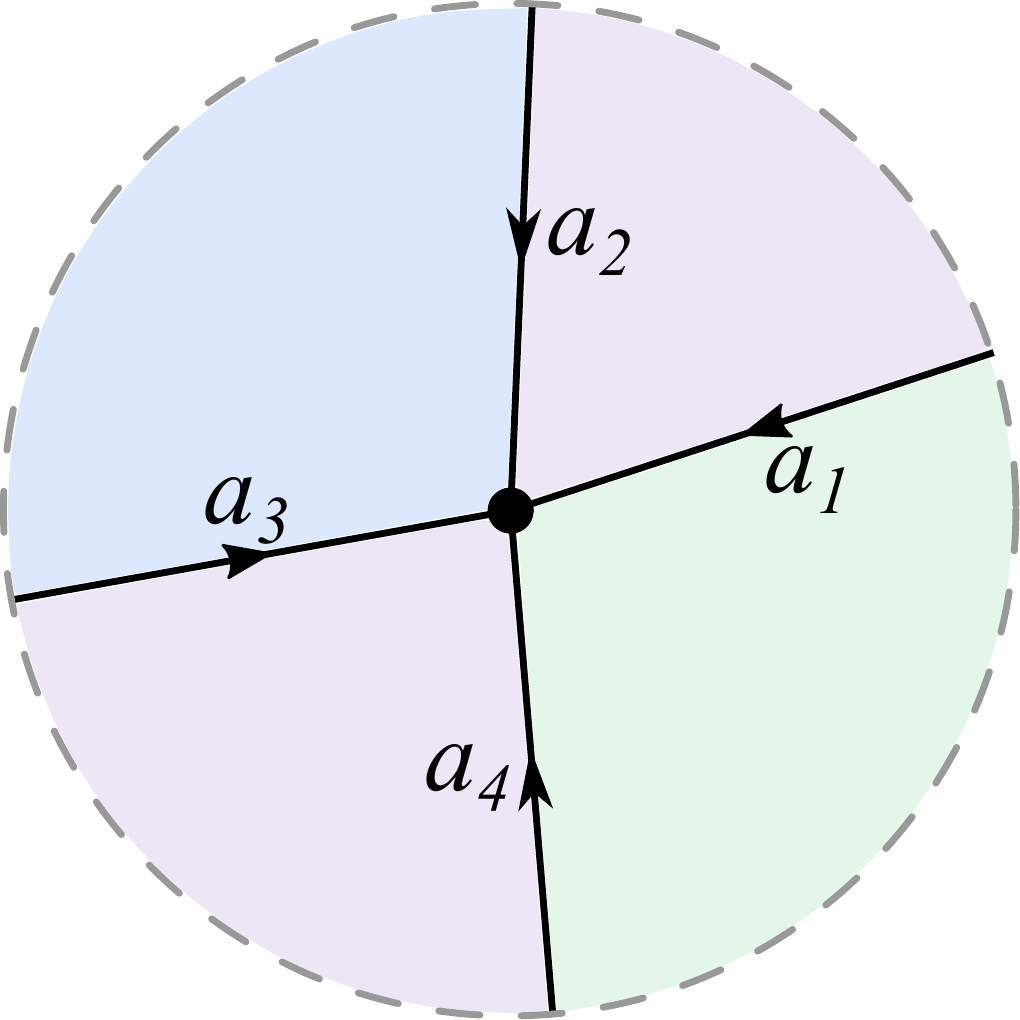}
    \caption{A cyclic sequence of successors. Note that multiple edges in a cyclic sequence of successors may appear on the same face.}
    \label{fig:cyclic_sequence_successors}
\end{figure}

\subsubsection{Cell Complexes and Combinatorial Surfaces}
\label{sec:cell_surf}

A combinatorial surface $S$ defines a cell complex $C=(E,F,B)$. The edges of $S$ are the edges of $C$ and the triangles of $S$ are the faces of $C$. While an oriented edge is purely a formal construction, we associate the oriented edges of an edge $\{u,v\}\in S$ with the ordered pairs $(u,v)$ and $(v,u)$. For a triangle $A=\{u,v,w\}\in S$, we define the boundary of $A$ as $B(A)=(\overbar{(u,v),(v,w),(w,u)})$ and $B(A^{-1})=(\overbar{(u,w),(w,v),(v,u)})$. All edges in a combinatorial surface are incident to at most two triangles, so conditions (1) and (2) in the definition of cell complex are satisfied.
\par
The conditions on the link of the vertices in $S$ can be described in the language of cell complexes. Specifically, the vertices of $S$ define sequences of successors in $C$. Let $(u,v)$ be an oriented edge. We say $(u,v)$ \textbf{\textit{enters}} $v$.

\begin{proposition}
\label{prop:inner_vertex}
Let $S$ be a combinatorial surface and $C$ the cell complex defined by $S$. Let $v\in S$ be a vertex such that $\lk_{S}{v}$ is a simple cycle. The set of edges entering $v$ form a cyclic sequence of successors.
\end{proposition}

\begin{proof}

Let $\lk_{K}{v}$ be the simple cycle $(\overbar{v_1,\ldots,v_k})$. For each $v_i\in\lk_{S}{v}$, one can verify using the definition of the link that $v$ is incident to the triangles $\{v,v_{i-1},v_{i}\}$ and $\{v,v_{i},v_{i+1}\}$. The boundary of $\{v,v_{i-1},v_{i}\}$ in $C$ is $(\overbar{(v,v_{i-1}), (v_{i-1},v_{i}),(v_{i},v)}).$ The boundary of the inverse of $\{v,v_{i},v_{i+1}\}$ is $(\overbar{(v_{i+1},v_{i}),(v_{i},v),(v,v_{i+1})})$. So $(v,v_{i-1})=(v_{i-1},v)^{-1}$ and $(v,v_{i+1})=(v_{i+1},v)^{-1}$ are a pair of successors to $(v_{i},v)$. The proposition follows by the definition of a cyclic sequence of successors.
\end{proof}

\begin{proposition}
\label{prop:boundary_vertex}
    Let $S$ be a combinatorial surface and $C$ the cell complex defined by $S$. Let $v\in S$ such that $\lk_{S}{v}$ is a simple path. The set of edges entering $v$ form a sequence of successors.
\end{proposition}
\begin{proof}
    The proof of this proposition is analogous to the proof of Proposition \ref{prop:inner_vertex}.
\end{proof}

\par
While a combinatorial surface can be represented as a cell complex, not all cell complexes can be represented as combinatorial surfaces.  A face in a cell complex can have more than three edges on its boundary, so we cannot just reverse the construction. However, there is always a combinatorial surface \textit{homeomorphic} to any cell complex. A cell complex describes a compact surface with boundary, and there is a combinatorial surface homeomorphic to any compact surface with boundary. Phrased differently, any compact surface with boundary can be triangulated. This is a famous result known as Rado's Theorem; see for instance \cite{gallier_xu}. 

\subsubsection{Surface Classification}
\label{sec:surface_classification}

Many cell complexes describe the same surface up to homeomorphism. We can define an equivalence relation on cell complexes that partitions cell complexes into homeomorphism classes. This is the famous Classification Theorem of Compact Surfaces.
\par
Let $C$ be a cell complex. A cell complex $C'$ is an \textit{\textbf{elementary subdivision}} of $C$ if (1) a pair of oriented edges $a$ and $a^{-1}$ in $C$ are replaced by two oriented edges $bc$ and $c^{-1}b^{-1}$ respectively in $C'$ in all boundaries containing $a$ or $a^{-1}$, where $b,c$ are distinct edges in $C$ and not in $C'$ or (2) a face $A$ in $C$ with $B(A)=(\overbar{a_1...a_ka_{k+1}...a_{l}})$ is replaced in $C'$ with two faces $A',A''$ such that $B(A')=(\overbar{a_1,...,a_k,c})$ and $B(A'')=(\overbar{c^{-1}a_{k+1}...a_{l}})$, where $A',A'',c$ are not in $C$ and the reverse operation is applied to $A^{-1}$. 
\par
Two cell complexes $C$ and $C'$ are \textit{\textbf{equivalent}} if they are equivalent in the least equivalence relation containing the elementary subdivision relation. The equivalence classes of the elementary subdivision relation are exactly the homeomorphism classes of compact surfaces, as evidenced by the theorem of Gallier and Xu \cite{gallier_xu}. 
\par
\begin{theorem}
[Classification Theorem for Compact Surfaces, Lemma 6.1 \cite{gallier_xu}]
\label{thm:classification_surfaces}

Each connected cell complex is equivalent to a cell complex $C=(F,E,B)$ with a single face $F=\{A\}$ and
$$
B(A)=(\overbar{a_1c_1a_1^{-1}c_1^{-1}...a_gc_ga_g^{-1}c_{g}^{-1}d_1e_1d_1^{-1}...d_be_bd_b^{-1}})
$$
in which case $C$ is an orientable surface of genus $g\geq 0$ with $b\geq 0$ boundary components, or
$$
B(A)=(\overbar{a_1a_1...a_ga_gd_1e_1d_1^{-1}\ldots d_be_bd_b^{-1}})
$$
in which case $C$ is a non-orientable surface of genus $g\geq 1$ with $b\geq 0$ boundary components.
\end{theorem}
We call a substring of a boundary of the form $aca^{-1}c^{-1}$ a \textit{\textbf{handle}}, a substring of the form $aa$ a \textit{\textbf{crosscap}}, and a substring of the form $ded^{-1}$ a \textit{\textbf{boundary}} if the edge $e$ is only on the boundary of one face. A cell complex is \textit{\textbf{non-orientable}} if it is equivalent to a cell complex with a crosscap and \textit{\textbf{orientable}} otherwise. The above theorem states that a cell complex is characterized by the number of handles and boundaries or the number of crosscaps and boundaries it has. The \textit{\textbf{genus}} of a cell complex is the number of handles or the number of crosscaps a cell complex has. An arbitrary cell complex is not connected, so an arbitrary cell complex is a collection of connected cell complexes of the above form.

\subsubsection{Equivalence-Preserving Moves and Annotated Cell Complexes}
\label{sec:moves}
One advantage of using cell complexes to describe surfaces is the equivalence relation defined by elementary subdivision. This equivalence relation allows us to develop a useful algebra on cell complexes. In this section, we give a concise algebraic description of cell complexes as a formal sum. We then define a collection of \textit{\textbf{equivalence-preserving moves}} on these formal sums, operations we can perform on a cell complex that preserves its equivalence class. We will then use these equivalence-preserving moves to define a more general data structure, the annotated cell complex.
\par
We now give an algebraic description of a cell complex that will allow us to easily describe and reduce cell complexes to their canonical form. Let a \textit{\textbf{section}} be a map $\phi:F\to F\sqcup F^{-1}$ such that $\phi(A)\in\{A,A^{-1}\}$ for each $A\in F$. If we fix a section $\phi$, we can represent a cell complex as the formal sum $\sum_{A\in F}B(\phi(A))$. Note that any such formal sum uniquely determines a cell complex regardless of the section $\phi$ used, as the boundary of any oriented face $A\not\in\im\phi$ is determined by the boundary of $A^{-1}\in\im\phi$. We will express a cell complex using these formal sums as
$$
\sum_{A\in F}B(\phi(A))=(\overbar{a_1\ldots a_k})+\ldots+(\overbar{a_l\ldots a_m}).
$$
We will use upper-case variables to denote substrings of boundaries, e.g. $X=a_i\ldots a_j$.
\par 
We are only interested in these formal sums up to the equivalence relation described in the previous section. We now introduce a series of operations on these formal sums that preserve the equivalence class of a cell complex.
\begin{enumerate}[font=\bfseries]
\item The first elementary subdivision says that an edge $a$ can be replaced with two edges $bc$. In the sum, we can replace summands $(\overbar{aX})+(\overbar{aY})=(\overbar{bcX})+(\overbar{bcY})$ to get an equivalent cell complex. As the elementary subdivision equivalence relation is symmetric, we can perform this move in reverse. If two edges $bc$ appear consecutively in two boundaries, we can replace these two edges with a single edge $a$, i.e. $(\overbar{bcX})+(\overbar{bcY})=(\overbar{aX})+(\overbar{aY})$.

\item The second elementary subdivision says we can replace a single face $A$ with two faces $A',A''$ that share an edge. We can express this algebraically as replacing a face $(\overbar{XY})$ with two faces $(\overbar{Xa})+(\overbar{a^{-1}Y}).$ As the elementary subdivision equivalence relation is symmetric, we can also apply this move in reverse, i.e. $(\overbar{Xa})+(\overbar{a^{-1}Y})=(\overbar{XY})$. 
\begin{figure}[H]
    \centering
    \includegraphics[width=0.45\linewidth]{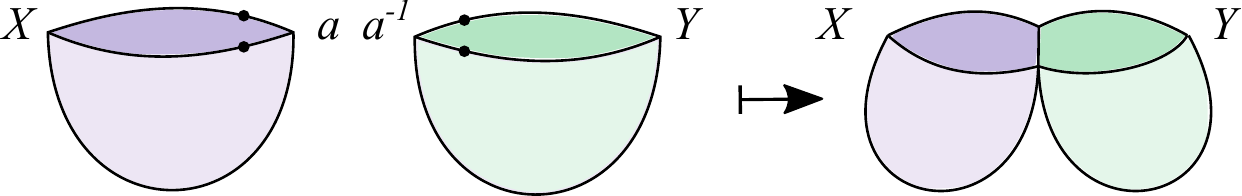}
    \captionsetup{margin=1cm}
    \caption{An example of Move 2. Two faces are glued along their common edge $a$ to create a single face.}
    \label{fig:merge}
\end{figure}

\item Any section $\phi:F\to F\sqcup F^{-1}$ defines the same surface, so we can interchange the section $\phi$ for any other section $\psi$. This replaces one or more summands $B(A)=(\overbar{a_1\ldots a_k})$ in the formal sum with $B(A^{-1})=(\overbar{a_k^{-1}\dots a_1^{-1}})$. 

\item As each summand $(\overbar{a_1\ldots a_k})$ is a cyclically ordered sequence, we can replace a face $(\overbar{a_1\ldots a_k})$ with $(\overbar{a_2\ldots a_k a_1)}$. 

\item We can remove any instance of $aa^{-1}$ from a boundary. Lemma \ref{lem:remove_aa-1} proves this.
\begin{lemma}
\label{lem:remove_aa-1}
Let $K$ be a cell complex. Let $A$ be a face of $K$ such that $B(A)=(\overbar{Xaa^{-1}})$. There is an equivalent cell complex $K'$ without $a$ and $A$ and with a new face $A'$ such that $B(A')=(\overbar{X})$.
\end{lemma}
\begin{proof}
We apply a series of equivalence-preserving moves to $(\overbar{Xaa^{-1}})$.
\begin{align*}
    (\overbar{Xaa^{-1}})=&(\overbar{Xab})+(\overbar{b^{-1}a^{-1}})&&\hfill\text{by (2)}\\
    =&(\overbar{Xab})+(\overbar{ab})&&\hfill\text{by (3)}\\
    =&(\overbar{Xc})+(\overbar{c})&&\hfill\text{by (1)}\\
    =&(\overbar{Xc})+(\overbar{c^{-1}})&&\hfill\text{by (3)}\\
    =&(\overbar{X})&&\hfill\text{by (2)}\qedhere
\end{align*}
\end{proof}

\begin{figure}[H]
    \centering
    \includegraphics[width=0.35\textwidth]{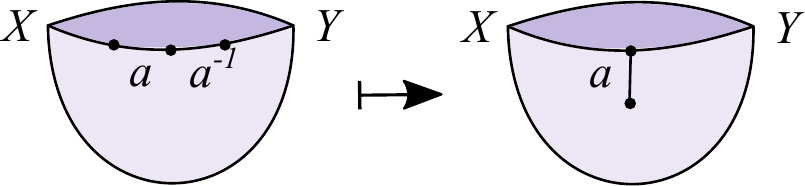}
    \captionsetup{margin=1cm}
    \caption{An example of Move 5. Identifying consecutive edges $a$ and $a^{-1}$ removes these edges from the boundary of the face.}
    \label{fig:edge_inverse}
\end{figure}

\item\textbf{Boundary Components.} The next two moves are inspired by a lecture by Wildberger \cite{wildberger_zip}. If $b$ and $b^{-1}$ both appear in the boundary of a face $B(A)=(\overbar{XbYb^{-1}})$, we think of $X$ and $Y$ as being separate boundary components of the sphere connected by a path $b$. As the interior of the face is an open disk and is therefore path connected, we can connect these boundary components anywhere along $X$ or $Y$. The following lemma formalizes this idea.

\begin{lemma}
\label{lem:cell_multiple_boundaries}
Let $K$ be a cell complex. Let $A$ be a face of $K$ such that $B(A)=(\overbar{X_1X_2bYb^{-1}})$. There is an equivalent cell complex $K'$ without $b$ and $A$ and with a new edge $c$ and a new face $A'$ such that $B(A')=(\overbar{X_2X_1cYc^{-1}})$.
\end{lemma}

\begin{proof}
We prove this by repeating applying moves (4) and (2) to the boundary of $A$.
\begin{align*}
    (\overbar{X_1X_2bYb^{-1}})&=(\overbar{b^{-1}X_1X_2bY})&&\hfill\text{by (4)}\\
    &=(\overbar{b^{-1}X_1c})+(\overbar{c^{-1}X_2bY})&&\hfill\text{by (2)}\\
    &=(\overbar{X_1cb^{-1}})+(\overbar{bYc^{-1}X_2})&&\hfill\text{by (4)}\\
    &=(\overbar{X_1cYc^{-1}X_2})&&\hfill\text{by (2)}\\
    &=(\overbar{X_2X_1cYc^{-1}})&&\hfill\text{by (4)}\qedhere
\end{align*}
\end{proof}

The lemma tells us that for a face containing the edges $b$ and $b^{-1}$ of the form $(\overbar{XbYb^{-1}})$, we can cyclically permute $X$ and $Y$ arbitrarily so long as they are connected by some edge and its inverse. We express this choice of connection by introducing new notation. We write $(\overbar{XbYb^{-1}})=(\overbar{X})(\overbar{Y})$ where the formal multiplication of boundaries denotes $(\overbar{X})$ and $(\overbar{Y})$ are connected by some edge $b$ and its inverse $b^{-1}$.  If $B(A)=(\overbar{X})(\overbar{Y})$, we call each factor a \textit{\textbf{boundary component}} of $A$.
\par 
The individual boundary components of $A$ have the same properties as the entire boundary $B(A)$. The lemma implies that $(\overbar{X_1X_2})(\overbar{Y})=(\overbar{X_2X_1})(\overbar{Y})$, i.e. that we can cyclically permute a boundary component while maintaining equivalence. While we cannot invert just one of $(\overbar{X})$ or $(\overbar{Y})$ while maintaining equivalence, it is easy to verify that $(\overbar{X})(\overbar{Y})=(\overbar{Y^{-1}})(\overbar{X^{-1}})$.
\par 
We can also have more than two boundary component. If we have boundary components $(\overbar{X})(\overbar{YaZa^{-1}})$, we can write write this $(\overbar{X})(\overbar{Y})(\overbar{Z})$. We can arbitrarily permute boundary components while maintaining equivalence. This can be seen as 
$$
(\overbar{X})(\overbar{Y})(\overbar{Z})=(\overbar{X})(\overbar{YaZa^{-1}})=(\overbar{X})(\overbar{ZaYa^{-1}})=(\overbar{X})(\overbar{Z})(\overbar{Y}).
$$
\par
For a face $(\overbar{XbYb^{-1}})$, if we remove $b$ to create the face $(\overbar{X})(\overbar{Y})$, we define the successors of edges as substrings of the individual boundary components. In particular, in the face $(\overbar{a_1\ldots a_k})
(\overbar{X})$, then we defined the successor of $a_k$ to be $a_1$.

\begin{figure}[H]
    \centering
    \includegraphics[width=0.45\linewidth]{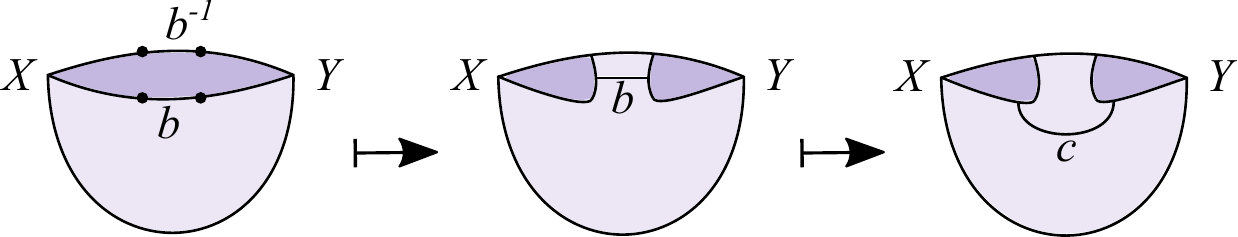}
    \captionsetup{margin=1cm}
    \caption{An example of Move 6. If we identify the edges $b$ and $b^{-1}$ on the boundary of the face on the right, then the face turns into the sphere with two boundary components in the middle. We can connect then these boundary components with any edge $c$, not just $b$, as in the surface on the right.}
    \label{fig:boundaries}
\end{figure}

\item\textbf{Annotations.} Theorem \ref{thm:classification_surfaces} tells us that a surface is completely characterized by the number of handles or crosscaps and number of boundaries it has. We take this idea a step further and prove that up to equivalence a handle, crosscap, or boundary on a face is independent of the rest of the face. The following lemmas formalize this idea. Proofs of these Lemmas can be found in Appendix \ref{appendix:handle_crosscap_boundary}.

\begin{lemma}
\label{lem:cell_handle}
Let $A$ be a face of $K$ such that $B(A)=(\overbar{aba^{-1}b^{-1}XY})$. There is an equivalent cell complex $K'$ without $A$, $a$, and $b$ and with a new face $A'$ and edges $e,f$ such that $B(A')=(\overbar{efe^{-1}f^{-1}YX})$.
\end{lemma}

\begin{lemma}
\label{lem:cell_crosscap}
Let $A$ be a face of $K$ such that $B(A)=(\overbar{aaXY})$. There is an equivalent cell complex $K'$ without $A$ and $a$ a new face $A'$ and edge $d$ such that $B(A')=(\overbar{ddYX})$.
\end{lemma}

\begin{lemma}
\label{lem:cell_boundary}
Let $A$ be a face of $K$ such that $B(A)=(\overbar{bab^{-1}XY})$ such that $a$ appears once in the boundary of all faces. There is an equivalent cell complex $K'$ without $A$ and $b$, a new face $A'$ and edges $c,d$ such that $B(A')=(\overbar{cdc^{-1}YX})$. 
\end{lemma}

If we have a face $(\overbar{HX})$ where $H$ is a handle, we can store a face equivalent to $(\overbar{HX})$ by storing $(\overbar{X})$ as a face and simply noting that $(\overbar{X})$ has a handle. We can later attach a handle anywhere along $(\overbar{X})$ by Lemma \ref{lem:cell_handle} and regain an equivalent cell complex. The same applies if $H$ is a crosscap or a boundary. These lemmas motivates a new data structure, the \textit{\textbf{annotated cell complex}}. An annotated cell complex is a cell complex where each face is annotated with a genus, number of boundary components, and a boolean to indicate whether or not this face is orientable. 
\par
With this extra information, we can store an annotated cell complex equivalent to a cell complex that is defined with fewer edges. Furthermore, by Theorem \ref{thm:classification_surfaces}, any connected cell complex is equivalent to an annotated cell complex with \textit{no} edges. For example, the cell complex of the torus $(\overbar{aba^{-1}b^{-1}})$ is equivalent to the annotated cell complex $()$ with the face annotated to have genus 1, 0 boundary components, and to be orientable. If a cell complex is disconnected, then the cell complex is equivalent to a sum of faces $()$ with no edges and just annotations. We call a face with no edges an \textit{\textbf{empty face}}. 
\begin{figure}[H]
        \centering
        \begin{subfigure}{0.4\textwidth}
            \centering
            \includegraphics[height=1in]{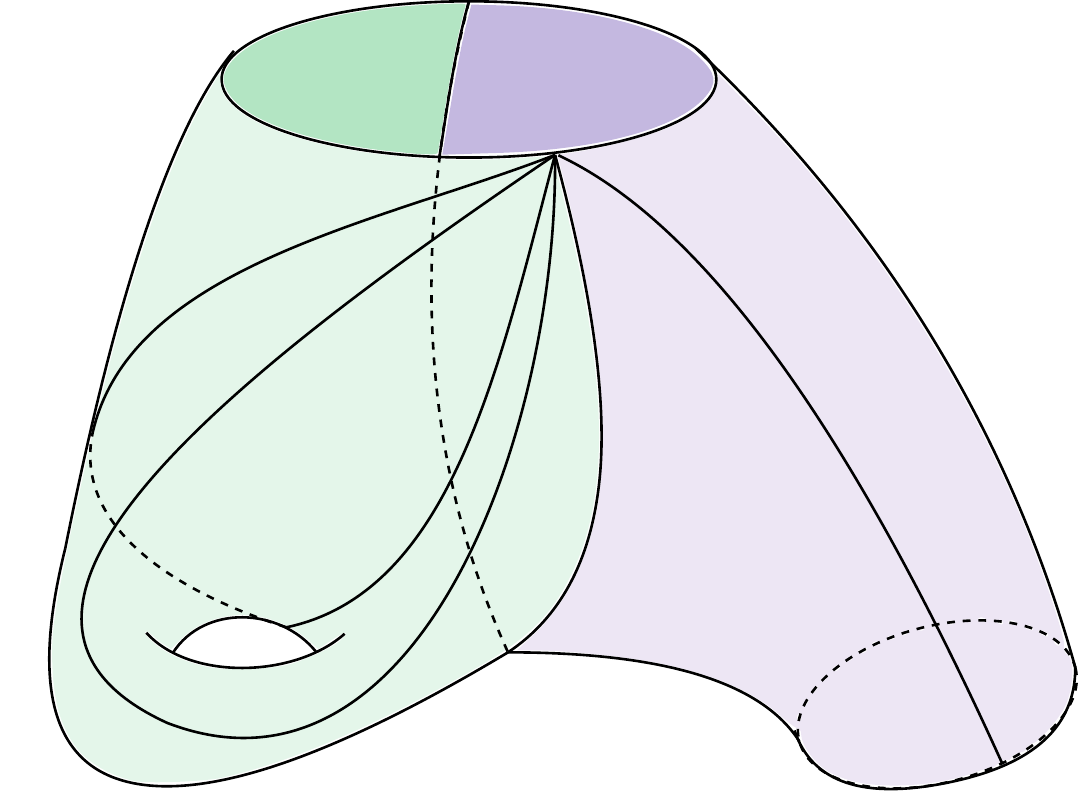}    
        \end{subfigure}
        \begin{subfigure}{0.55\textwidth}
            \centering
            \includegraphics[height=1in]{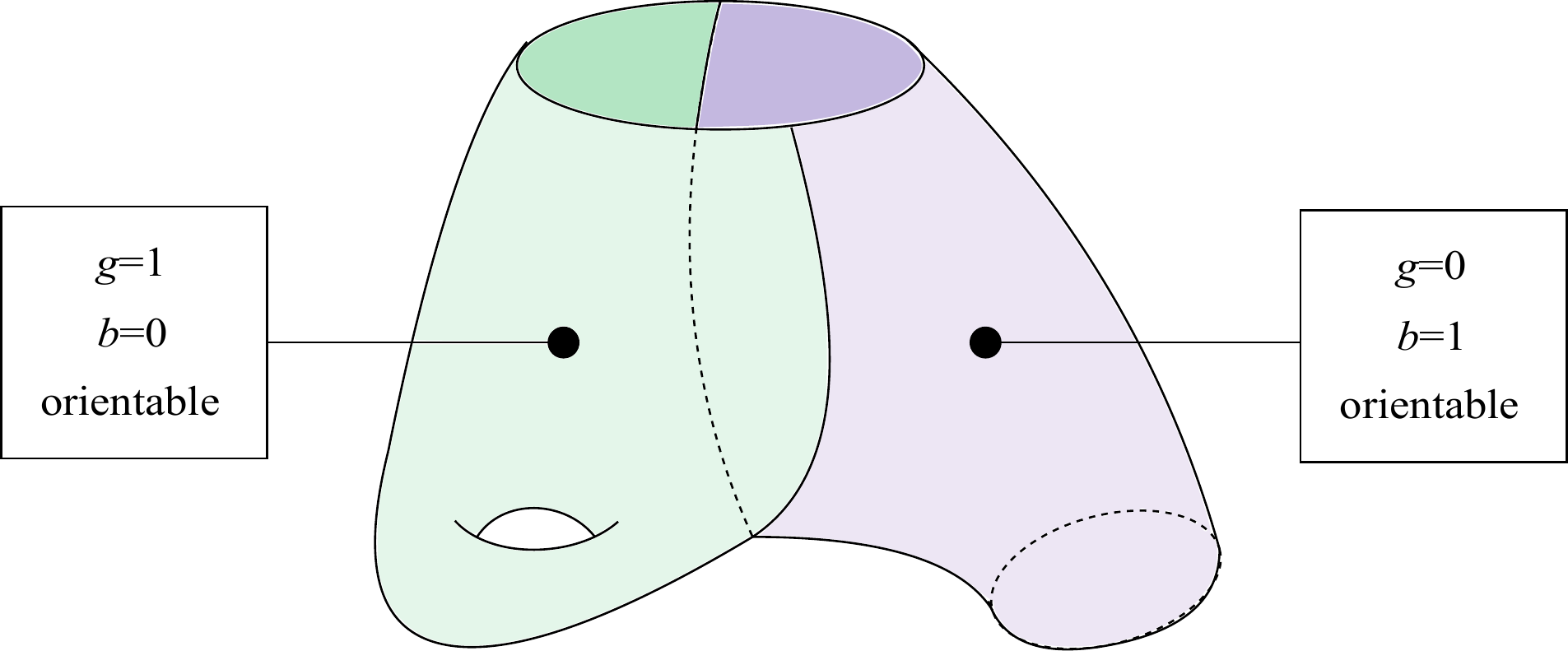}
        \end{subfigure}
        \label{fig:annotations}
        \captionsetup{margin=1cm}
        \caption{Annotations allow us to store cell complexes using fewer edges. In the cell complex in this example, the face on the left has a handle and the face on the right has a boundary. We are able to record these features with annotations instead of the edges that actually make up these features.}
\end{figure}

\end{enumerate}

\subsubsection{Cell Complex Miscellany}
We now present several lemmas without commentary that we will use later for classifying surfaces. The lemmas give more general criteria for identifying a handle or crosscap on a face than having substrings of the form $aba^{-1}b^{-1}$ or $aa$ in its boundary. Each of the lemmas is from the proof of Lemma 6.1 in \cite{gallier_xu}; each of the corollaries is some restatement of the lemmas using the notation of boundary components.

\begin{lemma}
\label{lem:handle_equivalence}
Let $A$ be a face of $K$ such that $B(A)=(\overbar{aUbVa^{-1}Xb^{-1}Y})$. There is an equivalent cell complex $K'$ without $A$, $a$, and $b$ and with a new face $A'$ and new edges $c,d$ such that $B(A')=(\overbar{cdc^{-1}d^{-1}YXVU})$.
\end{lemma}

\begin{corollary}
\label{cor:handle_boundary_equivalence}
Let $A$ be a face of $K$ such that $B(A)=(\overbar{aX})(\overbar{a^{-1}Y})$. There is an equivalent cell complex $K'$ without $A$ and $a$ and with a new face $A'$ and new edge $c,d$ such that $B(A')=(\overbar{cdc^{-1}d^{-1}YX})$.
\end{corollary}

\begin{lemma}
\label{lem:crosscap_boundary}
Let $A$ be a face of $K$ such that $B(A)=(\overbar{aXaY})$. There is an equivalent cell complex $K'$ without $A$ and $a$ and with a new face $A'$ and new edge $b$ such that $B(A')=(\overbar{bbY^{-1}X})$.
\end{lemma}

\begin{corollary}
\label{cor:cell_crosshandle}
Let $A$ be a face of $K$ such that $B(A)=(\overbar{aX})(\overbar{aY})$. There is an equivalent cell complex $K'$ without $A$ and $a$, a new face $A'$ and new edges $b,c$ such that $B(A')=(\overbar{ccXbbY^{-1}})$.
\end{corollary}

\begin{lemma}[Dyck's Theorem]
\label{lem:handle_crosscap_equivalence}
Let $A$ be a face of $K$ with a handle and a crosscap. There is an equivalent cell complex $K'$ without $A$  with a new face $A'$ such that $A'$ has three crosscaps.
\end{lemma}


\subsection{Closed Tree Decomposition} 
\label{sec:closed_td}

Our algorithm will use a special type of tree decomposition of the Hasse diagram that we call a \textit{\textbf{closed tree decomposition}}. In this section, we prove that we can convert any tree decomposition of the Hasse diagram to a closed tree decomposition while only increasing the width by a constant multiplicative factor. We then prove that we can always find a closed tree decomposition that is also nice. 
\par 
Let $H$ denote the Hasse diagram of a 2-dimensional simplicial complex $K$, and let $(T,X)$ be a tree decomposition of $H$. We define the \textit{\textbf{closure}} of $(T,X)$ to be the pair $(T,C)$, where $C$ is a set of bags of $T$ with $C_t = \cl(X_t)$. The bags $C_t$ are simplicial complexes, as the face of any simplex in $C_t$ is also contained in $C_t$. We claim that $(T,C)$ satisfies the conditions to be a tree decomposition of $H$. 

\begin{lemma}
\label{lem:td_to_closed_td}
     Let $H$ be the Hasse diagram of a 2-dimensional simplicial complex $K$. Let $(T,X)$ be a tree decomposition of $H$ of width $k$, and let $(T,C)$ be the closure of $(T,X)$. The pair $(T,C)$ is a tree decomposition of $H$. Moreover, the width of $(T,C)$ is $\OO(k)$.\footnote{The width of the tree decomposition only increases by a constant factor as $K$ is 2-dimensional. If $K$ were $d$-dimensional, then we could only bound the treewidth of the closure by $O(dk)$.}
\end{lemma}
\begin{proof}
    We first verify that $(T,C)$ satisfies the definition of being a tree decomposition of $H$. The first two conditions of a tree decomposition are that all vertices and all edges of $H$ are contained in some bag $C_t$; indeed, this follows from the fact that each vertex and edge of $H$ are contained in some bag $X_t$, and $X_t\subset C_t$ for each node $t\in T$
    \par 
    We next verify that the tree $T^C_{\sigma} = \{t\in T:\sigma\in C_t\}$ is connected for each $\sigma\in K$. We will use the fact that the trees $T^X_{\sigma} = \{t\in T:\sigma\in X_t\}$ are connected for each simplex $\sigma\in K$ as $(T,X)$ is a valid tree decomposition. A simplex $\sigma$ is contained in a bag $C_t$ iff there is a simplex $\tau\in X_t$ such that $\sigma\subset \tau$; this follows from the definition of closure. Therefore, the tree $T^{C}_{\sigma}$ equals the union $\cup_{\tau\in K:\sigma\subset\tau} T_{\tau}^{X}$. 
    \par 
    We need one more observation to verify that $T^{C}_{\sigma}$ is connected. Let $\sigma,\tau\in K$ such that $\sigma$ is a codimension 1 face of $\tau$; for example, $\sigma$ is an edge and $\tau$ is a triangle. We claim the tree $T^{X}_{\sigma}\cup T^{X}_{\tau}$ is connected. Indeed, both trees $T^{X}_{\sigma}$ and $T^{X}_{\tau}$ are individually connected, and there must be a vertex where the two trees intersect, as there must be some bag $X_t$ containing both $\sigma$ and $\tau$ as $\sigma$ and $\tau$ are connected by an edge in $H$. 
    \par 
    To see that $T^{C}_{\sigma}$ is connected, consider incrementally adding each tree to the union $\cup_{\tau\in K:\sigma\subset\tau} T_{\sigma}^{X}$ in order of dimension, i.e. add all tree corresponding to vertices, then all trees corresponding to edges, and so on. At each iteration, the union will be connected. There is a unique tree of smallest dimension: $T_\sigma^{X}$. Any tree $T_{\tau}^{X}$ added after $T_{\sigma}^{X}$ will have a codimension 1 face already in the union, so $T_{\tau}^{X}$ will intersect the current union in at least one vertex. Therefore, the union with $T_{\tau}^{X}$ is also connected. Applying this argument inductively, we see that $\cup_{\tau\in K:\sigma\subset\tau} T_{\sigma}^{X} = T^{C}_{\sigma}$ is connected.
    \par 
    We now show that the width of $(T,C)$ is $\OO(k)$, where $k$ is the width of $(T,X)$. Consider a bag $X_t$. The closure $C_t = \cl(X_t) = \cup_{\sigma\in X_t}\cl(\sigma)$. The closure of a triangle contains 7 simplices (namely, the triangle itself, 3 edges and 3 vertices); the closure of an edges contains 3 simplices. The closure of a vertex contains 1 simplex. Therefore, $|C_t| \leq \sum_{\sigma\in X_t} |\cl(\sigma)| \leq 7|X_t|$. It follows that the width of $(T,C)$ is at most $7k+6$.
\end{proof}

 In general, we define a \textit{\textbf{closed tree decomposition}} of $H$ to be a tree decomposition of $H$ such that each bag $X_t$ is a simplicial complex. Going forward, we will always assume the tree decompositions of $H$ are closed, as we can always convert an arbitrary tree decomposition of $H$ to a closed tree decomposition. Additionally, we can also assume we have a \textit{nice} closed tree decomposition of $H$, as we prove in the following lemma.
  
 \begin{lemma}
    Let $(T,X)$ be a closed tree decomposition of $H$ of width $k$. There is a nice, closed tree decomposition of $H$ of width $k$ with $\OO(kn)$ nodes that can be computed in time $\OO(k^{2}\max\{|K|,|T|\})$. 
\end{lemma}
\begin{proof}
      Given an arbitrary tree decomposition of $H$, there is a classic algorithm to make a nice tree decomposition with the given running time. We describe this algorithm and show how we can adapt this algorithm so that if the input is closed, the output will be closed too. The algorithm is as follows. Add an empty bag to each leaf. Arbitrarily pick a leaf to be the root. Replace each node $t$ with $c\geq 3$ children with a binary tree with $c$ leaves, each connected to a previous child of $t$, and set the bag of each node in the binary tree to $X_t$. For each pair of non-join neighbors $t$ and $t'$, add a path of nodes in between $t$ and $t'$. Consecutive nodes on this path should forget a single simplex in $X_{t'}\setminus X_t$ or introduce a single simplex in $X_{t}\setminus X_{t'}$. To ensure the output tree decomposition is closed, we must forget all the triangles, then edges, then vertices. As a simplex is forgotten before each of its faces, the bags are still closed. We must then introduce vertices, then edges, then triangles. As the faces of a simplex are added before the simplex itself, the bags are still closed. 
\end{proof}

\subsection{Subcomplexes}
\label{sec:subcomplexes}

 A closed tree decomposition $(T,X)$ of the Hasse diagram of a simplicial complex $K$ defines a recursively nested series of subcomplexes of $K$. We can use this series of subcomplexes to recursively build solutions to our problems. In this section, we define these subcomplexes and prove some properties that will be useful in designing our algorithm.
 \par 
 Let $(T,X)$ be a nice, closed tree decomposition of $H$. We define a subcomplex at each node $t$ of the tree decomposition. The \textit{\textbf{subcomplex rooted at t}} is $K_t = (\cup_{d\in D(t)} X_{d})\setminus (X_t)_2$, where $D(t)$ is the set of all descendants of $t$ including $t$ itself and $(X_t)_2$ is the set of triangles in $X_t$. Lemma \ref{lem:subcomplex_is_complex} proves that $K_t$ is indeed a simplicial complex. 
 \begin{lemma}
 \label{lem:subcomplex_is_complex}
    The set $K_t$ is a simplicial complex.
 \end{lemma}
 \begin{proof}
     Observe that $\cup_{d\in D(t)} X_{d}$ is a simplicial complex as each bag $X_d$ is a simplicial complex. Moreover, $K_t$ is a simplicial complex; removing $(X_t)_2$ does not break the face-closure property of being a simplicial complex as the triangles in $(X_t)_2$ are not the face of any simplex but themselves.
 \end{proof}
 One useful property of this definition is that if $t'$ is descendant of $t$, then $K_{t'}\subset K_t$. We prove this in the following lemma. 
\begin{lemma}
\label{lem:nested_subcomplexes}
    Let $t$ be a node of a closed tree decomposition, and let $t'$ be a descendant of $t$. Then $K_{t'}\subset K_t$.
\end{lemma}
\begin{proof}
    The union $\cup_{d'\in D(t')}X_{d'} \subset \cup_{d\in D(t)} X_{d}$ as any descendant of $t'$ is also a descendant of $t$. We must now show that any triangle $\sigma\in (X_t)_2$ is not in $K_{t'}$. If $\sigma\notin \cup_{d'\in D(t')} X_{d'}$, then this is immediate. Alternatively, if $\sigma\in X_{d'}$ for some descendant $d'$ of $t'$, then we claim that $\sigma\in (X_{t'})_2$. Indeed, the subtree $T_\sigma$ of $T$ of all nodes whose bag contain $\sigma$ is connected, and $t'$ is in this subtree $T_\sigma$ as $t'$ lies on the unique path connecting $t$ and $d'$. 
\end{proof}
We prove one final lemma in this section. 
\begin{lemma}
\label{lem:surface-coface}
    Let $t$ be a node in the tree decomposition $(T,X)$. Let $\sigma\in K_t\setminus X_t$, and let $\tau$ be any coface of $\sigma$. Then $\tau\in K_t\setminus X_t$. 
\end{lemma}
\begin{proof}
    We have that $\tau\notin X_t$ because if $\tau\in X_t$, then the fact that $\sigma\subset\tau$ would imply $\sigma\in X_t$ as $X_t$ is a simplicial complex, contradicting the assumption that $\sigma\notin X_t$. Next, we prove that $\tau\in K_t$. By the definition of tree decomposition, there is a node $t_\tau$ in $T$ such that $\tau\in X_{t_\tau}$. Suppose for the purposes of contradiction that $t_\tau$ is not a descendant of $t$. As $\sigma\subset\tau$, then $\sigma\in X_{t_\tau}$ as well. We also know there is a descendant $t_\sigma$ of $t$ such that $\sigma\in K[X_{t_\sigma}]$ by the definition of $K_t$. The set of nodes containing $\sigma$ form a connected subtree of $T$ by the definition of tree decomposition; however, $\sigma\notin X_t$ by assumption, a contradiction as $t$ lies on the unique path between $t_\tau$ and $t_\sigma$. So $t_\tau$ must be a descendant of $t$.
\end{proof} 

\subsection{Candidate Solutions}
\label{sec:candidate_solutions}
 Our algorithm is a dynamic program on the tree in the tree decomposition. At each node $t$ of our tree, we will store the set of subcomplexes of $K_t$ that could be extended to surfaces with boundary $B$. We call these subcomplexes \textit{\textbf{candidate solutions}}. As the subcomplexes $K_t$ in the tree decomposition are nested, then we can build the candidate solutions at a node $t$ by extending the candidate solutions at $t$'s children. In this section, we define a candidate solution at a node $t$.  
 \par 
 To see what properties candidate solutions should have, assume a combinatorial surface $S$ with boundary $B$ exists. Consider the intersection $S_t = S \cap K_t$. The complex $S_t$ need not be a surface, or even a pure 2-complex. So, the link of a vertex $v$ in $S_t$ need not be a simple cycle or a simple path. As $\lk_{S_t}{v}\subset\lk_{S}{v}$, then $\lk_{S_t}{v}$ is either a simple cycle or a collection of simple paths and vertices. The same is true of the links of any edge $e$; namely, it is always true that $\lk_{S_t}(e) \subset \lk_{S}(e)$. However, Lemma \ref{lem:surface_link} proves that it can only be the case that $\lk_{S_t}{\sigma}\neq\lk_{S}{\sigma}$ if $\sigma\in X_t$. Intuitively, $S_t$ is ``surface-like'' everywhere except possibly in the intersection $S_t\cap X_t$.
 \begin{lemma}
 \label{lem:surface_link}
    Let $\sigma\in K_t \backslash X_t$ be a vertex or edge. Then $\lk_{S}{v}=\lk_{S_t}{\sigma}$.
 \end{lemma}
 \begin{proof}
     As $S_t\subset S$, we immediately have that $\lk_{S_t}{v}\subset\lk_{S}{v}$. We need only show that $\lk_{S}{v}\subset\lk{S_t}{v}$. Any simplex $\tau\in\lk_{S}{v}$ is incident to a common coface $\sigma$ with $v$. Since $v\in K_t\setminus K[X_t]$ and $\sigma$ is a coface of $v$, we have $\sigma\in K_t\setminus K[X_t]$ by Lemma \ref{lem:surface-coface}. Thus, $\tau\in K_t\setminus K[X_t]$ because $\tau\subset\sigma$,  and $\lk_{S}{v}\subset\lk_{S_t}{v}$.
 \end{proof}
 Lemma \ref{lem:surface_link} gives criteria for defining a candidate solution $\Sigma$ at a node $t$. We need to verify that for each simplex $\sigma\in K_t$ the link $\lk_{\Sigma}{\sigma}$ could equal $\lk_{S_t}{\sigma}$ for a surface $S$ with boundary $B$, assuming such a surface exists. The conditions we place on $\lk_{\Sigma}(\sigma)$ will depend on whether $\sigma$ is an edge or vertex, whether or not $\sigma\in X_t$, and whether or not $\sigma\in B$. 
 \par 
 If a simplex $\sigma\in K_t\setminus X_t$, then by Lemma \ref{lem:surface_link}, we know $\lk_{\Sigma}{v}=\lk_{S}{v}$ for a hypothetical solution $S$. For a vertex $v\notin B$, this is true if $\lk_{\Sigma}{v}$ is a simple cycle or if $\lk_{\Sigma}{v}$ is empty (in the case that $v\notin\Sigma$.) For a vertex $v\in B$, this is true if $\lk_{\Sigma}{v}$ is a simple path with endpoints that are the neighbors of $v$ in $B$. (Note that it \textit{must} be the case that $v\in\Sigma$ if $v\in B$.) For an edge $e\notin B$, this is true if $\lk_{\Sigma}{e}$ is zero or two vertices. If $e\in B$, this is true if $\lk_{\Sigma}$ is a single vertex. In each of these cases, we say that $\lk_{\Sigma}{\sigma}$ is \textit{\textbf{complete}}. 
 \par 
 Alternatively, if a simplex $\sigma\in X_t$, then we can only say that $\lk_{\Sigma}{v}\subset\lk_{S}{v}$ for a hypothetical solution $S$. (Alternatively, we could say that the link of $\sigma$ may become complete after adding more triangles to $\Sigma$.) For a vertex $v$, this is true if $\lk_{\Sigma}{v}$ is a either a (possibly empty) collection of simple paths, or if $v\notin B$, a single simple cycle. See Figure \ref{fig:admissible_link_example}. For an edge $e$, this is true if $\lk_{\Sigma}{e}$ contains at most two vertices if $e\notin B$ or contains at most a single vertex if $e\in B$. In each of these cases, we say that $\lk_{\Sigma}{\sigma}$ is \textit{\textbf{admissible}}. Observe that a complete link is also admissible.
 
\begin{figure}
    \centering
    \includegraphics[width=1in]{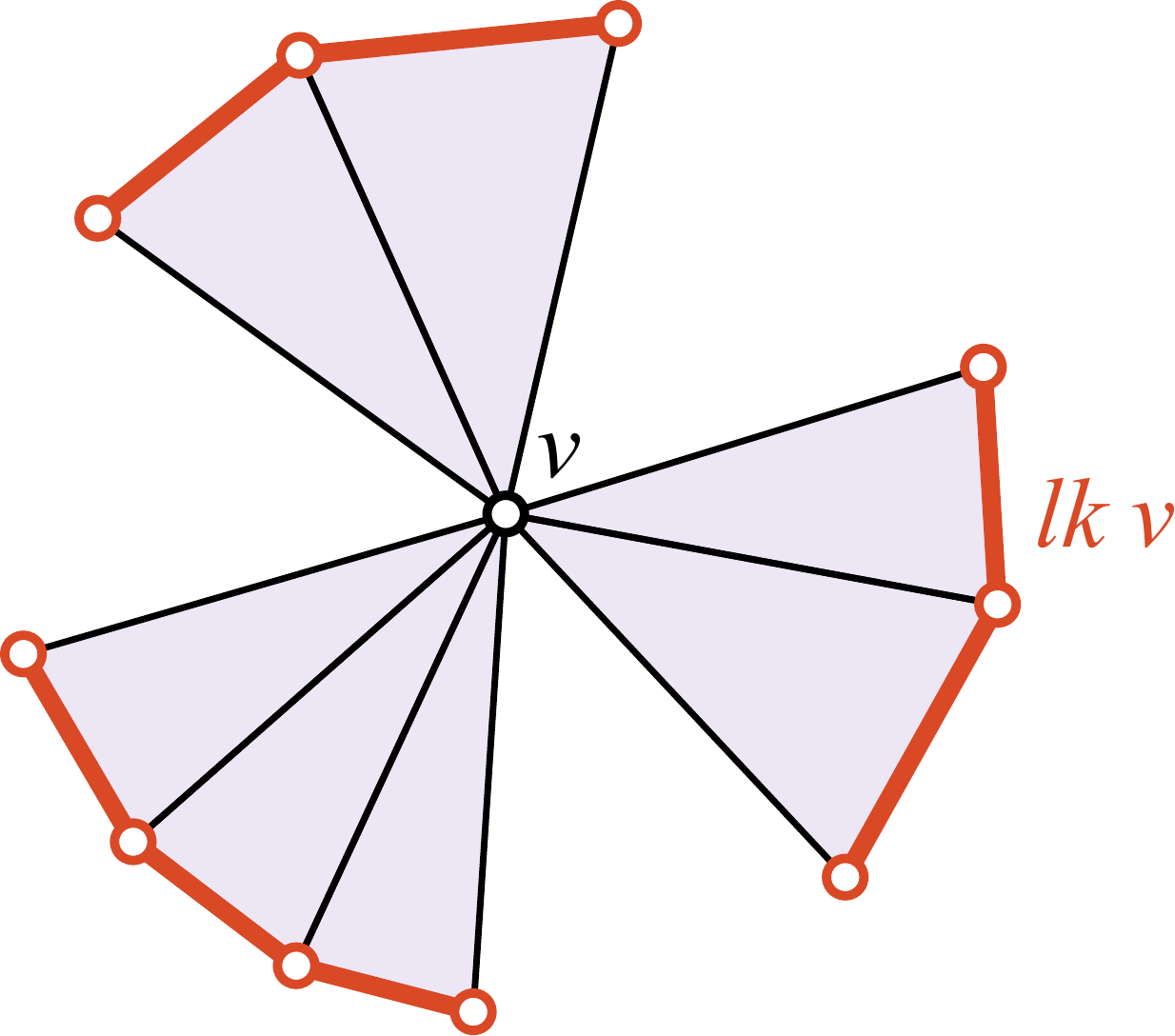}
    \caption{An example of a vertex $v$ with admissible link.}
    \label{fig:admissible_link_example}
\end{figure}

We define a \textit{\textbf{candidate solution at a node t}} to be a pure 2-dimensional subcomplex of $\Sigma\subset K_t$ such that 
\begin{enumerate}
    \item the link of each edge or vertex in $X_t$ is admissible, or
    \item the link of each edge or vertex $K_t\setminus X_t$ is complete.
\end{enumerate}
If $\Sigma$ satisfies conditions 1 and 2, we say that $\Sigma$ satisfies the \textit{\textbf{link conditions}} at $t$.
\par 
Our algorithm will compute the set of candidate solutions at each node $t$, denoted $\D[t]$. At this point, we should prove that our definition of candidate solutions is correct. That is, we must show that the set of candidate solution at the root $\D[r]$ are the set of all solutions to the problem. 

\begin{lemma}
\label{lem:dp_root}
    Let $(T,X)$ be a nice tree decomposition of the Hasse diagram with root $r$. Then the set $\D[r]$ is the set of all combinatorial surfaces $S\subset K$ with boundary $B$.
\end{lemma}
\begin{proof}
    We must show inclusion both ways. Let $\Sigma\in\D[r]$. The bag $X_r=\emptyset$, and the complex $K_r = K$ as each node in the tree decomposition is a descendant of $K$. Therefore, $K_r\setminus X_r = K$ and the link of each vertex or edge in $K$ is complete. The definition of complete link implies that $\Sigma$ is a combinatorial surface. Moreover, it follows from the definition of complete link that each simplex $\sigma\in B$ must be on the boundary of $\Sigma$, and no simplex $\sigma\notin B$ is on the boundary of $\Sigma$. Therefore, the boundary of $\Sigma$ is exactly $B$, as claimed. 
    \par 
    Alternatively, let $S\subset K$ be a combinatorial surface with boundary $B$. It follows from the definitions of combinatorial surface and boundary that $S$ is a pure 2-dimensional complex and that each simplex in $K_r\setminus X_r = K$ has complete link. Therefore, $S\in\D[r]$.
\end{proof}

\subsection{Dynamic Program}
\label{sec:dynamic program}

We use a dynamic program on a nice, closed tree decomposition of the Hasse diagram of $K$ to compute the set of candidate solutions at a node $t$. In particular, we will compute the set of candidate solutions at $t$ using the candidate solutions at each of the children of $t$. Therefore, we assume that the set $\D[t']$ is computed for each child $t'$ of a node $t$ before we compute the set $\D[t]$. We now present a case analysis of how to compute $\D[t]$ for each type of node in a nice tree decomposition.

\subsubsection{Leaf Node}

Let $t$ be a leaf node of $T$. The subcomplex $K_t$ is empty, so there are no candidate solutions at $t$. Therefore, $\D[t] = \emptyset$.

\subsubsection{Introduce Nodes}

Let $t$ be an introduce node with child $t'$. The bag $X_t$ has one more simplex $\sigma$ than the bag of the child $X_{t'}$, but it is not immediate how the subcomplex $K_t$ differs from $K_{t'}$. In fact, whether $\sigma$ is contained in $K_{t}$ depends on the dimension of $\sigma$. We prove this is Lemma \ref{lem:introduce_complex}

\begin{lemma}
\label{lem:introduce_complex}
    Let $t$ be an introduce node with child $t'$. Let $\sigma$ be the introduced simplex. Then
    $$
        K_t = 
        \begin{cases}
            K_{t'}\cup\{\sigma\} & \text{if $\sigma$ is a vertex or edge} \\
            K_{t'} & \text{if $\sigma$ is a triangle} 
        \end{cases}
    $$
\end{lemma}
\begin{proof}
    The complexes $K_t$ and $K_{t'}$ are defined $K_t = (\cup_{d\in D(t)} X_d)\setminus(X_t)_2$ and $K_{t'} = (\cup_{d'\in D(t')} X_{d'})\setminus(X_{t'})_2$, where $D(t)$ and $D(t')$ are the set of descendants of $t$ and $t'$.
    \par 
    We first show that $(\cup_{d\in D(t)} X_d)\setminus\{\sigma\} = (\cup_{d'\in D(t')} X_{d'})$. We know that $(\cup_{d\in D(t)} X_d)=(\cup_{d'\in D(t')} X_{d'})\cup X_{t} = (\cup_{d'\in D(t')} X_{d'})\cup\{\sigma\}$ where the second equality follows from the fact that every simplex in $X_t$ but $\sigma$ is contained in $X_{t'}$. It only remains to be shown that $\sigma\notin \cup_{d'\in D(t')} X_{d'}$. Suppose that $\sigma \in X_{d'}$ for some descendant $d'$ of $t'$. This would imply that $\sigma\in X_{t'}$ as $t'$ lies on the unique path between $d'$ and $t$; however, we assume that $\sigma\notin X_{t}$. Therefore, $\sigma\notin\cup_{d'\in D(t')} X_{d'}$. 
    \par 
    We know that $(\cup_{d\in D(t)} X_d)\setminus\{\sigma\} = (\cup_{d'\in D(t')} X_{d'})$ no matter the dimension of $\sigma$. If $\sigma$ is a triangle, then $(X_t)_2\setminus\{\sigma\} = (X_{t'})_2$; otherwise, $(X_t)_2 = (X_{t'})_2$. Therefore, $K_{t}=K_{t'}$ if $\sigma$ is a triangle, and $K_{t}\setminus\{\sigma\}=K_{t'}$ if $\sigma$ is a vertex or edge.
\end{proof}

\paragraph*{Vertex and Edge Introduce Nodes}

\begin{lemma}
\label{lem:vertex_introduce_correctness}
    Let $t$ be a vertex or edge introduce node, and let $t'$ be its unique child. Then $\D[t] = \D[t']$.
\end{lemma}

\begin{proof}
    Let $\Sigma$ be a candidate solution at $t$. We will show that $\Sigma$ is a candidate solution at $t'$.
    \par 
    We first show that $\Sigma\subset K_{t'}$. The only simplex in $K_t\setminus K_{t'}$ is the introduced simplex $\sigma$ by Lemma \ref{lem:introduce_complex}, so we must show that $\sigma$ cannot be in $\Sigma$. As $S$ is a pure 2-dimensional complex, if $\sigma$ were in $\Sigma$, then $\sigma$ would have to be the face of a triangle in $K_t$. This triangle would have to be in $K_{t'}$, as the only simplex in $K_t\setminus K_{t'}$ is $\sigma$. However, no such triangle can exist, as this would imply $\sigma\in K_{t'}$ as $K_{t'}$ is a simplicial complex. Therefore, $\Sigma\subset K_{t'}$.
    \par 
    We now verify that $\Sigma$ satisfies the link conditions at $t'$. The bag at $t$ is $X_{t} = X_{t'}\cup\{\sigma\}$ and the complex $K_t = K_{t'}\cup\{\sigma\}$. Therefore, $X_{t'} \subset X_t$ and $K_{t'}\setminus X_{t'} = K_t\setminus X_t$. It follows that $\Sigma$ satisfies the link conditions at $t'$ as it satisfied the link conditions at $t$.
    \par 
    Now let $S'$ be a candidate solution at $t'$. We will show that $S'$ is also a candidate solution at $t$. We know that $S'$ is a subcomplex of $K_{t}$ as $K_{t'}\subset K_{t}$. We now verify that $S'$ satisfies the link conditions at $t$. We know that $\sigma\notin \Sigma'$ as $\sigma\notin K_{t'}$. Therefore, $\lk_{\Sigma}{\sigma}$ is admissible as it is empty. Any other simplex is in $X_{t'}$ iff it is in $X_t$. It follows that $\Sigma$ satisfies the link conditions at $t'$ as $\Sigma$ satisfied the link conditions at $t$.
\end{proof}

\paragraph*{Triangle Introduce Nodes}

\begin{lemma}
\label{lem:triangle_introduce_correctness}
    Let $t$ be a triangle introduce node, and let $t'$ be its unique child. Then $\D[t] = \D[t']$.
\end{lemma}
\begin{proof}
    As $K_t=K_{t'}$, then any 2-dimensional subcomplex of $K_t$ is a subcomplex of $K_{t'}$ and vice versa. Moreover, as $X_{t}$ and $X_{t'}$ differ by a triangle, then $X_t$ and $X_{t'}$ contain the same vertices and edges, and so do $K_t\setminus X_t$ and $K_{t'}\setminus X_{t'}$. Therefore, any subcomplex of $K_t$ that satisfies the link conditions at $t$ will also satisfy the link conditions at $t'$ and vice versa. These two facts prove $\D[t] = \D[t']$.
\end{proof}

\subsubsection{Forget Nodes}

Let $t$ be a forget node, and let $t'$ be its unique child. The bag $X_{t}$ has one fewer simplex $\sigma$ than the bag of its child $X_{t'}$; however, as we exclude triangles in the bag $X_{t}$ from the complex $K_{t}$, the complex $K_t$ will have one \textit{more} simplex than the complex at its child $K_{t'}$ if $\sigma$ is a triangle. Alternatively, if $\sigma$ is a vertex or edge, then $\sigma$ was already in the complex $K_{t'}$, so $K_{t} = K_{t'}$. We summarize this in Lemma \ref{lem:forget_complex}.

\begin{lemma}
\label{lem:forget_complex}
    Let $t$ be a forget node with child $t'$. Let $\sigma$ be the forgotten simplex. Then
    $$
        K_t = 
        \begin{cases}
            K_{t'} & \text{if $\sigma$ is a vertex or edge} \\
            K_{t'}\cup\{\sigma\} & \text{if $\sigma$ is a triangle} 
        \end{cases}
    $$
\end{lemma}
\begin{proof}
    The complexes $K_t$ and $K_{t'}$ are defined $K_t = (\cup_{d\in D(t)} X_d)\setminus(X_t)_2$ and $K_{t'} = (\cup_{d'\in D(t')} X_{d'})\setminus(X_{t'})_2$, where $D(t)$ and $D(t')$ are the set of descendants of $t$ and $t'$.
    \par 
    We first show that $(\cup_{d\in D(t)} X_d) = (\cup_{d'\in D(t')} X_{d'})$. This is true as $(\cup_{d\in D(t)} X_d)=(\cup_{d'\in D(t')} X_{d'})\cup X_{t} = (\cup_{d'\in D(t')} X_{d'})$, where the second equality follows from the fact that $X_t\subset X_{t'}$.
    \par
    The difference between $K_t$ and $K_{t'}$ will be determined by the sets of triangles $(X_t)_2$ and $(X_{t'})_2$. As $X_t\cup\{\sigma\} = X_{t'}$, then $(X_t)_2 = (X_{t'})_2\cup\{\sigma\}$ is $\sigma$ is a triangle and $(X_t)_2 = (X_{t'})_2$ otherwise.
\end{proof}

\paragraph*{Triangle Forget Nodes}

\begin{lemma}
    Let $t$ be a triangle forget node, and let $t'$ be its unique child. Let $\Delta$ be the triangle being forgotten. Then $\D[t] = \D[t'] \cup E^{\Delta}(t)$, where the set $E^{\Delta}(t)$ is defined as
    $$
    E^{\Delta}(t) = \left\{ \Sigma := \Sigma'\cup\cl(\Delta) \:\middle|
    \begin{array}{l}
        (1)\: \Sigma'\in D[t'] \\ 
        (2)\: \text{The links of all faces of $\Delta$ are admissible in $\Sigma$}  \\
    \end{array}  \right\}.
    $$ 
\end{lemma}
\begin{proof}
    Let $\Sigma$ be a candidate solution at $t$. There are two cases: $\Sigma$ does not contain $\Delta$, or $\Sigma$ contains $\Delta$.
    \par 
    In the case that $\Sigma$ does not contain $\Delta$, then we claim that $\Sigma$ is a candidate solution at $t'$. The only simplex in $K_t\setminus K_{t'}$ is $\Delta$, so $S\subset K_{t'}$. Moreover, $X_t$ and $X_{t'}$ contain the same set of vertices and edge, so $\Sigma$ satisfies the link conditions at $t'$ as it satisfied them at $t$. Therefore, $S$ is a candidate solution at $t'$.
    \par 
    In the case that $\Delta\in\Sigma$, we claim that $\Sigma$ belongs $E^\Delta(t)$. As $\Sigma$ is a candidate solution at $t$, then the links of all faces of $\Delta$ are admissible as they are all contained in $X_t$, so $\Sigma$ meets the second condition to belong to $\Sigma$. We now verify that $\Sigma$ is of the form $\Sigma'\cup\cl(\Delta)$ for some candidate solution $\Sigma'$ at $t'$. Let $\Sigma'$ be the simplicial complex obtained by removing $\Delta$ and any vertex or edge that is only incident to $\Delta$ from $S$; we will show that $\Sigma'$ is a candidate solution at $t$. The complex $S'$ is a pure 2-dimensional simplicial complex by construction. Also, the complex $S'\subset K_{t'}$ as the only simplex in $K_t\setminus K_{t'}$ is $\Delta$. Finally, we verify that $\Delta'$ satisfies the link conditions at $t'$. The set $X_{t'}$ and $X_{t}$ have the same set of vertices and edges, so the requirements on each vertex and edge are the same at $t$ and $t'$. The only vertices and edges with different links in $\Sigma$ and $\Sigma'$ are the faces of $\Delta$. However, the links of these simplices in $\Sigma'$ will be a subset of their links in $\Sigma$. The subset of a admissible link is still is admissible, so $S'$ satisfies the link conditions at $t'$.
    \par 
    We now show that $\D[t']\subset \D[t]$. Let $\Sigma'\in\D[t']$. The complex $\Sigma'$ is a subcomplex of $K_{t}$ as $K_{t'}\subset K_{t}$. As well, the vertices and edges in $K_t\cap X_t$ and $K_{t'}\setminus X_{t'}$ are the same, so the link conditions hold for $\Sigma$ at $t$ as they held for $\Sigma$ at ${t'}$. Therefore, $\Sigma$ is a candidate solution at $t$.
    \par 
    We now show that $E^{\Delta}(t)\subset \D[t]$. Let $\Sigma\in E^{\Delta}(t)$. We first show that $\Sigma$ actually is in $K_t$. Indeed, we know that $\Sigma'\subset K_{t'}$ by assumption, so $\Sigma'\subset K_{t}$ as $K_{t'}\subset K_{t}$. Moreover, $\cl(\Delta)\subset K_{t}$ as $\Delta\in K_{t}$. Therefore, $\Sigma\subset K_{t}$. We now verify that $\Sigma$ satisfies the link conditions at $t$. Adding $\Delta$ to $S'$ only changes the links of the face of $\Delta$, and as the links of these vertices are admissible by the assumptions on $E^{\Delta}(t)$, we see $\Sigma$ satisfies the link conditions at $t$.
\end{proof}

\paragraph*{Vertex and Edge Forget Nodes}

\begin{lemma}
\label{lem:ve_forget_node}
    Let $t$ be an edge or vertex forget node, and let $t'$ be its unique child. Let $\sigma$ be the vertex being forgotten. Then 
    $$
        \D[t] = \left\{ \Sigma'\in\D[t'] \;\middle|\; \text{$\lk_{\Sigma'}\sigma$ is complete} \\
        \right\}
    $$
\end{lemma}
\begin{proof}
    Let $\Sigma\in\D[t]$ be a candidate solution at $t$. We must show that $\Sigma$ is a candidate solution at $t'$. By Lemma \ref{lem:forget_complex}, we know that $K_t = K_{t'}$, so $\Sigma$ is a subcomplex of $K_{t'}$. Moreover, we know that $X_{t'} = X_{t}\cup{\sigma}$. The simplices in $X_{t}$ have admissible link as $\Sigma$ is a candidate solution at $t$. The simplex $\sigma$ has complete (and therefore admissible) link, so each simplex in $X_{t'}$ has admissible link. Likewise, each simplex in $K_{t'}\setminus X_{t'}$ is also in $K_t\setminus X_t$, so these simplices have complete link as $\Sigma$ is a candidate solution at $t$. Therefore, $\Sigma$ is a candidate solution at $t'$.
    \par 
    Now let $\Sigma'$ be a candidate solution at $t'$ such that $\sigma$ has complete link in $\Sigma'$. We will show that $\Sigma'$ is a candidate solution at $t$. By Lemma \ref{lem:forget_complex}, we know that $K_t = K_{t'}$, so $\Sigma'$ is a subcomplex of $K_{t'}$. We now verify that $\Sigma'$ satisfies the link conditions at $t$. The bag $X_{t} = X_{t'}\setminus\{\sigma\}$. The simplices in $X_{t'}$ have admissible link in $\Sigma'$, so each simplex in $X_t$ has admissible link in $\Sigma'$ as $X_t\subset X_{t'}$. Likewise, the simplices in $K_{t'}\setminus X_{t'}$ have complete link. As $\sigma$ has complete link, each simplex in $K_{t}\setminus X_t$ has complete link as $K_t = K_{t'}$ and $X_t = X_{t'}\setminus\{\sigma\}$. Therefore, $\Sigma'$ is a candidate solution at $t$.
\end{proof}

\subsubsection{Join Nodes}

\begin{lemma}
\label{lem:join_complex}
    Let $t$ be a join node, and let $t'$ and $t''$ be the children of $t$. Then $K_t = K_{t'}\cup K_{t''}$ and $K_{t'}\cap K_{t''}\subset X_t$
\end{lemma}
\begin{proof}
    The complex $K_s$ is defined $K_s= \cup_{d\in D(s)}X_s$ for $s=t,t',t''$, where $D(s)$ is the set of descendants of $s$.
    \par 
    We first show that the union $\cup_{d\in D(t)} X_d = (\cup_{d'\in D(t')} X_{d'})\cup(\cup_{d''\in D(t'')} X_{d''})$. Indeed, each descendant $d$ of $t$ except $t$ itself is either a descendant of $t'$ or $t''$, so $X_d$ is included in own of the unions on the right. As for $t$ itself, we know that $X_t = X_{t'}= X_{t''}$, so the simplices in $X_t$ are included in both unions on the right.
    \par 
    We now compare the complexes $K_t$ with $K_{t'}$ and $K_{t''}$. Each of the bags $X_t$, $X_{t'}$, and $X_{t''}$ have the same set of triangles as they are all equal, so 
    \begin{align*}
        K_t &= (\cup_{d\in D(t)} X_d) \setminus (X_t)_2 \\
        &= ((\cup_{d'\in D(t')} X_{d'})\cup(\cup_{d''\in D(t'')} X_{d''}))\setminus(X_{t})_2 \\
        &= (((\cup_{d'\in D(t')} X_{d'})\setminus(X_{t'})_2)\cup((\cup_{d''\in D(t'')} X_{d''})\setminus(X_{t''})_2) \\
        &= K_{t'} \cup K_{t''}. 
    \end{align*}
    We now prove that $K_{t'}\cap K_{t''} \subset X_{t}$. Let $\sigma\in K_{t'}\cap K_{t''}$. By definition, there are descendants $d'$ and $d''$ of $t'$ and $t''$ respectively such that $\sigma\in X_{d'}$ and $\sigma\in X_{d''}$. The node $t$ lies on the path connecting $d'$ and $d''$, so $\sigma\in X_{t}$.
\end{proof}

\begin{lemma}
    Let $t$ be a join node, and let $t'$ and $t''$ be the unique children of $t$. Then
    $$
        \D[t] = \left\{ S'\cup S'' \:\middle|\: 
        \begin{array}{l}
            (1)\: S'\in\D[t'],\, S''\in\D[t''] \\
            (2)\: \text{Each simplex in $(S'\cup S'')\cap X_t$ has admissible link.}  \\
        \end{array} \right\}. 
    $$
\end{lemma}
\begin{proof}
    Let $S$ be a candidate solution at $t$. We will show that $S$ is of the form $S'\cup S''$ for candidate solutions at $t'$ and $t''$ respectively. Define $S'$ to be the closure of all triangles in $S\cap K_{t'}$. The complex $S'$ is a pure 2-dimensional subcomplex of $K$ as each simplex is the face of a triangle. By Lemma \ref{lem:surface-coface}, the complex $S'$ is also a subset of $K_{t'}$, as each simplex of $S'$ is the face of a triangle in $K_{t'}$. 
    \par 
    We now verify that the $S'$ satisfies the link conditions at $t'$. We begin with the observation that $\lk_{S'}(\sigma) \subset \lk_{S}(\sigma)$ for each simplex $\sigma\in S'$ as $S'\subset S$. We now analyze the link of $\sigma$ in $S'$ in two cases. Let $\sigma\in S'\cap X_{t'}$ be a vertex or edge. As $X_t = X_{t'}$, then $\sigma\in S\cap X_{t}$, which means that the link of $\sigma$ in $S$ is admissible by assumption. As $S'\subset S$, then $\lk_{S'}(\sigma) \subset \lk_{S}(\sigma)$. A subset of a admissible link is still admissible, so the link of $\sigma$ in $S'$ is also admissible. Alternatively, let $\sigma\in (S'\setminus X_{t'})$ be a vertex or edge. Again, the fact that $X_{t'} = X_{t}$ implies that $\sigma\in (S\setminus X_{t})$, so the link of $\sigma$ in $S$ is complete. We claim that $\lk_{S}(\sigma) = \lk_{S'}(\sigma)$. Any vertex in $\lk_{S}(\sigma)$ is incident to a common triangle with $\sigma$, and by Lemma \ref{lem:surface-coface}, this triangle is contained in $K_{t'}\setminus X_{t'}$. By construction, this triangle and all of its faces are in $S'$, so $\lk_{S}(\sigma) = \lk_{S'}(\sigma)$.  
    \par
    Now let $S'$ and $S''$ be candidate solutions at $t'$ and $t''$ respectively such that each vertex in $S'\cup S''$ has admissible link in $S'$ and $S''$. We will show that $S'\cup S''$ is a candidate solution at $t$. The complex $S\cup S''$ is a pure subcomplex of $K_t$ as $S'$ and $S''$ are both pure subcomplexes of $K_t$. We now show that $S'\cup S''$ satisfies the link conditions at $t$. We assume that each simplex in $(S'\cup S'')\cap X_{t}$ has admissible link. So let $\sigma \in (S'\cup S'')\setminus X_t$ be an edge or vertex. As $\sigma\notin X_t$, we know by Lemma \ref{lem:join_complex} that $\sigma$ is in either $K_{t'}$ or $K_{t''}$ (and thus either $S'$ or $S''$). WLOG assume that $\sigma\in S'$. Moreover, we conclude that $\lk_{S}(\sigma) = \lk_{S'}(\sigma)$ as $\sigma\notin S''$ so $\lk_{S''}(\sigma)$ is empty. The link $\lk_{S'}(\sigma)$ is complete as $S'$ is a candidate solution, so the link $\lk_{S}(\sigma)$ is also complete. 
\end{proof}

\subsection{Storing Candidate Solutions As Cell Complexes}
Let $\mathbf{D[t]}$ be the set of candidate solutions at a node $t$. We will not store the set of candidate solutions $\mathbf{D[t]}$ explicitly, as the number of candidate solutions at a node can be exponential in the size of the complex, even for complexes with fixed treewidth. Instead, we store a set of annotated cell complexes equivalent to each candidate solution. Storing candidate solutions as cell complexes allows us to dramatically reduce the number of candidate solutions we need to store, as many candidate solutions that are different as simplicial complexes may be equivalent as cell complexes. See Figure \ref{fig:sc_vs_cc} for an example.
\par 
The set of annotated cell complexes $\mathbf{V[t]}$ at a node $t$ is computed as follows. For each candidate solution $\Sigma\in\mathbf{D[t]}$, we use equivalence preserving moves to remove all edges or vertices $\sigma\notin X_t$ from $\Sigma$, as described in Section \ref{sec:remove_simplices}. (As a cell complex is a collection of edges, removing a vertex $v$ means removing all edges incident $v$.) The cell complex obtained by removing simplices in $\Sigma\setminus X_t$ from $\Sigma$ using our algorithm is the \textit{\textbf{corresponding cell complex at t}} and is denoted $\Tilde\Sigma$. In particular, the entry at the root $\mathbf{V[r]}$ will contain annotated cell complexes with a single face that has no edges in its boundary and instead only stores topological information in its annotation. 
\par
The size of an annotated cell complex in $\Tilde{\Sigma}\in\mathbf{V[t]}$ is $\OO(k)$. Intuitively, this is because there are $\OO(k)$ edges in $X_t$, although we also add $\OO(k)$ dummy edges not in $X_t$. Dummy edges are used to store information on the link of vertices in $X_t$ after edges incident to these vertices are removed. We will see in Section \ref{sec:remove_simplices} when we add these dummy edges.
\par
After removing each of the simplices in $\Sigma\setminus X_t$, then many candidate solutions may be transformed into the same annotated cell complex. All candidate solutions that are transformed into the same annotated cell complex are homeomorphic as they are equivalent to a common cell complex, so we only keep one of these annotated cell complexes. 
\begin{figure}
    \centering
    \includegraphics[width=0.35\textwidth]{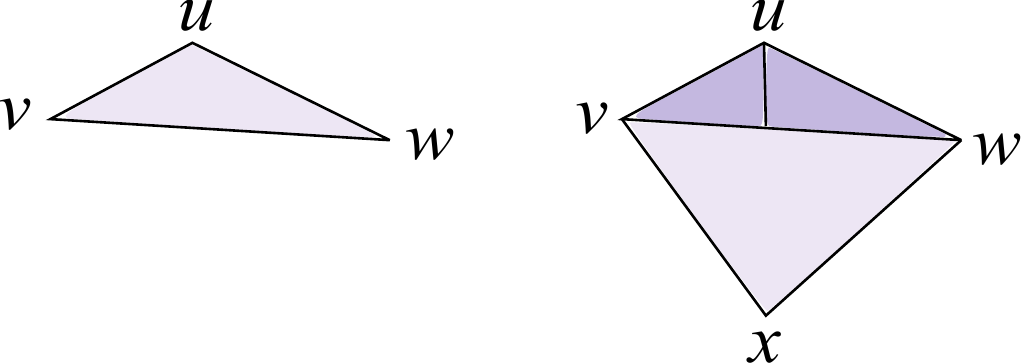}
    \caption{Multiple candidate solutions may reduce to the same cell complex after removing simplices $\sigma\notin X_t$. If the bag $X_t=\{u,v,w,\{u,v\},\{v,w\},\{w,v\}\}$, then both of the above candidate solutions reduce to the same cell complex at $t$.}
    \label{fig:sc_vs_cc}
\end{figure}
If we are going to store candidate solutions as cell complexes in this way, we need to verify two things. First, we are still able to verify that a cell complex is equivalent to a candidate solution, even after removing the simplices $\Sigma\setminus X_t$. Second, the set of equivalence-preserving moves are sufficient for removing all simplices in $\Sigma\setminus X_t$. We will explore these requirements in the next two sections respectively.

\subsubsection{Checking Candidacy on Cell Complexes}
\label{sec:check_candidacy}

We store a candidate solution $\Sigma$ at a node $t$ as a cell complex $\Tilde\Sigma$ obtained by removing all simplices in $\Sigma\setminus X_t$; however, to verify that $\Sigma$ is in fact a candidate solution using our dynamic program, we need to verify that the link of the simplices in $X_t$ satisfy the link conditions. This presents an apparent problem. Even if a simplex $\sigma$ is contained in $X_t$ and hasn't been removed from $\Tilde\Sigma$, the simplices in $\lk_{\Sigma}\sigma$ need not be contained in $X_t$. For this reason, it is not obvious that we can draw any conclusions about $\lk_{\Sigma}\sigma$ from $\Tilde\Sigma$.
\par 
We can in fact check the link conditions of a simplex $\Sigma$ using only the corresponding cell complex $\Tilde{\Sigma}$. In this section, we give an informal description of how we can do this. However, proving this formally is dependent on the algorithm we use to remove simplices in the next section, so the proofs of the lemmas in this section are deferred to Appendix \ref{sec:check_candidacy_proofs}.
\par 
We can check the link conditions on an edge $e\in X_t$ using the fact that the number of times an edge $e$ appears in the boundary of a face of $\Tilde\Sigma$ is the number of triangles incident to $e$ in $\Sigma$. Intuitively, an appearance of $e$ is only added to $\Tilde\Sigma$ when a triangle incident to $e$ is added to $\Sigma$, and no appearance of $e$ is removed from $\Tilde\Sigma$ until $e$ is forgotten from the bag $X_t$. 
\par 
To state this as a lemma, we need to know one of thing about our algorithm in advance. Namely, we will be adding \textit{\textbf{dummy edges}} to our cell complex. These are edges that do not correspond to edges in $K$ but appear in $\Tilde\Sigma$. A caveat is that a dummy edge may have the same endpoints as an edge in $K$. For the time being, all we need to know about dummy edges is that they are marked to distinguish them from real edges in $K$.
\par
We now state the conditions for an edge to have complete link. Analogous conditions for an edge to have admissible link can be found in Appendix \ref{sec:check_candidacy_edges_proofs}.

\begin{restatable}{lemma}{edgecompletelink}
\label{lem:edge_link_complete}
    Let $t$ be a node in the tree decomposition. Let $\Sigma$ be a candidate solution at $t$, and let $\Tilde\Sigma$ be the corresponding cell complex at $t$. Let $e\in X_t$ be an edge. The link of $e$ in $\Sigma$ is complete if and only if
    \begin{itemize}
        \item $e\in B$ and the real edge $e$ appears once in the boundary of a face in $\Tilde{\Sigma}$.  
        \item $e\notin B$ and the real edge $e$ appears zero or twice in the boundary of faces in $\Tilde{\Sigma}$.
    \end{itemize}
    Moreover, these conditions can be checked on $\Tilde\Sigma$ in $\OO(k)$ time.
\end{restatable}

To check the candidacy of a vertex $v$, we can deduce information on the link of the vertex $v$ in a candidate solution $\Sigma$ based on the set of edges that enter $v$ in the corresponding cell complex $\Tilde\Sigma$. In particular, a path in the link of $v$ in $\Sigma$ exactly corresponds to a sequence of edges in $\Tilde\Sigma$, and a cycle in the link of $v$ in $\Sigma$ exactly corresponds to a \textit{cyclic} sequence of successors in $\Tilde\Sigma$. As an example of why this is true, suppose the edges entering $v$ form a cyclic sequence of successors. If we remove one of these edges by merging the two incident faces, then the edges entering $v$ will still form a cyclic sequence of successors. This is a simple case, but we can prove that something analogous happens in all cases.
\par 
We now state the conditions for an edge to have complete link. Analogous conditions for an edge to have admissible link can be found in Appendix \ref{sec:check_candidacy_vertices_proofs}.

\begin{restatable}{lemma}{vertexcompletelink}
\label{lem:vertex_complete_link}
    Let $t$ be a node in the tree decomposition. Let $\Sigma$ be a candidate solution at $t$, and let $\Tilde\Sigma$ be the corresponding cell complex at $t$. Let $v\in X_t$ be a vertex. The link of $v$ in $\Sigma$ is complete if and only if
    \begin{enumerate}
        \item $v\notin B$ and either
        \begin{enumerate}[label=(\roman*)]
        \item no edges in $\Tilde{\Sigma}$ enter $v$, or
        \item the edges entering $v$ in $\Tilde{\Sigma}$ form a cyclic sequence of successors $(\overbar{a_1\ldots a_k})$; or
        \end{enumerate}
        \item $v\in B$ and the edges entering $v$ in $\Tilde{\Sigma}$ form a sequence of successors $(a_1\ldots a_k)$ such that $a_1$ and $a_k$ are either real edges in $B$ or are boundary dummy edges.
    \end{enumerate}
    Moreover, these conditions can be checked on $\Tilde\Sigma$ in $\poly(k)$ time.
\end{restatable}

\subsubsection{Removing Simplices}
\label{sec:remove_simplices}

Each time we forget a vertex or edge, we remove this simplex and all incident simplices from each of our candidate solutions. Removing a simplex may mean changing the faces or annotations of a cell complex or adding a dummy edge to the cell complex, but the new cell complex will always be equivalent to the old one. The following sections provide a case analysis of all the ways we might remove a simplex from a cell complex.

\paragraph*{Removing Edges}
\label{sec:remove_edges}

Let $t$ be an edge forget node that forgets the edge $e$, and let $t'$ be the unique child of $t$. Let $\Sigma$ be a candidate solution at $t$, and let $\Tilde\Sigma$ be the corresponding at $t'$ (not $t$.) To convert $\Tilde\Sigma$ to be a cell complex at $t$, we must remove $e$ from $\Tilde\Sigma$ if it appears in this cell complex.
\par 
Recall that when we forget an edge, by Lemma \ref{lem:ve_forget_node}, we know that $\lk_{\Sigma}{e}$ is complete. Lemma \ref{lem:edge_link_complete} tells us that $e$ will appear either once or twice in the boundary of faces in $\Tilde\Sigma$. 
\par 
 We now present a case analysis of all ways we can remove $e$ from our annotated cell complex. The cases will depend on the number of times $e$ appears in the boundary of a face, whether $e$ appears multiple times on the boundary of the same or different face, whether $e$ appears on the same of different boundary component of the same face, and whether it is the edge $e$ or its inverse $e^{-1}$ that appears on a given face.

\begin{enumerate}[font=\bfseries]

\item \textbf{Edge or Inverse on Different Faces.}
    If $a$ is on the boundary of two faces $(\overbar{Xa})+(\overbar{Ya})$, we can invert one of the faces $(\overbar{Ya})=(\overbar{a^{-1}Y^{-1}})$ and combine the faces $(\overbar{Xa})+(\overbar{a^{-1}Y^{-1}})=(\overbar{XY^{-1}})$ using move (2). If some other combination of $a$ or $a^{-1}$ appear on different faces, we can invert faces as necessary so that $a$ appears in one face and $a^{-1}$ appears in the other.
    \par 
    If either $(\overbar{Xa})$ or $(\overbar{a^{-1}Y^{-1}})$ is non-orientable, then $(\overbar{XY^{-1}})$ is non-orientable. If $(\overbar{Xa})$ had genus $g_1$ and $(\overbar{Ya})$ had genus $g_2$, then $(\overbar{XY^{-1}})$ has genus $g_1+g_2$. If one face is orientable and the other is non-orientable, then by Lemma \ref{lem:handle_crosscap_equivalence}, we double the genus of the orientable face before adding the two genuses. The number of boundary components of $(\overbar{XY^{-1}})$ is likewise the sum of the number of boundary components of $(\overbar{Xa})$ and $(\overbar{Ya})$. See Figure \ref{fig:merge}.

\item\textbf{Edge and Inverse Non-Consecutively on Same Boundary Component of Same Face.}
    If $a$ and $a^{-1}$ appear non-consecutively on the same boundary component of a face $(\overbar{XaYa^{-1}})$, we can break this boundary component into two boundary components $(\overbar{X})(\overbar{Y})$ using move (6). See Figure \ref{fig:boundaries}.

\item\textbf{Edge Twice on Same Boundary Component of Same Face.}
    If $a$ appears twice on the boundary component of the same face $(\overbar{XaYa})$, then by Lemma \ref{lem:crosscap_boundary}, this face is equivalent to $(\overbar{bbY^{-1}X})$. We can remove the substring $bb$, keep the face $(\overbar{Y^{-1}X})$, and update the face's surface information. The face is non-orientable. If the face was orientable before removing $bb$ and had genus $g$, there are $g$ handles on the face $f$. One handle is equivalent to two crosscaps in the presence of a crosscap by Lemma \ref{lem:handle_crosscap_equivalence}, so there are $2g+1$ crosscaps after removing $bb$. If the face was non-orientable before removing $bb$, there are now $g+1$ crosscaps.
    \begin{figure}[H]
        \centering
        \includegraphics[width=0.35\linewidth]{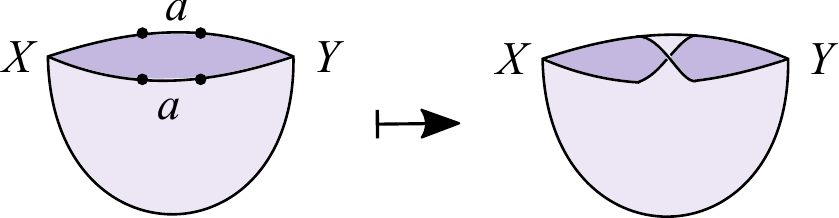}
        \captionsetup{margin=1cm}
        \caption{An example of Case 3. Identify edges $a$ and $a$ on the same boundary component creates a crosscap on the face. In particular, the faces is non-orientable after this identification.}
        \label{fig:crosscap}
    \end{figure}

\item\textbf{Edge and Inverse on Different Boundary Components of Same Face.}
    If $a$ and $a^{-1}$ appear on the boundary of the same face but on different boundary components $(\overbar{Xa})(\overbar{Ya^{-1}})$, then by Corollary \ref{cor:handle_boundary_equivalence}, we can combine these boundaries into a single boundary component $(\overbar{cdc^{-1}d^{-1}YX})$. We can remove the edges $c,d$ and annotate the face $(\overbar{YX})$ to have $+1$ genus if the face is orientable and $+2$ genus if the face is not orientable.
    \begin{figure}[H]
        \centering
        \includegraphics[width=0.35\linewidth]{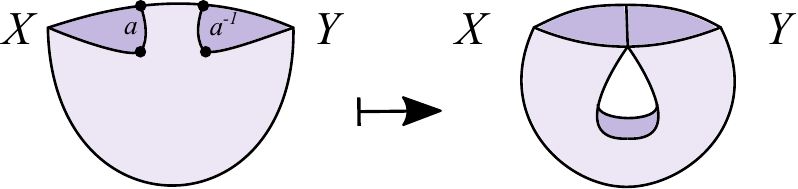}
        \captionsetup{margin=1cm}
        \caption{An example of Case 4. Identifying $a$ and $a^{-1}$ on different boundary components of the same face creates a handle.}
        \label{fig:handle}
    \end{figure}
    
\item\textbf{Edge Twice on Different Boundary Components of Same Face.}
    If $a$ appears twice on different boundary components of the same face $(\overbar{Xa})(\overbar{Ya})$, by Corollary \ref{cor:cell_crosshandle} this face is equivalent to $(\overbar{bbccXY^{-1}})$. We update the boundary of the face $(\overbar{XY^{-1}})$. This face is non-orientable. If this face was orientable with genus $g$ before removing $bbcc$, we update the genus to $2g+2$. If this face was non-orientable with genus $g$ before removing $bbcc$, we update the genus to $g+2$.
    \begin{figure}[H]
        \centering
        \includegraphics[width=0.35\linewidth]{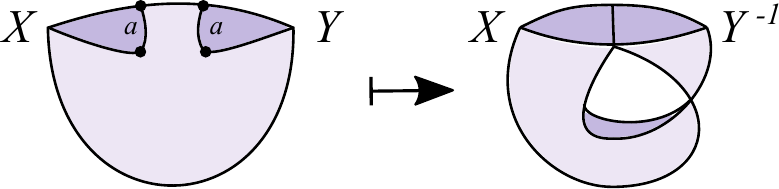}
        \captionsetup{margin=1cm}
        \caption{An example of Case 5. Identifying $a$ and $a$ on different boundary components of the same face creates two crosscaps. In the literature, this is also known as a \textit{\textbf{crosshandle}}.}
        \label{fig:crosshandle}
    \end{figure}
    
    \item\textbf{Edge and Inverse Consecutively on Same Boundary Component of Same Face.} If $a$ and $a^{-1}$ appear consecutively in some face, we would like to use move (5) to simplify $(\overbar{aa^{-1}X})$ to $(\overbar{X})$; however, we need to take an additional step to retain information on vertex links. Assume that $a=\{v,w\}$ for some vertices $v$ and $w$, and that $a$ enters $v$ (i.e. $a=(w,v)$). We use the edges entering $a$ to store the information on the link of $v$, so we need to take an additional step to remember this information about $v$.
\begin{figure}[H]
    \centering
        \begin{subfigure}{0.3\textwidth}
            \centering
            \includegraphics[width=0.65\linewidth]{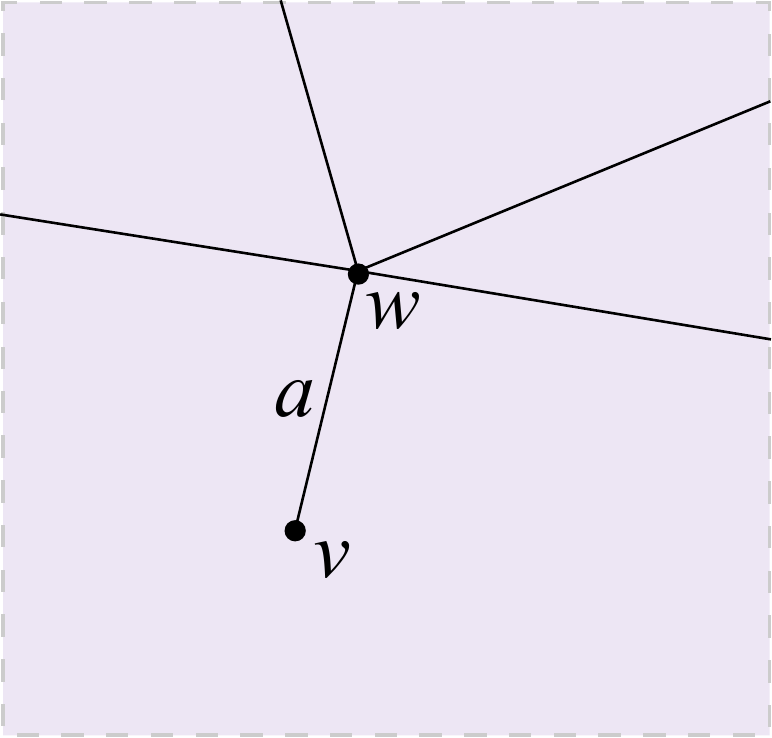}
        \end{subfigure}
        \begin{subfigure}{0.3\textwidth}
            \centering
            \includegraphics[width=0.65\linewidth]{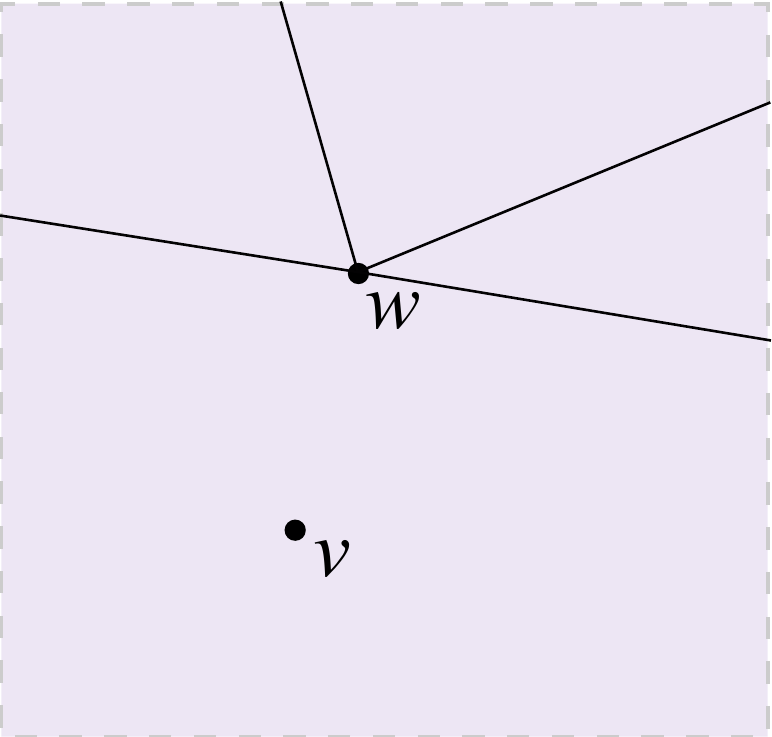}
        \end{subfigure}
        \begin{subfigure}{0.3\textwidth}
            \centering
            \includegraphics[width=0.65\linewidth]{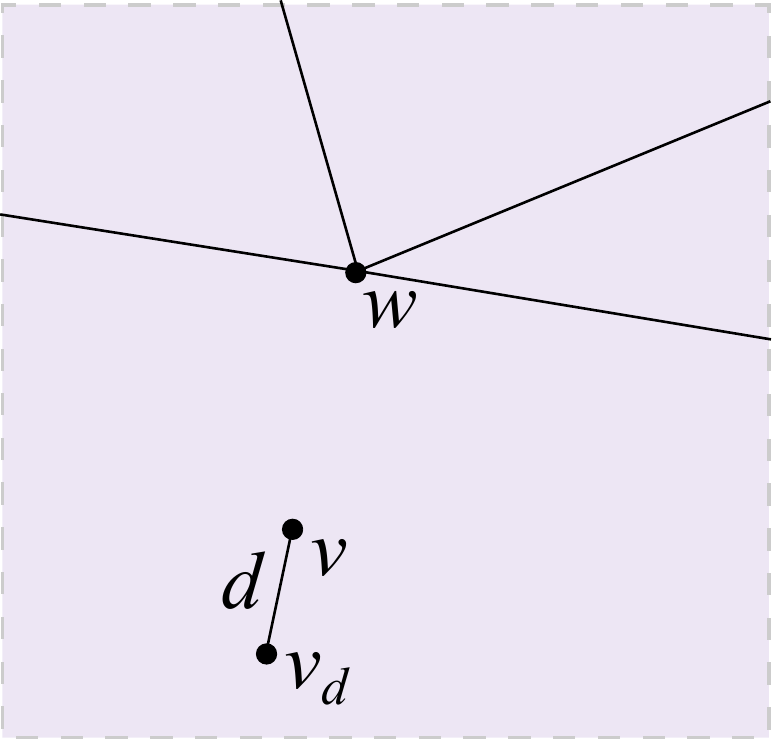}
        \end{subfigure}
    \captionsetup{margin=1cm}
    \caption{Forgetting the edge $a=(w,v)$ will remove any reference of $v$ from $\Tilde{\Sigma}$. We remember $v$ with a dummy edge $d$.}
    \label{fig:boundary_dummy_edge}
\end{figure}
    \par
    We will keep a record of $v$ with \textit{\textbf{interior dummy edges}}. Before removing $a$, we first add in dummy edges $(\overbar{add^{-1}a^{-1}X})$ using move (5). We can imagine these edges as connecting $v$ to a dummy vertex $v_d$, where the edge $d=(v,v_d)$. We then use move (6) to break the face into boundary components $(\overbar{dd^{-1}})(\overbar{X})$. The edges $dd^{-1}$ are added to $\Tilde{\Sigma}$ while maintaining equivalence, so adding the substring $dd^{-1}$ does not change the homeomorphism class of $\Tilde{\Sigma}$. A \textit{\textbf{real edge}} is an edge in $\Tilde{\Sigma}$ that is not a dummy edge. Real edges always correspond to edges in $K$. We assume dummy edges are marked.
    \par
    The sole purpose of the dummy edge $d$ is to denote that $v$ is still in the cell complex, even after all vertices that share an edge with $v$ have been forgotten. Note that before removing $a$, the edges incident to $v$ formed a cyclic sequence of successors $(\overbar{a})$ and $\lk_{\Sigma}{v}$ was complete. After adding $d$ and removing $a$ from the boundary, the edges incident to $v$ still form a cyclic sequence of successors, namely $(\overbar{d})$, so our algorithm will still recognize $\lk_{\Sigma}{v}$ as being complete.
    \par 
    Note that a vertex $v$ can only be incident to at most 2 interior dummy edges as $\lk_{\Sigma}{v}$ is complete. No more edges incident to $v$ will be added to $\Tilde{\Sigma}$ by our algorithm, as this would make $\lk_{\Sigma}{v}$ inadmissible. Therefore, there can only be $\OO(k)$ interior dummy edges in $\Tilde\Sigma$, and the number of edges (real or dummy)å in $\Tilde{\Sigma}$ is still $\OO(k)$.
    
 \item\textbf{Edge on Boundary of One Face.}
    Let $a$ be an edge that only appears once in the cell complex $\Tilde{\Sigma}$. By Lemma \ref{lem:edge_link_complete}, we know that $a$ is in $B$. We will replace $a$ with a \textit{\textbf{boundary dummy edge}} $d$, which is an edge with the same endpoints as $a$ but that is marked as being a dummy edge.  Note that a vertex $v$ can be incident to at most two boundary dummy edges as $v$ is incident to at most two edges in $B$. Therefore, there can only be $\OO(k)$ boundary dummy edges in $\Tilde\Sigma$.
\end{enumerate}   

\paragraph*{Removing Vertices}
\label{sec:remove_vertices}

When we forget a vertex $v$, we want to remove all edges incident to $v$ from cell complexes at $t'$. It turns out this is a relatively easy process. Any real edge incident to $v$ will already have been forgotten, so the only edges in $\Tilde\Sigma$ incident to $v$ are dummy edges. The follow lemma shows that the edges incident to $v$ in a cell complex $\Tilde\Sigma$ must be of one of two types. We give a proof of this lemma at the end of Appendix \ref{sec:check_candidacy_vertices_proofs}

\begin{restatable}{lemma}{vertexlinkcompleteforget}
\label{lem:vertex_link_complete_forget}
    Let $t$ be a vertex forget node that forget a vertex $v$, and let $t'$ be the unique child of $t$. Let $\Sigma$ be a candidate solution at $t'$, and let $\Tilde\Sigma$ be the corresponding cell complex at $t'$. Then the link of $v$ in $\Sigma$ is complete if and only if 
    \begin{enumerate}
    \item $v\notin B$ and either
    \begin{enumerate}[label=(\roman*)]
        \item no edges in $\Tilde\Sigma$ enter $v$, or
        \item a single interior dummy edge $d$ enters $v$ that forms a cyclic sequence of successors $(\overbar{d})$; or
    \end{enumerate}
    \item $v\in B$ and two boundary dummy edges $d_1$ and $d_2$ enter $v$ form a sequence of successors $(d_1,d_2)$.
\end{enumerate}
\end{restatable}

We now describe how to remove the edges incident to $v$ in the two cases in Lemma \ref{lem:vertex_link_complete_forget}.

\begin{enumerate}[font=\bfseries]
    \item$\mathbf{v\notin B}$ By Lemma \ref{lem:vertex_link_complete_forget}, there is a single interior dummy edge $d$ that enters $v$ that forms a cyclic sequence of successors. We conclude there must be a face of the form $(dd^{-1}X)$ in $\Tilde\Sigma$, so we can simplify this face to $(X)$ with move (5). 
    \item$\mathbf{v\in B}$ By Lemma \ref{lem:vertex_link_complete_forget}, there are two boundary dummy edges $d_1$ and $d_2$ that enter $v$ and form a sequence of successors. We distinguish between two subcases. 
    \begin{enumerate}
        \item$\mathbf{d_1 = d_2^{-1}}$ If $d_1 = d_2^{-1}$, then we conclude that $d_1$ must be the only edge in some boundary component $(\overbar{d_1})(\overbar{X})$. This face is equivalent to $(\overbar{d_1})(\overbar{X})=(\overbar{a^{-1}d_1aX})$ by Move (6). The string $b^{-1}ab$ is a boundary component, so we remove $b^{-1}ab$ and update the annotation of this face to have +1 boundary components.
        \item$\mathbf{d_1 \neq d_2^{-1}}$ If $d_1\neq d_2^{-1}$, then $d_1$ and $d_2^{-1}$ will be consecutive on the same face $(d_1d_2^{-1}X)$. We can then use move (1) and replace $d_1d_2^{-1}$ with a dummy edge $d_3$. In particular, if $d_1=(w_1,v)$ and $d_2=(v,w_2)$, then $d_3=(w_1,w_2)$. See Figure \ref{fig:interior_dummy_edge}. We define a \textit{\textbf{merge boundary dummy edge}} to be an edge that replaces two boundary dummy edges, although we usually just call them boundary dummy edges unless it is necessary to specify.
        \par 
        After replacing $d_1$ and $d_2$ with the edge $d_3$, the vertices $w_1$ and $w_2$ will still be incident to at most two boundary dummy edges. Therefore, there can only be $\OO(k)$ boundary dummy edges in $\Tilde\Sigma$. Moreover, the boundary dummy edges incident to a vertex $w\in B\cap X_t$ are the same in any candidate solution. Specifically, these dummy  connect $w$ to its the closest vertices in $B$ that have not yet been forgotten. 
        \begin{figure}[H]
        \centering
            \begin{subfigure}{0.49\textwidth}
                \centering
                \includegraphics[width=0.4\linewidth]{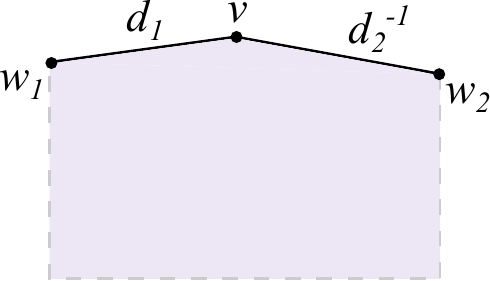}
            \end{subfigure}
            \begin{subfigure}{0.49\textwidth}
                \centering
                \includegraphics[width=0.4\linewidth]{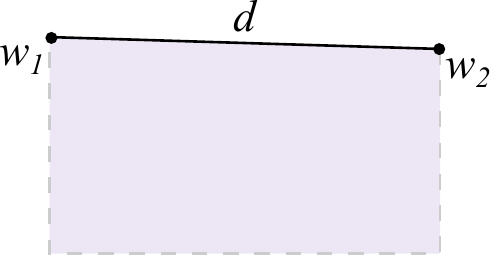}
            \end{subfigure}
            \captionsetup{margin=2cm}
            \caption{Replacing the edges $a_1$ and $a_{k}^{-1}$ with a boundary dummy edge $d$.}
        \label{fig:interior_dummy_edge}
        \end{figure}
    \end{enumerate}
\end{enumerate}

\subsection{Running Time Analysis}

In this section, we give a running time analysis of for each of our algorithms. We begin by giving a complete analysis of the algorithm for Subsurface Recognition, then in subsequent sections, we explain how the algorithm or analysis can be modified to find running times for algorithms for our other problems. 

\subsubsection{Subsurface Recognition}

The first step in analyzing the running time of our algorithm bound the number of cell complexes at each node in our tree decomposition. 

\begin{lemma}
    Let $t$ be a node in the tree decomposition, and let $g$ and $b$ be natural numbers. There are $2^{\OO(k\log k)}(gb)^{\OO(k+c)}$ cell complexes in $\V[t]$ of genus at most $g$ and with at most $b$ boundary components and $c$ connected components. 
\end{lemma}
\begin{proof}
    Assume for the time being that each face in a cell complex contains at least one edge. There are $\OO(k)$ edges, real or dummy, in any cell complex $\Tilde\Sigma\in\V[t]$. An annotated cell complex can be described as a bijection between a subset of these edges, where each edge is mapped to its successors, and the orbits of the bijection are the boundaries of the cell complex. There are $2^{\OO(k)}$ such subset of edges, and $\OO(k!)$ bijections for each subset. Next, we need to group boundary components into faces. There are $\OO(k)$ faces in any cell complex, so there are $k^{\OO(k)} = 2^{\OO(k\log k)}$ ways to partition these boundary components into faces. Multiplying the number of subsets by the number of bijections by the number of partitions, we see that are $2^{\OO(k)}\cdot \OO(k!)\cdot 2^{\OO(k\log k)} = 2^{\OO(k\log k)}$ cell complexes in $\V[t]$, not accounting for different annotations.
    \par 
    As we are capping the genus of the cell complexes at $g$, we can discard any solution with a face that exceeds genus $g$. Likewise, we discard any annotated cell complex with a face with more than $b$ boundary components. Therefore, each of the $\OO(k)$ faces of $\Tilde\Sigma$ can have one of $\OO(gb)$ annotations, so there are $(gb)^{\OO(k)}$ annotations for $\Tilde\Sigma$. Therefore, there are $2^{\OO(k\log k)}(gb)^{\OO(k)}$ possible different cell complexes in $\V[t]$.
    \par 
    A cell complex can also have connected components that contain no edges in $X_t$; these connected components are represented by empty faces in the cell complex. We can discard any cell complex that has more than $c$ such empty faces. Each of these empty faces will have an annotation, so there are $gb^{\OO(c)}$ possible annotations on these empty faces. In total, there are $2^{\OO(k\log k)}(gb)^{\OO(k+c)}$ possible annotated cell complexes in $\V[t]$.
\end{proof}

We claim there are $\OO(n^2)$ possible annotations for each face. The following lemma proves this. 

\begin{lemma}
    Let $S$ be a connected combinatorial surface. Let $n$ be the number of simplices in $S$, $g$ the genus of $S$, and $b$ the number of boundary components of $S$. Then $g,b\in \OO(n)$.
\end{lemma}
\begin{proof}
    Let $V$, $E$, and $F$ be the number of vertices, edges, and triangles of $S$ respectively. The Euler characteristic of $S$ is $V-E+F=2-2g-b$ if $S$ is orientable and $V-E+F=2-g-b$ if $S$ is non-orientable. The bound of  $g+b\leq E-V-F+2\in O(n)$ holds in either case.
\end{proof}

\begin{theorem}
    Let $K$ be a 2-dimensional simplicial complex with treewidth $k$ Hasse diagram. Let $X$ be a compact surface of genus $g$ with $c$ connected components. Let $B \subset K$ be a disjoint union of $b$ simple cycles. There is an algorithm to determine if there is a subcomplex $S\subset K$ homeomorphic to $X$ with boundary $B$ in $2^{\OO(k\log k)}(gb)^{\OO(k+c)}$ time.
\end{theorem}
\begin{proof}
    Let $r$ be the root of the nice tree decomposition $(T,X)$. By Lemma \ref{lem:dp_root}, we know that the set of candidate solutions at the root $\D[r]$ are equivalent to each subcomplex of $K$ that is a combinatorial surface with boundary $B$.  We first show that a set of cell complexes $\V[r]$ equivalent to the candidate solutions in $\D[r]$ with at most $c$ connected components and $b$ boundary components and with genus at most $g$ can be computed in the $2^{\OO(k\log k)}(gb)^{\OO(k+c)}$ using the dynamic program in Section \ref{sec:dynamic program}; if there is a subcomplex homeomorphic to $X$, it will be contained in this set $\V[r]$. There are $\OO(kn)$ nodes in a nice tree decomposition, so we just need to verify that the set $\V[t]$ can be computed for each node $t$ in $2^{\OO(k\log k)}(gb)^{\OO(k+c)}$ time.
    \par 
    Leaf nodes and introduce nodes both can be processed in $\OO(1)$ constant time. Leaf nodes require no work as $\V[t]$ is empty, and introduce nodes require no work as $\V[t] = \V[t']$ where $t'$ is the child of $t$.
    \par 
    Forget nodes can be processed in $2^{\OO(k\log k)}(gb)^{\OO(k+c)}$ time. Let $t$ be a forget node with child $t'$. For a vertex or edge forget node, we just need to check whether or not the link of the forgotten vertex is complete in each candidate solution in $\V[t']$. Checking one candidate solution takes $\poly(k)$ time according to Lemmas \ref{lem:vertex_complete_link} and \ref{lem:edge_link_complete}, so checking all candidate solutions takes $2^{\OO(k\log k)}(gb)^{\OO(k+c)}n$ time as there are $2^{\OO(k\log k)}(gb)^{\OO(k+c)}$ candidate solution in $\V[t']$. For a triangle forget node, we just need to compute the set $E^\Delta(t)$. Computing a single entry of $E^\Delta(t)$ takes $\poly(k)$ time, as we need to verify the links of a constant number of simplices. Computing each entry of $E^\Delta(t)$ takes $2^{\OO(k\log k)}(gb)^{\OO(k+c)}$ time as there are $2^{\OO(k\log k)}(gb)^{\OO(k+c)}$ entries of $\V[t']$.
    \par 
    Join nodes can also be processed in $2^{\OO(k\log k)}(gb)^{\OO(k+c)}$ time. Let $t$ be a join node with children $t'$ and $t''$. Each entry in $\V[t]$ is the sum of an entry from $\V[t']$ and $\V[t'']$. Therefore, to compute $\V[t]$, we perform a nested iteration over $\V[t']$ and $\V[t'']$, sum an entry from each, and check that the link of all simplices in $X_t$ is admissible.
    \par 
     Any cell complex $\Tilde\Sigma$ is a collection of empty faces, as any simplex in $K\setminus X_r = K$ has been removed using equivalence-preserving moves as described in Section \ref{sec:remove_simplices}. Therefore, we can determine if there is a combinatorial surface homeomorphic to $X$ by checking if there is cell complex in $\V[r]$ with the correct number of connected components, each having the correct genus. There are $(gb)^{O(c)}$ cell complexes in $\V[r]$, so this takes $(gb)^{O(c)}\poly(c)$ time.
\end{proof}

The previous theorem gives a parameterized version of the algorithm for SR. If we plug in the upper bound of $n$ for $g$, $b$, and $c$, then we get a running time of $2^{\OO(k\log k)}n^{\OO(k+n)}$. However, we can perform a tighter worst-case analysis. There are at most $2^{O(n)}$ possible candidate solutions corresponding to the $2^{\OO(n)}$ subsets of triangles of $K$. Our algorithm for SR therefore takes $2^{O(n)}$ time. 
\par 
However, this bound does not depend on the treewidth at all! Indeed, we have no FPT algorithm for SR, and our algorithm matches the complexity of the naive algorithm of testing all possible subset of triangles. Our algorithm relies on having a ``local representation'' of a candidate solution in term of the bag $X_t$, but defining a local representation becomes difficult when some connected components of our candidate solution do not even intersect the bag $X_t$. 

\subsubsection{Sum-of-Genus Subsurface Recognition}

The \SoGSR problem is distinct from the Subsurface Recognition as we don't care how the genus of a disconnected surface is distributed among its connected components. Dropping this requirement allows us to obtain an FPT algorithm. In our algorithm for SR, we obtained a factor of $gb^{\OO(k+c)} = n^{\OO(k+c)}$ as we store an annotation for each of the $k+c$ faces. We could instead store a single \textit{\textbf{global annotation}} for the entire cell complex. Whenever a topological feature is found on any face, this feature is recorded in the global annotation. Storing a global annotation is less discriminative than storing an annotation for each face. For example, a global annotation could not distinguish two tori from a genus 2 surface and a sphere, as both are genus 2 surfaces with 2 connected components. However, a global annotation is sufficient for \SoGSRdot. 

\begin{theorem}
    Let $K$ be a 2-dimensional simplicial complex with treewidth $k$ Hasse diagram. Let $B \subset K$ be a disjoint union of $b$ simple cycles. Let $g$ and $c$ be natural numbers. There is an algorithm to determine if there is a subcomplex $S\subset K$ with boundary $B$, total genus $g$, and $c$ connected components in $2^{\OO(k\log k)}gbcn=2^{\OO(k\log k)}n^{4}$ time.
\end{theorem}

\subsubsection{Connected Subsurface Recognition}

The \CSR problem is actually a special case of \SoGSR where the number of connected components $c=1$.

\begin{theorem}
    Let $K$ be a 2-dimensional simplicial complex with treewidth $k$ Hasse diagram. Let $X$ be a connected, compact surface of genus $g$. Let $B \subset K$ be a disjoint union of $b$ simple cycles. There is an algorithm to determine if there is a subcomplex $S\subset K$ homeomorphic to $X$ with boundary $B$ in $2^{\OO(k\log k)}gbn=2^{\OO(k\log k)}n^{3}$ time.
\end{theorem}

\subsubsection{Subsurface Packing}

In the \SP problem, we can discard annotations entirely, as we do not care about the genus of our surface. This saves a factor of $gb$ in our running time.

\begin{theorem}
    Let $K$ be a 2-dimensional simplicial complex with treewidth $k$ Hasse diagram. Let $B \subset K$ be a disjoint union of $b$ simple cycles. Let $c$ be a positive integer. There is an algorithm to determine if there is a subcomplex $S\subset K$ with $c$ connected components in $2^{\OO(k\log k)}cn=2^{\OO(k\log k)}n^{2}$ time.
\end{theorem}

\subsubsection{A Note on Boundaries}

For each of our problems, we assume that the boundary $B$ is given. We could instead ask to find a surface with a given \textit{number} of boundary components. Our algorithm can be adapted to handle this problem, but at the cost of making the running time slightly worse. For this problem, any edge in our complex could potentially be on the boundary of the surface. For our algorithm, this means that any edge could be replaced by a boundary dummy edge. Therefore, a cell complex at a node $t$ will contain a subset of $O(k^2)$ possible edges: the edges in $X_t$, and boundary dummy edges between any pair of vertices in $X_t$. However, we can now make it a requirement that a vertex in a candidate solution is incident to at most two boundary dummy edges. In our analysis, this means that a single cell complex a node $t$ will still have $O(k)$ edges, but these edges come from a set of $O(k^2)$ possible edges.  Therefore, the running time of our algorithm for c-SR, SoG, and SP becomes $k^{O(k^2)}n^{O(1)} = 2^{O(k^2\log k)}n^{O(1)}.$


\section{Lower Bounds} \label{sec:reduction_all}

This section is about proving lower bounds on the runtime for some of the problems we have studied in this paper. In particular, for both \SoGSR and \SP we prove that the algorithms from \cref{sec: algorithms} are optimal under the ETH. This means that we have essentially pinned down the computational complexity of these two problems. We have also proved a lower bound for the \SR problem, but this lower bound does not match the runtime of any known algorithm.

Let $k$ be the width of a given (nice) path decomposition of the Hasse diagram of the simplicial complex given as input. This section centers around proving the following theorem.

\begin{theorem}\label{thm:lower_bound}
Assuming the ETH, no algorithm can solve \SRdot, \SoGSR or \SP in $2^{\oo(k\log k)}n^{\OO(1)}$ time. The parameter $k$ denotes the width of a given (nice) path decomposition of the Hasse diagram of the input simplicial complex.
\end{theorem} 

Since every path decomposition is also a tree decompositions, the treewidth of a graph is never higher than its pathwidth. \cref{thm:lower_bound} therefore implies that none of our problems can be solved in $2^{\oo(k\log k)}n^{\OO(1)}$ time, where $k$ is now the treewidth of the Hasse diagram.
\par
We focus on proving the result for \SRdot. After this, it will be easy to modify our arguments to prove similar results for the two other problems. 
At a conceptual level there are two parts to the proof.

\begin{enumerate}
    \item Define a reduction from \DCP to \SRdot.
    \item Show that the reduction can always be chosen so that the pathwidth of the Hasse diagram of the output space is bounded by some linear function of the pathwidth of the input graph.
\end{enumerate}
\subsection{Directed Cycle Packing}
\label{sec:directed cycle packing}
\DCP asks us to find as many vertex disjoint cycles in a directed graph as possible (see \cref{fig:disjoint_cycles}). This is essentially a directed, $1$-dimensional version of the SP problem, as the only compact 1-manifolds are circles (cycles) and closed intervals (paths).


\begin{problem}The \DCP (DCP) problem\\
INPUT: A directed graph $D$ on $n$ vertices and an integer $\ell$.\\ 
PARAMETER: The pathwidth $k$ of $D$.\\
QUESTION: Does $D$ contain $\ell$ vertex disjoint cycles?
\end{problem}

\begin{figure}[!ht]
    \centering
    \includegraphics[width = \textwidth]{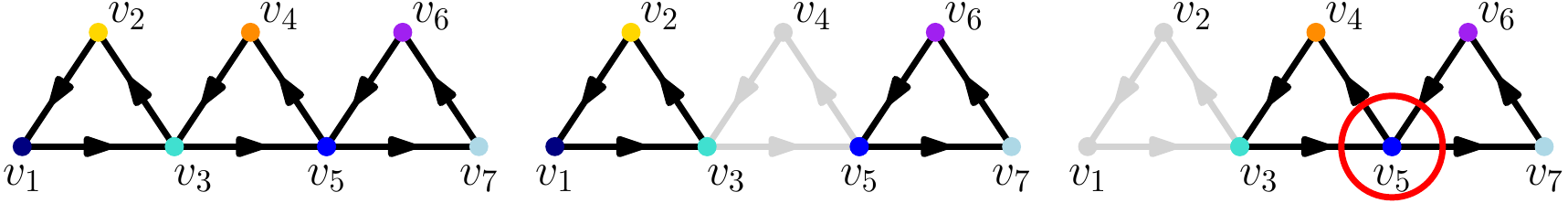}
    \caption{A directed graph $D$ (left), two vertex disjoint cycles contained in $D$ (middle) and two cycles in $D$ intersecting at a common vertex (right). This will be a guiding example for this section. 
    }
    \label{fig:disjoint_cycles}
\end{figure}

The DCP problem is a good starting point for our reduction not only because of its similarity to the SP problem but also because of the following theorem. 

\begin{theorem}[\cite{cycle_packing_paper}]\label{thm:ETH-cycles}
Assuming the ETH, the DCP problem cannot be solved in $2^{\oo(k\log k)}n^{\OO(1)}$ time, where the parameter $k$ denotes the width of a given (nice) path decomposition of the input graph.
\end{theorem}

Given a digraph $D$, the reduction will construct a 2-dimensional simplicial complex $Y$ that contains $\ell$ disjoint tori if and only if $D$ contains $\ell$ vertex disjoint cycles. In fact, we show that the only connected subsurfaces without boundary in $Y$ are tori and that these are in a bijection with the directed cycles in $D$. Furthermore, any pair of these tori are \textit{disjoint} if and only if the corresponding directed cycles are vertex disjoint. 

\subsection{Important Shorthand Notation}

\cref{fig:notation} introduces some important shorthand notation that will help make the reduction easier to follow. Each column of the figure shows a different component that we will use when constructing the space $Y$. 
\par
The first row shows the symbol we use for the space. The second row shows the shorthand notation that we will use. The third row shows the ``topological space'' the notation represents. The fourth row indicates how we triangulate the space. Finally, the last row shows clearly and in detail the triangulation of each space. Here we have made ``cuts'' in the spaces so that they could be flattened down onto the plane. The identifications that undo these cuts are indicated by the use of differently colored arrows. 

\begin{figure}[!ht]
    \centering
    \includegraphics[width = \textwidth]{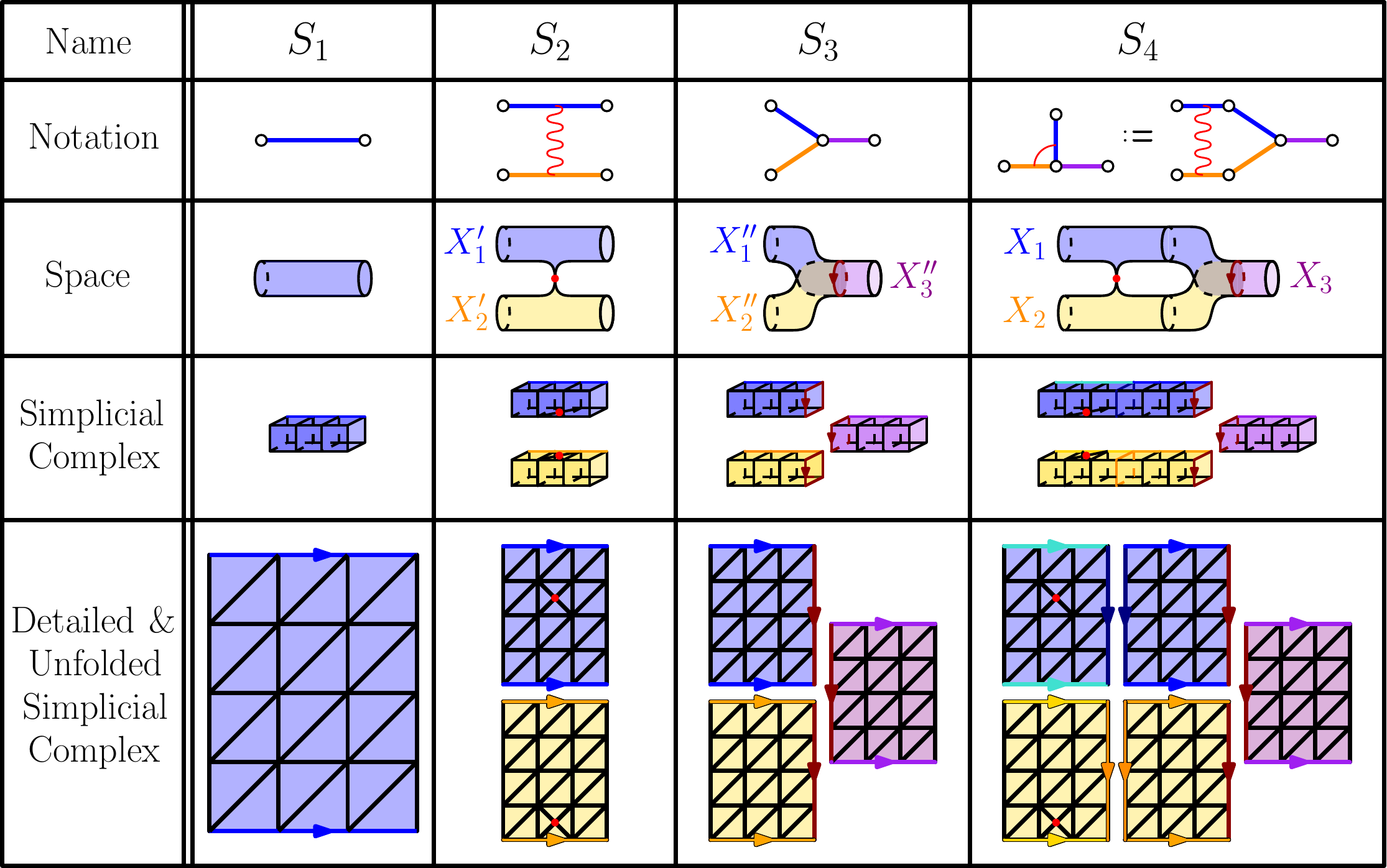}
    \caption{\label{fig:notation} Shorthand notation for specific triangulations of $S_1, \dots, S_4$ that we will use frequently throughout the section. 
    }
\end{figure}

The first column shows a cylinder, $S_1$. The second column shows a space $S_2$ consisting of two cylinders, $X_1'$ and $X_2'$. These cylinders are glued together at a single interior point, called a (0-dimensional) \textit{\textbf{singularity}}. The third column shows a space $S_3$ consisting of three cylinders $X_1'', X_2''$ and $X_3''$, each with a single boundary component attached to the same circle. The fourth and final column shows the space $S_4$, obtained by gluing $S_2$ and $S_3$ together. More precisely, $S_4$ also consists of three cylinders, $X_1 = X_1'\cup X_1'', X_2= X_2' \cup X_2''$ and $X_3 = X_3''$, each having a single boundary component attached to the same circle. Additionally, $X_1\cup X_2$ contains a 0-dimensional singularity.
\par
We establish some important properties of the spaces $S_1,S_2, S_3$ and $ S_4$ from \cref{fig:notation}. 
In order to describe these properties we temporarily extend the notion of a boundary, a term usually reserved for manifolds, to the world of simplicial complexes. For the remainder of the section, we use the below definition of the boundary of a simplicial complex. 

\begin{definition}
The \textbf{boundary of a simplicial complex} $K$ is the closure of the set of $1$-simplices in $K$ that only have a single coface. We denote the boundary as $$\boundary(K) = \cl{\{\rho \in K_1 | \#\{ \sigma | \rho \subset \sigma  \} = 1 \}}.$$
\end{definition}


The space $Y$ we will construct in the reduction is made by gluing together multiple copies of $S_1,S_2,S_3$ and $S_4$. It will therefore be important to know how a manifold contained in $Y$ can intersect these smaller components. In particular, we would like to know which manifolds $X$ are contained in each of $S_1, S_2, S_3,$ and $S_4$ respectively, where $\boundary(X) \subset \boundary(S_i)$. The following remark answers this question.

\begin{remark}
\label{remark:cylinder_boundaries}
Let $S_1,S_2,S_3$ and $S_4$ be the spaces introduced in \cref{fig:notation}.
\begin{enumerate}
    \item The only (non-empty) 2-manifold $X\subseteq S_1$ where $\boundary(X) \subseteq \boundary(S_1)$ is $S_1$ itself. 
    \item The only 2-manifolds $X\subseteq S_2$ where $\boundary(X) \subseteq \boundary(S_2)$ are $X_1'$ and $X_2'$. 
    \item The only 2-manifolds $X\subseteq S_3$ where $\boundary(X) \subseteq \boundary(S_3)$ are $X_1''\cup X_2''$, $X_1''\cup X_3''$ and $X_2''\cup X_3''$. 
    \item The only 2-manifolds $X\subseteq S_4$ where $\boundary(X)\subseteq \boundary(S_4)$ are $X_1 \cup X_3$ and $X_2 \cup X_3$.
\end{enumerate}
\end{remark}

\begin{proof}
Each of the four statements are intuitively obviously true. Formally, they can be proved easily by brute force: Simply go through all the 2-simplices in $S_i$ and assume that it is contained in a submanifold $X$. It is then easy to see which adjacent $2$-simplices must necessarily also be contained in the same submanifold. 
Whenever there is a choice to be made, simply branch and try all possibilities. 
\end{proof}

\subsection{Main Ideas of the Reduction}
\label{sec:reduction_outline}
This section gives an informal description of the simple idea behind the more technical reduction presented in \cref{sec:triangulation}.

\subsubsection{Cycles to Tori}
\label{sec:reduction_idea}
The reduction is best understood in terms of vertex gadgets and edge gadgets. In particular, \cref{fig:local_reduction} shows how a vertex $\xi$ is mapped to the vertex gadget $Y^{\xi}$, using the notation from \cref{fig:notation}. The figure also shows six edge gadgets (in black), three corresponding to the edges entering $\xi$ and three corresponding to the edges leaving $\xi$. The edge gadgets are unlabeled in the figure but can be identified by the vertex gadgets they are attached to. We think of each vertex gadget as composed of two sub-cylinders, one half for the incoming edge gadgets and the other half for outgoing edge gadgets. To better see this separation we draw the vertex gadget with a U-turn at the location of this divide in our figures.

\begin{figure}[!ht] 
    \centering
    \includegraphics[width = \textwidth]{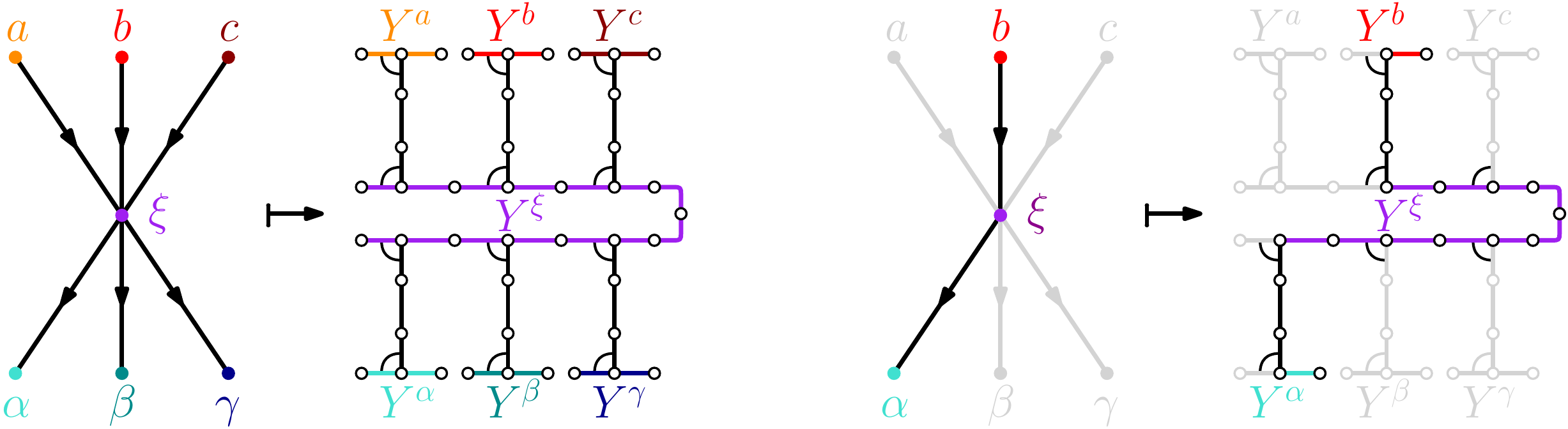}
    \caption{A local view of how a vertex $\xi$ is mapped to its vertex gadget $Y^\xi$ (left) and an illustration of how a directed cycle passing through the vertex $\xi$ is mapped to a submanifold in the space (right).
    }
    \label{fig:local_reduction}
\end{figure}
Each edge gadget is connected to the vertex gadgets corresponding to each of its two ends through a copy of $S_4$. The edge gadget contains the cylinder $X_1$ while the vertex gadget contains the other cylinders $X_2$ and $X_3$. Both the incoming and outgoing part of the vertex gadget consists primarily of a sequence of smaller cylinders, $X_2\cup X_3$, one for each incoming/outgoing edge. The boundary of the $X_3$ corresponding to one edge is attached to the boundary of the copy of $X_2$ corresponding to the next edge. The boundary of the ``last'' $X_3$ of the incoming edges is attached to one boundary component of a single additional cylinder, while the ``last'' $X_3$ of the outgoing edges is attached to the other boundary component of this same additional cylinder.

\begin{remark}
The order in which edge gadgets are attached to a vertex gadget is currently chosen arbitrarily. This is problematic, as we will see in \cref{sec:potential_pathwidth_explosion}.
\end{remark}

\begin{figure}[!ht]
    \centering
    \includegraphics[width = \textwidth]{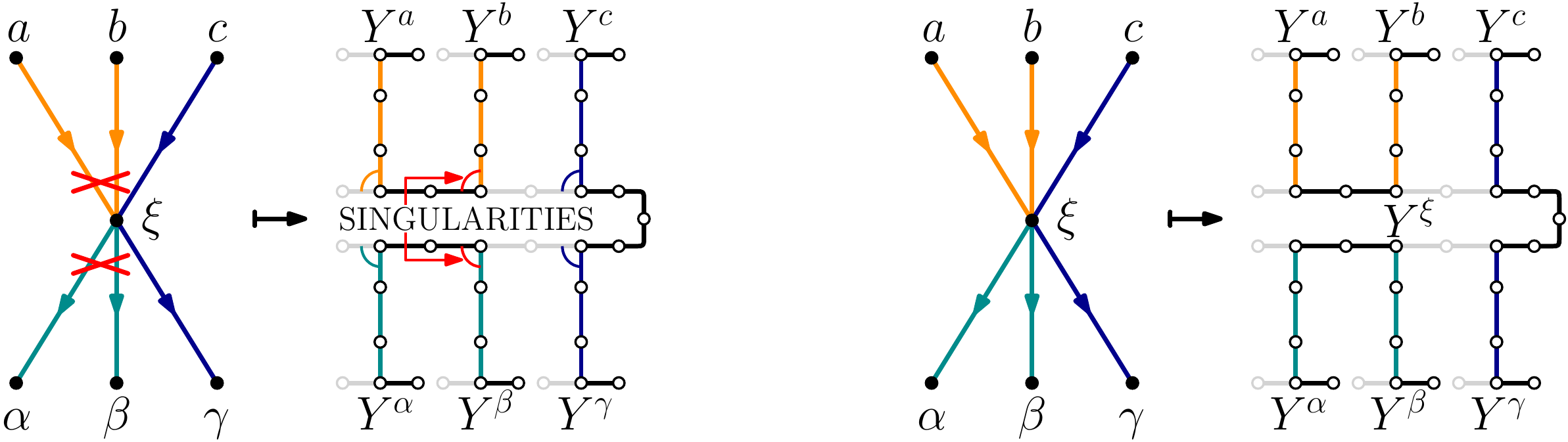}
    \caption{The leftmost figure shows how the singularities keeps ``badly behaved'' subcomplexes from becoming manifolds. The rightmost figure shows how the reduction would fail without the use of singularities between the vertex gadgets and edge gadgets.}
    \label{fig:singularities_reduction}
\end{figure}

By repeated use of Property 4 of Remark \ref{remark:cylinder_boundaries}, we can prove any potential manifold contained in this space must contain precisely one incoming and one outgoing edge gadget per vertex, assuming the manifold is not allowed to have a boundary. This is illustrated in  \cref{fig:singularities_reduction}. This figure also shows the importance of the 0-dimensional singularities in the reduction. The resulting space could otherwise contain tori that do not correspond to any directed cycle. An example of the correspondence between disjoint tori and vertex disjoint directed cycles is shown in \cref{fig:global_reduction}.

\begin{figure}[!ht]
    \centering
    \includegraphics[width = \textwidth]{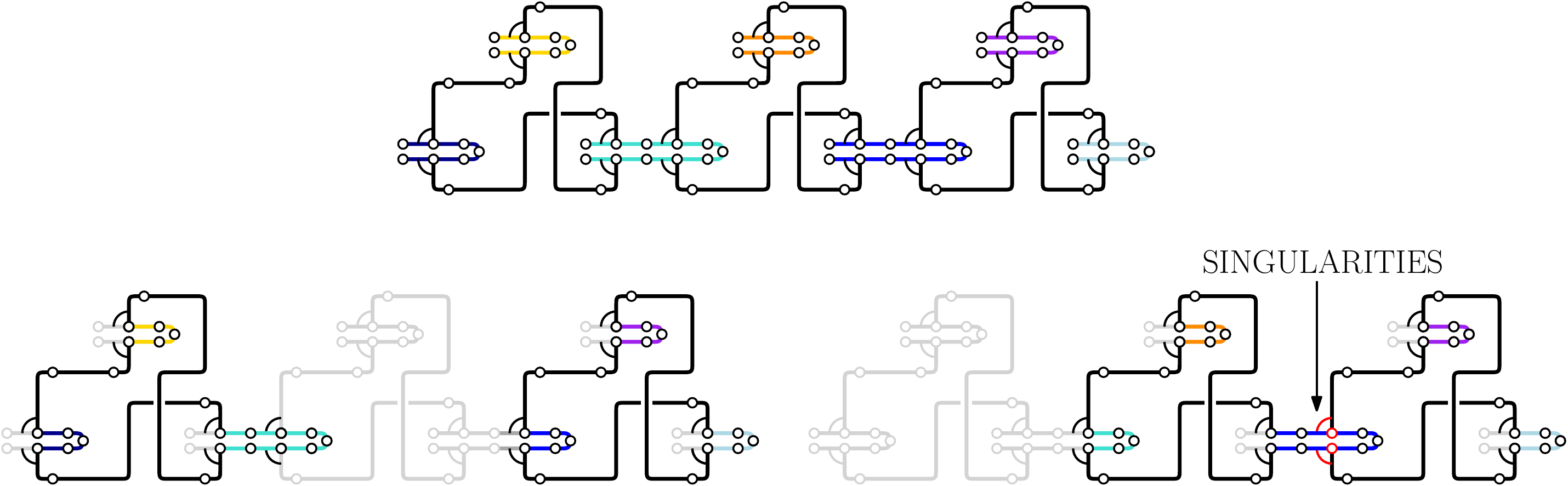}
    \caption{An illustration of how the graph from \cref{fig:disjoint_cycles} is mapped to spaces and how valid/invalid subsets of edges are mapped to manifolds/non-manifolds respectively.
    }
    \label{fig:global_reduction}
\end{figure}

We see in \cref{fig:global_reduction} that we can associate any pair of vertex disjoint cycles in the input graph to a pair of non-intersecting tori in the output space in an obvious way. Concretely, a cycle is mapped to a torus by sending the edges to edge gadgets and by then connecting these through the vertex gadgets. This association turns out to be a bijection with an inverse that maps a submanifold $X$ to the set of edges whose corresponding edge gadgets contains a $2$-simplex of $X$. That this inverse is well-defined is proved for the pathwidth-preserving reduction described in \cref{sec:triangulation} below. 

\subsubsection{A Potential Explosion in Pathwidth} \label{sec:potential_pathwidth_explosion}

This section investigates how the reduction described in \cref{sec:reduction_idea} can potentially blow up the pathwidth of the Hasse diagram of the output space. In particular, by choosing how the various gadgets are attached to each other in an ``adversarial'' way we can show that the gap between the pathwidth of an input graph and the Hasse diagram of the output space can be made arbitrarily large.

First we describe the family of graphs that we are interested in. It is a countable family, $G_1,G_2,G_3,\dots G_n,\dots$, and the first three graphs of the family are shown in \cref{fig:flexibility_in_reduction_graphs}. Described in words, the graph $G_n$ contains a directed paths on $n^2$ vertices, $v_1,v_2,v_3\dots, v_{n^2}$, as well as an additional vertex $v_0$. From $v_0$ there are directed edges $v_0v_i$ for all $1\leq i \leq m$.

\begin{figure}[!ht]
    \centering
    \includegraphics[width = \textwidth]{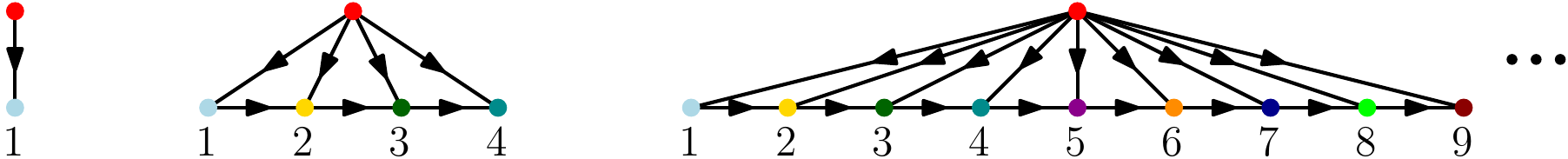}
    \caption{Graphs that are directed paths on $m$ vertices $v_i$ for $m=1$ (left), $m=4$ (middle) and $m=9$ (right) having also an additional vertex $u$ from which there are edges $uv_i$ for all $1\leq i \leq m$.}
    \label{fig:flexibility_in_reduction_graphs}
\end{figure}

The graphs in this family all have pathwidth at most $2$. In particular, the first graph has pathwidth $1$, since it is a path. The remaining graphs have pathwidth $2$. This is because a) they all have an obvious path decomposition of width $2$ and b) they all contain cycles when viewed as undirected graphs, so they do not have pathwidth $1$.

\begin{remark}
We can apply two (different) versions of the construction described in \cref{sec:reduction_idea} to the sequence of directed graphs $G_1,G_2,\dots, G_{n}, \dots $ and obtain two different sequences of spaces, $Y_1,Y_2,\dots, Y_n,\dots$ and $Y_1',Y_2',\dots, Y_n',\dots$, where:
\begin{itemize}
    \item the pathwidth of $Y_n$ is at most some fixed constant $c \leq 1000000$.
    \item the pathwidth of $Y_n'$ is at least $n$.
\end{itemize}
\end{remark}

The rest of this subsection will be focused on giving an informal justification of this remark. 

The construction of the spaces $Y_1,Y_2,\dots, Y_n,\dots$ is somehow the ``natural'' one. The general idea is that the order in which gadgets corresponding to outgoing edges are attached to the vertex gadget $v_0$ is given by the topological ordering of $G$, order where $v_1$ comes ``first'', then $v_2$, then $v_3$ etc. An example is shown in \cref{fig:flexibility_in_reduction_graphs_good} for $Y_3$. 

\begin{figure}[!ht]
    \centering
    \includegraphics[width = \textwidth]{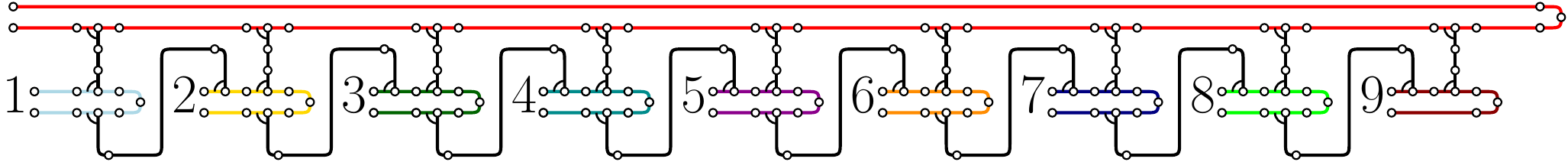}
    \caption{An example of one way the graph on $9$ vertices pictured in \cref{fig:flexibility_in_reduction_graphs} may be mapped to a space (using the notation of \cref{fig:notation}). By generalizing this ordering we get a family of spaces whose pathwidth is clearly bounded by some fixed constant.}
    \label{fig:flexibility_in_reduction_graphs_good}
\end{figure}

The spaces we produce using the obvious generalization of this pattern all have Hasse diagram with pathwidth bounded by some constant. A valid (but sub-optimal) path decomposition is shown in \cref{fig:flexibility_in_reduction_graphs_good_bags}. The idea is to use $n^2$ bags, where the bags $X_i$ for $1 \leq i < n^2$ contain
\begin{itemize}
    \item the simplices in the vertex gadget corresponding to $v_{i}$ and $v_{i+1}$.
    \item the simplices of the edge gadgets $v_{i}v_{i+1}$, $v_0v_{i}$ and $v_0v_{i+1}$.
    \item the simplices of a sub-cylinder of the vertex gadget $v_0$ having fixed size while also containing the intersection with the edge gadgets of $v_0v_{i}$ and $v_0v_{i+1}$.
\end{itemize}

Finally, the last bag $X_{n^2}$ contains the simplices of $X_{n^2-1}$ as well as the constant number simplices in the vertex gadget of $v_0$ that is not contained in the other bags. It is elementary to unwrap the notation and to see that this is indeed a valid path decomposition of the Hasse diagram of the spaces.

\begin{figure}[!ht]
    \centering
    \includegraphics[width = \textwidth]{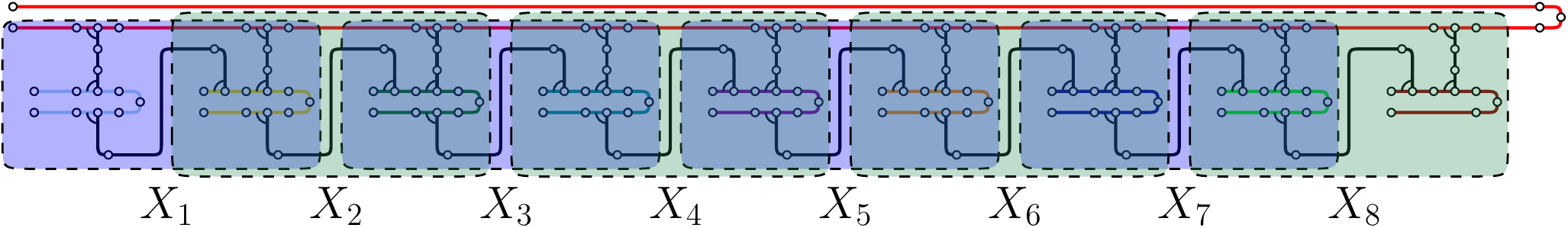}
    \caption{The path decomposition of the graph shown in \cref{fig:flexibility_in_reduction_graphs_good}. The last bag, $X_9$ is not pictured due to lack of space but it contains precisely $X_8$ and the remaining simplices not covered by the other bags.}
    \label{fig:flexibility_in_reduction_graphs_good_bags}
\end{figure}

An upper bound on the width of this particular kind of path decomposition can be found by a ``back-of-the-envelope'' computation. More precisely, each bag contains fewer than $100$ of the components from \cref{fig:notation}. Further more, each of these component contains fewer than $10000$ simplices. This means that each bag in the decomposition contains fewer than $1000000$ simplices, making the path width of the space less than $1000000$, as claimed.

We now turn our attention to the ``badly behaved'' spaces, $Y_n'$. These spaces are a bit more complicated than the ``nicely behaved'' spaces $Y_n$ we just saw. We will therefore start to carefully unravel the special case of $Y_3'$, shown in \cref{fig:flexibility_in_reduction_bad}. Once this is done, the general pattern will become obvious.

\begin{figure}[!ht]
    \centering
    \includegraphics[width = \textwidth]{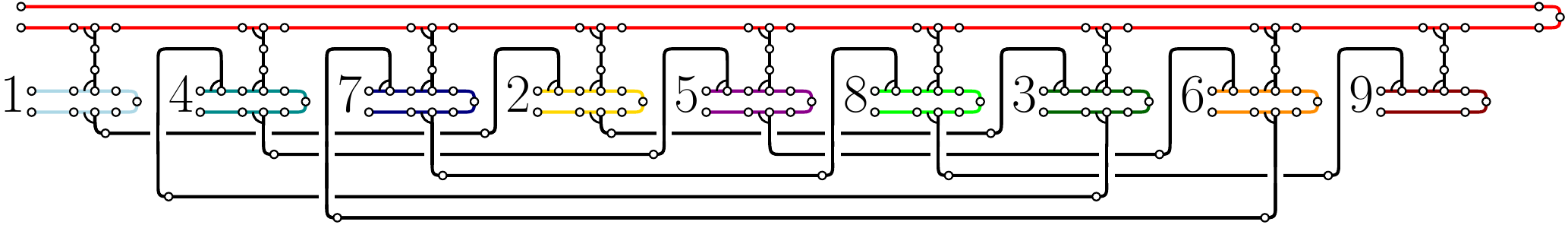}
    
    \vspace{0.3cm}
    
    \centering
    \includegraphics[width = \textwidth]{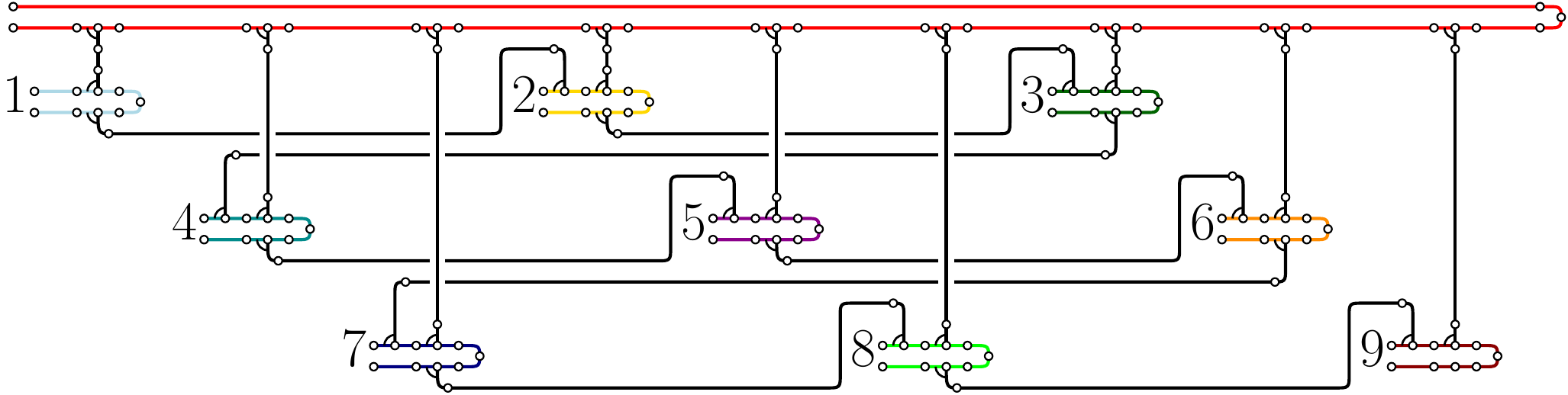}
    
    \caption{Two different figures depicting the same space, $Y_3'$.
    \label{fig:flexibility_in_reduction_bad}}
\end{figure}

Let us recall our claim for this space, namely that the pathwidth of the Hasse diagram of $Y_3'$ is higher than $3$. This is in many ways a modest goal, as intuition tells us that it must actually be much higher. However, it is the way in which we prove that the pathwidth of this space is higher than $3$ that is important, as this argument can be generalized to any $n$.

Concretely, our goal is to show that the Hasse diagram of this space contains a particular $3\times 3$ grid as a graph minor. This is useful, as it is well known that graphs that contains $n\times n$ grids as graph minors always have treewidth (and therefore also pathwidth) at least $n$. 

The first step towards this end was to deform the space slightly, as we did in the lower half of \cref{fig:flexibility_in_reduction_bad}. The next step is to look at a particular subcomplex shown in \cref{fig:flexibility_in_reduction_bad_2}. Since Hasse diagrams of subcomplexes are subgraphs of the Hasse diagram of the original complex, the graph minors of the Hasse diagram of any subcomplex are also graph minors of the Hasse diagram of the original graph.

\begin{figure}[!ht]
    \centering
    \includegraphics[width = \textwidth]{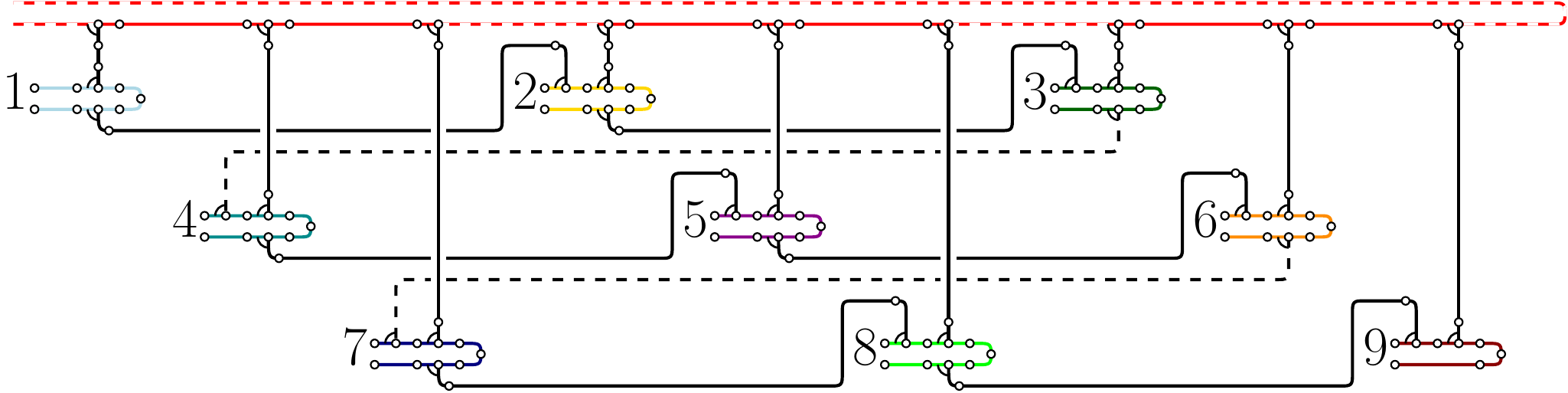}
    \vspace{0.3cm}
    
    \includegraphics[width = \textwidth]{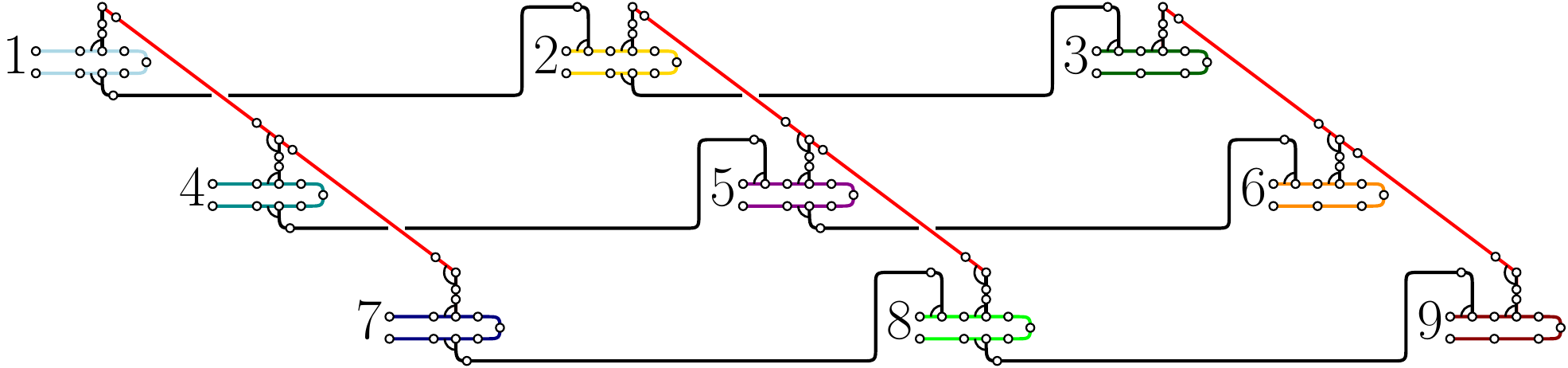}
    
    \caption{A subcomplex of the simplicial complex described in \cref{fig:flexibility_in_reduction_bad} (top) and space homeomorphic to this space that highlights the location of the location of the grid in the Hasse diagram of the space (bottom). \label{fig:flexibility_in_reduction_bad_2}}
\end{figure}

Finally, we pick an arbitrary simplex from each of the vertex gadgets (except for the vertex gadget corresponding to $v_0$). These simplices can be used to form the set of vertices of a $3\times 3$-grid contained as a minor in the Hasse diagram of the space.

\begin{figure}[!ht]
    \centering
    \includegraphics[width = \textwidth]{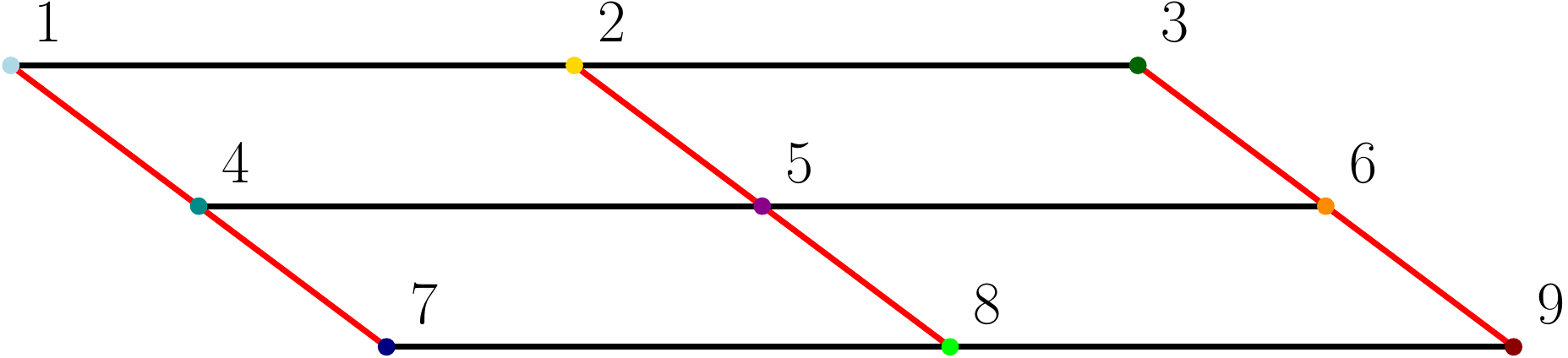}
    
    \caption{The $3\times 3$-grid minor of the subcomplex of $Y_3'$ shown in \cref{fig:flexibility_in_reduction_bad_2}. Hence $Y_3'$ has treewidth (and therefore also pathwidth) at least $3$. \label{fig:flexibility_in_reduction_bad_3}}
\end{figure}


    
The construction pictured in \cref{fig:flexibility_in_reduction_bad} can be generalized to produces a family of spaces whose treewidth can get arbitrarily bad. This is done by first attaching the vertex gadget $v_{1}$ to vertex gadget $v_0$ and then continue the process in the following order (reading from left to right, line by line from top to bottom):

$$
\begin{matrix}
v_{1+0\cdot n}, v_{1+1\cdot n}, \dots, v_{1+i\cdot n}, \dots, v_{1+(n-1)\cdot n},\\
v_{2+0\cdot n}, v_{2+1\cdot n}, \dots, v_{2+i\cdot n}, \dots, v_{2+(n-1)\cdot n},\\
\vdots\\
v_{j+0\cdot n}, v_{j+1\cdot n}, \dots, v_{j+i\cdot n}, \dots, v_{j+(n-1)\cdot n},\\
\vdots\\
v_{n+0\cdot n}, v_{n+1\cdot n}, \dots, v_{n+i\cdot n}, \dots, v_{n+(n-1)\cdot n}     
\end{matrix}
$$

This gives us the same kind of grid shape as we saw in the lower part of \cref{fig:flexibility_in_reduction_bad}, except it is now $n\times n$ rather than $3\times 3$. Fix a set of $n^2$ 0-simplices in the Hasse diagram, one from the interior of each vertex gadget. By attaching the vertex gadgets of each of the $v_i$'s to $u$ in the way we described, we have now ensured that in consecutive gadgets along each row there are paths in the Hasse diagram that passes through the vertex gadget $u$. Furthermore, between arbitrary pairs of simplices contained in consecutive gadgets along each column there are paths in the Hasse diagram going through the gadget associated to the edge $v_{j' + i'\cdot n}v_{j'+1 + i'\cdot n}$. Finally, these paths can be chosen so that they only intersect each other at their endpoints. By contracting these paths, we get our $n\times n$-grid as a grid minor of the Hasse diagram. It is well known that this implies that the Hasse diagram of the space has treewidth (and therefore also pathwidth) at least $n$.


\subsection{Avoiding the Explosion of Pathwidth}\label{sec:triangulation}

To make sure that the reduction works we need to guarantee that the pathwidth of the outputted space is bounded linearly by the pathwidth of the inputted graph. In order to ensure this we will work with several objects, each playing a part in the inductive construction of the reduction. We list the notation for all of these objects here and also add to each a short description.

\begin{itemize}
\itemsep0pt
    \item We are given a directed graph $D$ which is the input to the DCP problem.
    \item We are given a nice path decomposition of $D$ denoted by $\ntd{D}$.
    \begin{itemize}
        \item The bag at node $t$ of $\ntd{t}$ is denoted as $\bag{t}$.
        \item The set of vertices of $D$ that has been forgotten at node $t$ is denoted by $F_t$.
        \item We let the nodes of $\ntd{D}$ be totally ordered so that consecutive nodes are adjacent.
        \item We let the first bag be the leaf bag and the last bag be the root bag $X_r$.
    \end{itemize}
    \item We construct a space $\aspace{t}$ for each node $t$ of $\ntd{D}$.
    \begin{itemize}
        \item This space is defined by induction over the path decomposition $\ntd{D}$.
        \item The spaces are all nested so that $\aspace{t'} \subset \aspace{t}$ for every node $t' < t$ in $\ntd{D}$.
        \item The space $Y_D:= Y_r$ constructed at the end of the induction will contain $\ell$ disjoint manifolds (all of them tori) if and only if $D$ contains $\ell$ vertex-disjoint cycles.
     
    \end{itemize}
    \item We construct a path decomposition $\tdspace{t}$ of $\aspace{t}$ for every node $t$ of $\ntd{D}$.
    \begin{itemize}
        \item The path decompositions is defined by induction over the path decomposition $\ntd{D}$.
        \item The path decompositions are nested in the sense that every path decomposition $\tdspace{t'}$ appears as a sub-path of $\tdspace{t}$ for every node $t' \leq t$.
        \item The path decomposition $\tdspace{t}$ will always be bounded in size so that $\tdspace{t} < c\cdot \max_{t'\leq t}|X_{t'}|$ for some fixed constant $c$ (that remains fixed for all problem instances).
    \end{itemize}
\end{itemize}

\begin{figure}[!ht]
    \centering
    
    \includegraphics[width = \textwidth]{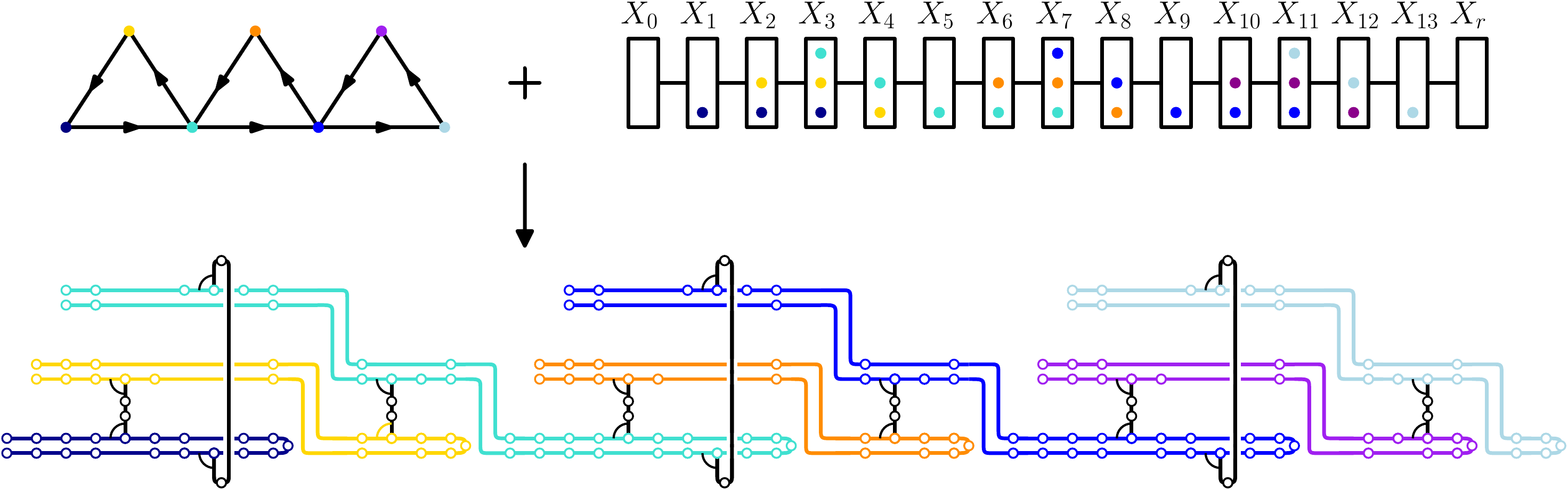}
    \caption{How the graph from \cref{fig:disjoint_cycles} (top left) is mapped to a space (bottom) having the same ``structure''/``ordering'' as the given nice path decomposition (top right) of the graph. This is essentially the space you would construct by following the inductive procedure described in \cref{sssection: inductive construction}.}
    \label{fig:example of low treewidth reduction}
\end{figure}

\subsubsection{Induction Hypothesis}
\label{sssection: induction hypothesis}

The induction hypothesis is somewhat complicated, involving multiple statements. We have separated these statements into three groups to make them more manageable. Roughly speaking, the first group is about the existence and properties of important subcomplexes of the space $\aspace{t}$. The second group is about key properties of $\tdspace{t}$. The third group will be important for proving that the construction results in a space that can be used as a reduction.

The space $Y_t$ consists of many components which we believe might be easier to understand visually. We therefore recommend that the reader study \cref{fig:example of low treewidth reduction} and \cref{fig:example of partitioning in the induction step } in order to get a better understanding of the role of the different components. The induction hypothesis we will use in our construction is as follows. 

\par

\textbf{\textit{Existence of components of $\aspace{t}$:}}
\begin{enumerate}
\itemsep0pt
    \item the 2-simplices of $\aspace{t}$ are partitioned into two families of subcomplexes:
    \begin{itemize}
    \itemsep0pt
        \item for every edge $uv$ with $u,v\in F_t\cup X_t$ and where $u$ and/or $v$ in $F_t$, $Y_t$ contains a cylinder shaped subcomplex $\aspace{t}^{uv}$. The two boundary components of $\aspace{t}^{uv}$ are denoted $IN(\aspace{t}^{uv})$ and $OUT(\aspace{t}^{uv})$. \footnote{Intuitively, this is the edge gadget for $uv$ described in \cref{sec:reduction_idea}.}
        \item for every vertex $v \in F_t \cup X_t$, $\aspace{t}$ contains a cylinder shaped subcomplex $\aspace{t}^v$.\footnote{Essentially, $Y_t^v$ is the vertex gadget described in \cref{sec:reduction_idea} for all $v\in F_t$. Technically, the two simplicial complexes are slightly different as the gadget constructed in this induction contains some extra cylinders added for cosmetic reasons, but they are still homeomorphic.}
    \end{itemize}
    \item the 2-simplices of $\aspace{t}^v$ are in turn partitioned by into the two subcomplexes $\mathcal{A}_t^v$ and $\mathcal{B}_t^v$. \footnote{The subcomplex $\mathcal{A}_t^v$ is the subcomplex of the vertex gadget that is attached to outgoing edge gadgets, and $\mathcal{B}_t^v$ is the subcomplex of the vertex gadget attached to incoming edge gadgets.} The complexes $\mathcal{A}_t^v$ and $\mathcal{B}_t^v$ satisfy the following properties:
        \begin{itemize}
            \itemsep0pt
            \item $\mathcal{A}_t^v$ contains a boundary component $\overline{\mathcal{A}_t^v}$ with a cylinder neighbourhood.
            \item $\mathcal{B}_t^v$ contains a boundary component $\overline{\mathcal{B}_t^v}$ with a cylinder neighbourhood.
            \item if $v \in X_t$ then $\mathcal{A}_t^v$ and $\mathcal{B}_t^v$ are disjoint.
            \item if $v \in F_t$ then $\mathcal{A}_t^v$ intersects $\mathcal{B}_t^v$ at $\overline{\mathcal{A}_t^v} = \overline{\mathcal{B}_t^v}$.
        \end{itemize}
        
    \item for every edge $uv$ with $u,v\in F_t\cup X_t$ and where $u$ and/or $v$ in $F_t$:
    \begin{itemize}
            \itemsep0pt
            \item $\mathcal{B}_t^u$ intersects $\aspace{t}^{uv}$ at $IN(\aspace{t}^{uv})$.
            \item $\mathcal{A}_t^v$ intersects $\aspace{t}^{uv}$ at $OUT(\aspace{t}^{uv})$.
        \end{itemize}
\end{enumerate}

\begin{figure}[!ht]
    \centering
    \includegraphics[width = \textwidth]{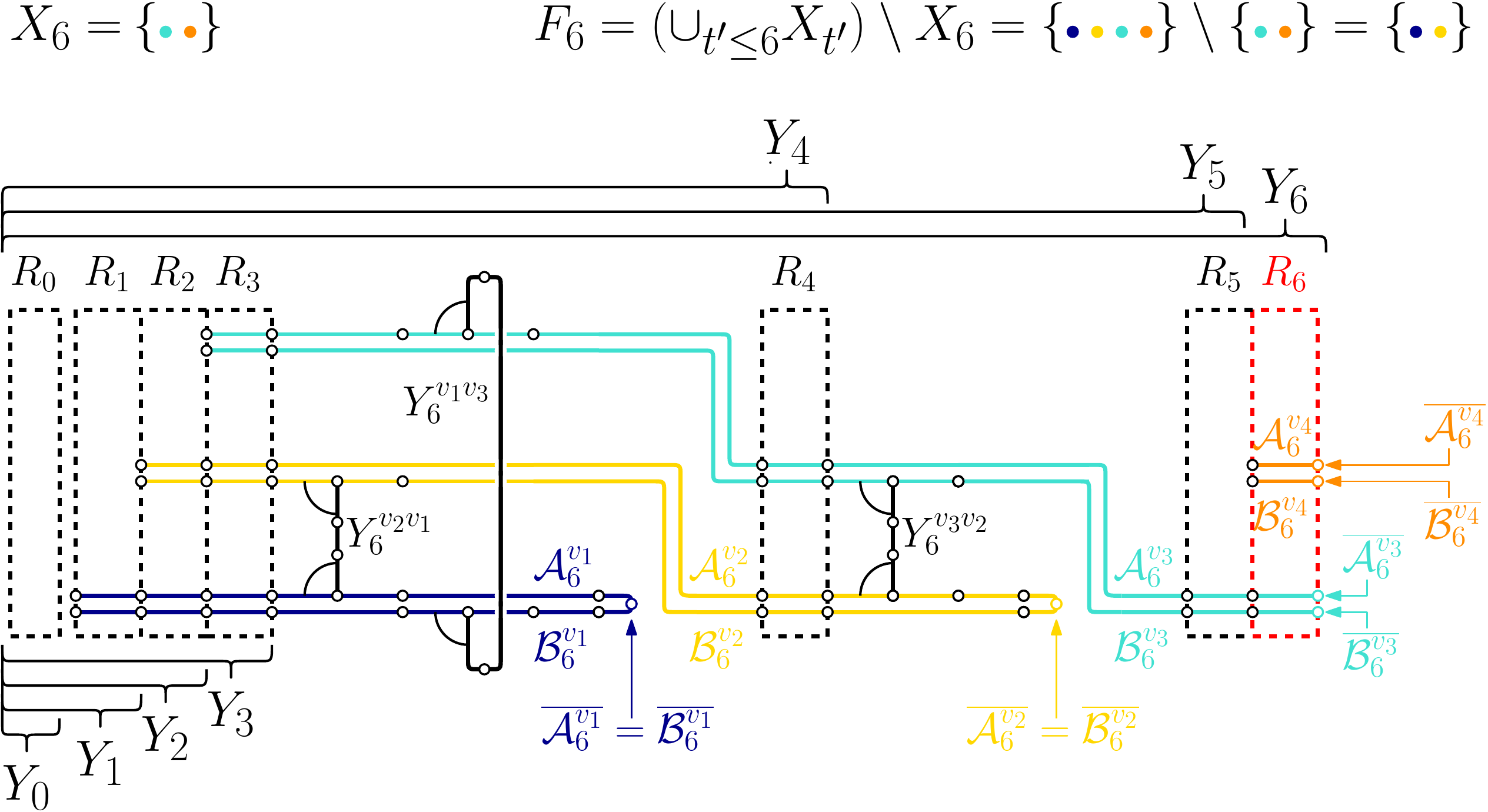}
    \caption{\label{fig:example of partitioning in the induction step } An example of the different components of $Y_t$. In particular, we have zoomed in on the construction of the space from \cref{fig:example of low treewidth reduction} at node $t=6$. We followed the construction presented in \cref{sssection: inductive construction} to make this space.}
\end{figure}

\textbf{\textit{Conditions on the path decomposition $\tdspace{t}$ of $Y_t$:}}
\begin{enumerate}
\itemsep0pt
    \item $\tdspace{t}$ is a path decomposition of the Hasse diagram of $Y_t$
    \item $\tdspace{t}$ contains a bag $R_t$ with degree $\leq 1$ that is precisely equal to set of the (cylindrical) neighbourhoods of $\overline{\mathcal{A}_t^v}$ and $\overline{\mathcal{B}_t^v}$ for every $v\in X_t$.
    \item The width of $\tdspace{t} \leq c\cdot \operatorname{npd}(D)$ for some fixed constant $c$. 
\end{enumerate}



\textbf{\textit{Properties of $Y_t$ used in the proof of correctness:}}
\begin{enumerate}
\itemsep0pt
    \item For every edge $uv$ with $u,v\in F_t\cup X_t$ where $u$ and/or $v$ in $F_t$ there is a...
    \begin{itemize}
        \itemsep0pt
        \item cylinder in $\mathcal{B}_t^u$ with boundary components $OUT(\aspace{t}^{uv})$ and $\overline{\mathcal{B}_t^u}$
        \item cylinder in $\mathcal{A}_t^v$ with boundary components $IN(\aspace{t}^{uv})$ and $\overline{\mathcal{A}_t^v}$
    \end{itemize}
    \item If $M$ is a (non-empty) $2$-manifold in $\mathcal{B}_t^u$ having a boundary that a) intersects some other component and b) has a boundary on the boundary of $\mathcal{B}_t^u$ then $M$ is a cylinder having $OUT(\aspace{t}^{uv})$ as one boundary component and $\overline{\mathcal{B}_t^u}$ as the other. Furthermore, any two such $2$-manifolds in $\mathcal{B}_t^u$ intersect and $M\cup\aspace{t}^{uv}$ is a cylinder.
    \item If $M$ is a (non-empty) $2$-manifold in $\mathcal{A}_t^v$ having a boundary that a) intersects some other component and b) has a boundary on the boundary of $\mathcal{A}_t^v$ then $M$ is a cylinder having $IN(\aspace{t}^{uv})$ as one boundary component and $\overline{\mathcal{A}_t^v}$ as the other. Furthermore, any two such $2$-manifolds in $\mathcal{A}_t^v$ intersect and $M\cup\aspace{t}^{uv}$ is a cylinder.
    
\end{enumerate}

\subsubsection{Reduction by Structural Induction on Path Decompositions}
\label{sssection: inductive construction}

There are three kinds of bags in a path decomposition of a graph, each of which requires special attention. We need to prove that for every kind of node $t$ we can construct a space $\aspace{t}$ and a path decomposition $\tdspace{t}$ (in polynomial time) satisfying the induction hypothesis. If $t$ has a child node $t'$ we may assume by induction that we have access to a space $\aspace{t'}$ and a path decomposition $\tdspace{t'}$ that satisfies the induction hypothesis at the child node.
\par
\textbf{The Leaf Node} is particularly easy to deal with. In particular, the empty simplicial complex and the ``empty'' path decomposition (containing a single empty bag) satisfies the criteria of the induction hypothesis.

\begin{lemma}
Let $D$ be a directed graph, $\ntd{D}$ be a nice path decomposition of $D$ and $t$ be an leaf node in $\ntd{D}$. Then we can construct a space $\tdspace{t}$ and a path decomposition $\tdspace{t}$ satisfying the induction hypothesis for the node $t$ in constant time. 
\end{lemma}
\begin{proof}
The leaf node $t$ is the base case of the inductive construction. By definition there are no explored/forgotten nodes. This means that setting $\aspace{t}=\emptyset$ is a valid construction. The path decomposition is also easily defined, as it consists of a single bag, $R_t$, which is empty. It is trivial to verify that this construction satisfies the induction hypothesis.
\end{proof}

\textbf{The Introduce Bag} is a bit more complicated than the leaf bag. Here we need to extend the space and path decomposition we have constructed inductively by adding the space $Y_t^{v}$ for the new vertex. We have decided to add edge gadgets when vertices are forgotten, so this procedure is not too complicated. The space representing the new vertex is just the disjoint union of two cylinders.

\begin{lemma}
Let $D$ be a directed graph, $\ntd{D}$ be a nice path decomposition of $D$ and $t$ be an introduce node in $\ntd{D}$ with a child node $t'$. Given a space $\aspace{t'}$ together with a path decomposition $\tdspace{t'}$ satisfying the induction hypothesis for the node $t'$, we can construct a space $\tdspace{t}$ and a path decomposition $\tdspace{t}$ satisfying the induction hypothesis for the node $t$ in polynomial time.
\end{lemma}

\begin{figure}[!ht]
    \centering
    \includegraphics[width = .8\textwidth]{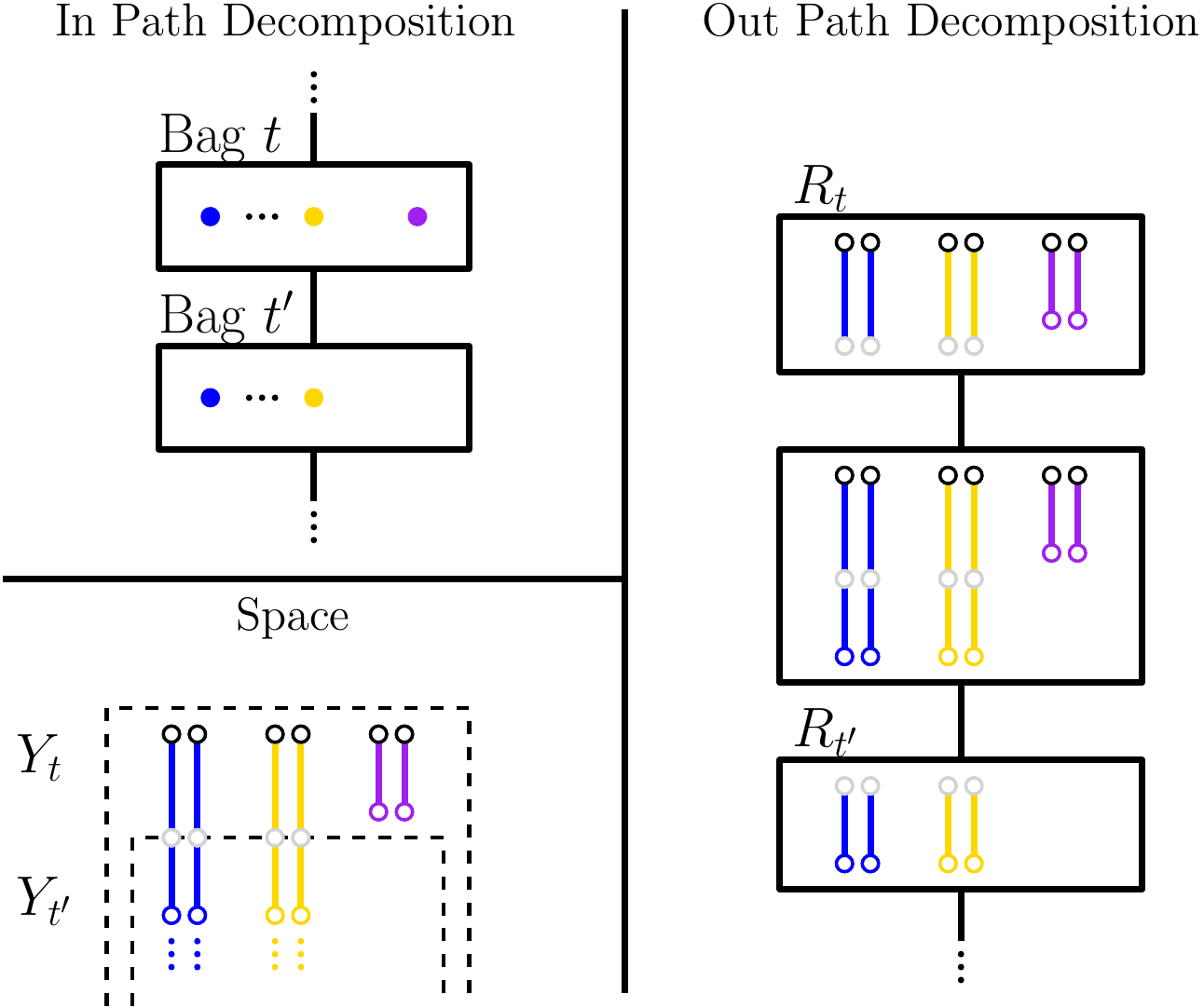}
    \caption{An overview of how to construct the space $\aspace{t}$ and the path decomposition $\ntd{Y_t}$ from $\aspace{t'}$ and $\ntd{Y_{t'}}$ at an introduce bag. In particular, the top left figure shows the contents of the bags of the input path decomposition at node $t$ and $t'$. The lower left figure shows how to extend the space $Y_{t'}$ to $Y_t$ and the rightmost figure shows how to extend the path decomposition $\ntd{Y_{t'}}$ by adding two bags to $R_{t'}$. The correspondence between nodes and spaces is given by color coding. }
    \label{fig:reduction1_intro_bag}
\end{figure}
\begin{proof}
Let $t$ be an introduce bag with a child bag $t'$ and let $w$ be the newly introduced vertex. By induction, we can assume that we have already constructed the space $\aspace{t'}$ and that it's $2$-simplices are partitioned by subcomplexes in the way specified by the induction hypothesis. We also assume that we have a path decomposition $\tdspace{t'}$ and that contains the bag $R_t$. The new space we construct, $\aspace{t}$, is defined by \cref{fig:reduction1_intro_bag_fig_2}. Note that we extend each space $Y_t^{u}$ for each existing vertex $u$ with a cylinder and that we add components $\mathcal{A}_t^w$ and $\mathcal{B}_t^w$ for the newly introduced vertex $w$. The path decomposition $\tdspace{t}$ of $\aspace{t}$ is similarly defined as an extension of the path decomposition of $\tdspace{t'}$ where we add one bag attached to $R_{t'}$ containing the simplices of $R_{t'}$ together with the new simplices. Then we attach $R_t$ to this intermediate bag, which contains the new boundaries and their cylindrical neighbourhoods, see \cref{fig:reduction1_intro_bag} above for details on the path decomposition.

We are now ready to prove the statements of the induction hypothesis for the new space and path decomposition. The two first groups of statements are rather trivial to prove, as they generally hold by construction and an elementary induction argument. We first discuss the statements regarding the properties of the components of $Y_t$

We added a component $Y_t^w$ for the vertex introduced to the bag. This component is in turn partitioned into $\mathcal{A}_t^w$ and $\mathcal{B}_t^w$, each containing precisely one cylinder, with one of the boundaries of each cylinder set to be equal to $\overline{\mathcal{A}_t^w}$ (respectively $\overline{\mathcal{B}_t^w}$). We also add two cylinders for each vertex $x$ in the bag and we attached these to the boundaries $\overline{\mathcal{A}_{t'}^x}$ and $\overline{\mathcal{B}_{t'}^x}$. The boundaries of these components are now set to be the other end of these cylinders (see \cref{fig:reduction1_intro_bag_fig_2}) and they now belong to the component they are attached to. The other $2$-simplices of the space are partitioned in the same way they where in $Y_{t'}$. Since the newly added component is not adjacent to any forgotten vertex (by basic properties of path decompositions) and since no new nodes have been forgotten moving from bag $t'$ to bag $t$, the rest of the properties about the components of the space remain true.

\begin{figure}[!ht]
    \centering
    \includegraphics[width = .8\textwidth]{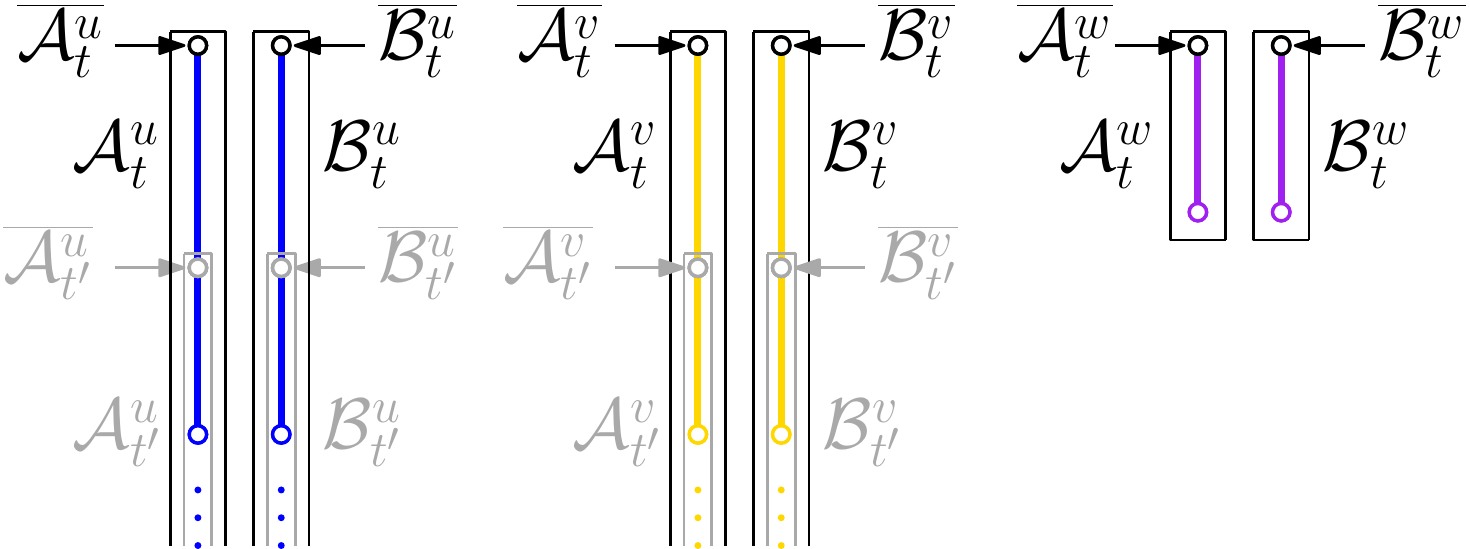}
    \caption{A more detailed description of how to construct the space $\aspace{t}$ from $\aspace{t'}$. In particular, the figure includes some additional labels naming some of the components of the space that where not clearly marked in \cref{fig:reduction1_intro_bag}.}
    \label{fig:reduction1_intro_bag_fig_2}
\end{figure}
To check that the properties of the path decomposition $\ntd{Y_t}$ holds is equally elementary. First, the proposed set of bags is a path decomposition because every simplex is in some bag, every face/coface pair is contained in some bag and the bags containing any given simplex forms a connected sub path. A proof of each statement follows from construction and together with a simple proof by induction. The bag $R_t$ exists purely by construction and the proof that the width of the path decomposition $\ntd{Y_t}$ is linearly bounded by the $\ntd{D}$ again only requires a simple induction proof (combined with the observation that neither of the bags we added are too large).

Finally, the properties required for the correctness proof are equally simple at the introduce nodes. First, the existence of cylinders with boundaries $IN(Y_t^{uv})$ and $\overline{\mathcal{A}_t^v}$ follows by construction and an elementary inductive proof. 

Properties 2 and 3 have similar proofs so we will only argue for property 3. In the case of components corresponding to the newly introduced vertex this is trivially true by construction, as this component contains no such manifolds. For the components corresponding to the other vertices in the bag these properties can be proved by induction. More concretely, let $M$ be a manifold contained in $\mathcal{A}_t^v$ and assume it has a boundary that a) intersects some other component of the space and b) has a boundary on the boundary of $\mathcal{A}_t^v$. Then $M$ cannot be contained entirely in $\mathcal{A}_{t'}^v$ because by induction this means that $M$ must have $\overline{\mathcal{A}_{t'}^v}$ as a boundary component which is not a boundary of $\mathcal{A}_t^v$ nor is part of another component. By construction, if $M$ contains one of the newly added simplices it must contain all to satisfy a) and b) and hence one boundary component of $M$ is $\overline{\mathcal{A}_t^v}$. Furthermore, $M \cap \mathcal{A}_{t'}^v$, is now a manifold and must therefore be a cylinder with boundary components $\overline{\mathcal{A}_{t'}^v}$ and some $IN(\aspace{t}^{uv})$.
\end{proof}

\textbf{The Forget Bag} is the most complicated, but we have been building up to it gradually. It is in many ways similar to the introduce bag in that we we need to extend the space and path decomposition we have constructed inductively. Here, instead of beginning to encode a vertex we need to finish it. In particular, we need to encode the edges between the forgotten vertex and its neighbours in the bag. This guarantees that every edge is encoded precisely once, since every edge there is precisely one bag containing both ends of the edge and where one end is forgotten. After these edges have been encoded, we just have to connect the two components $\mathcal{A}_t^v$ and $\mathcal{B}_t^v$ for the forgotten vertex.

\begin{lemma}
Let $D$ be a directed graph, $\ntd{D}$ be a nice path decomposition of $D$ and $t$ be an forget node in $\ntd{D}$ with a child node $t'$. Given a space $\aspace{t'}$ together with a path decomposition $\tdspace{t'}$ satisfying the induction hypothesis for the node $t'$, we can construct a space $\tdspace{t}$ and a path decomposition $\tdspace{t}$ satisfying the induction hypothesis for the node $t$ in polynomial time.
\end{lemma}

\begin{figure}[!ht]
    \centering
    \includegraphics[width = \textwidth]{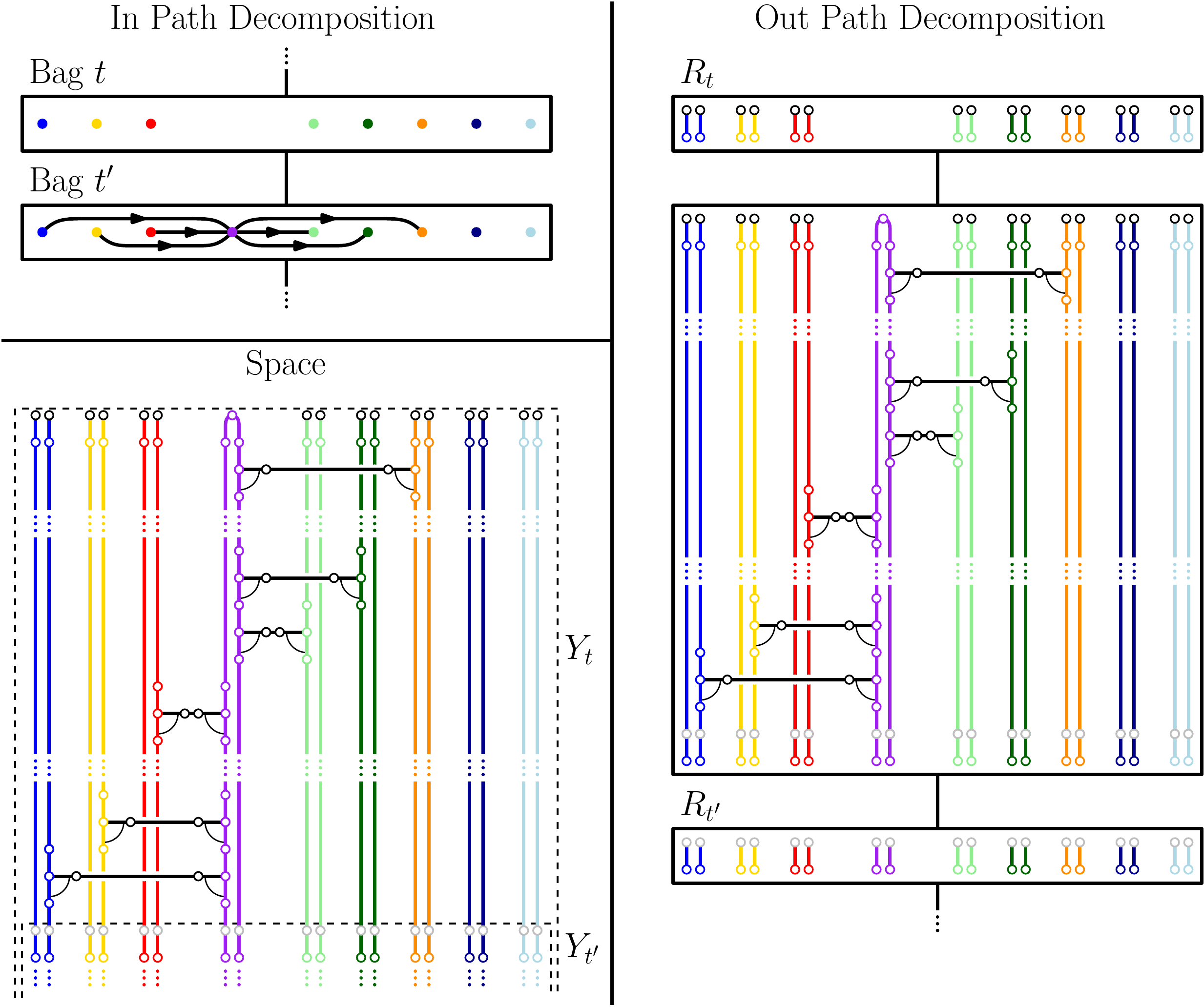}
    \caption{
    An overview of how to construct the space $\aspace{t}$ and the path decomposition $\ntd{Y_t}$ from $\aspace{t'}$ and $\ntd{Y_{t'}}$ at a forget bag. In particular, the top left figure shows the contents of the bags of the input path decomposition at node $t$ and $t'$. The lower left figure shows how to extend the space $Y_{t'}$ to $Y_t$ and the rightmost figure shows how to extend the path decomposition $\ntd{Y_{t'}}$ by adding two bags to $R_{t'}$. The correspondence between nodes and spaces is given by color coding. Each of the black components are edge gadgets.}
    \label{fig:reduction1_forget_bag}
\end{figure}

\begin{proof}
Let $t$ be an forget bag with a child bag $t'$ and let $v$ be the newly introduced vertex. By induction, we can assume that we have already constructed the space $\aspace{t'}$ and that it is partitioned as specified in the induction hypothesis. We also assume that we have a path decomposition $\tdspace{t'}$ containing a bag $R_t$. The construction of the new space, $\aspace{t}$, and its path decomposition is sketched in \cref{fig:reduction1_forget_bag}. The partitioning of the simplicies in the new space is shown in \cref{fig:reduction1_forget_bag_fig_2}.

\begin{figure}[!ht]
    \centering
    \includegraphics[width = \textwidth]{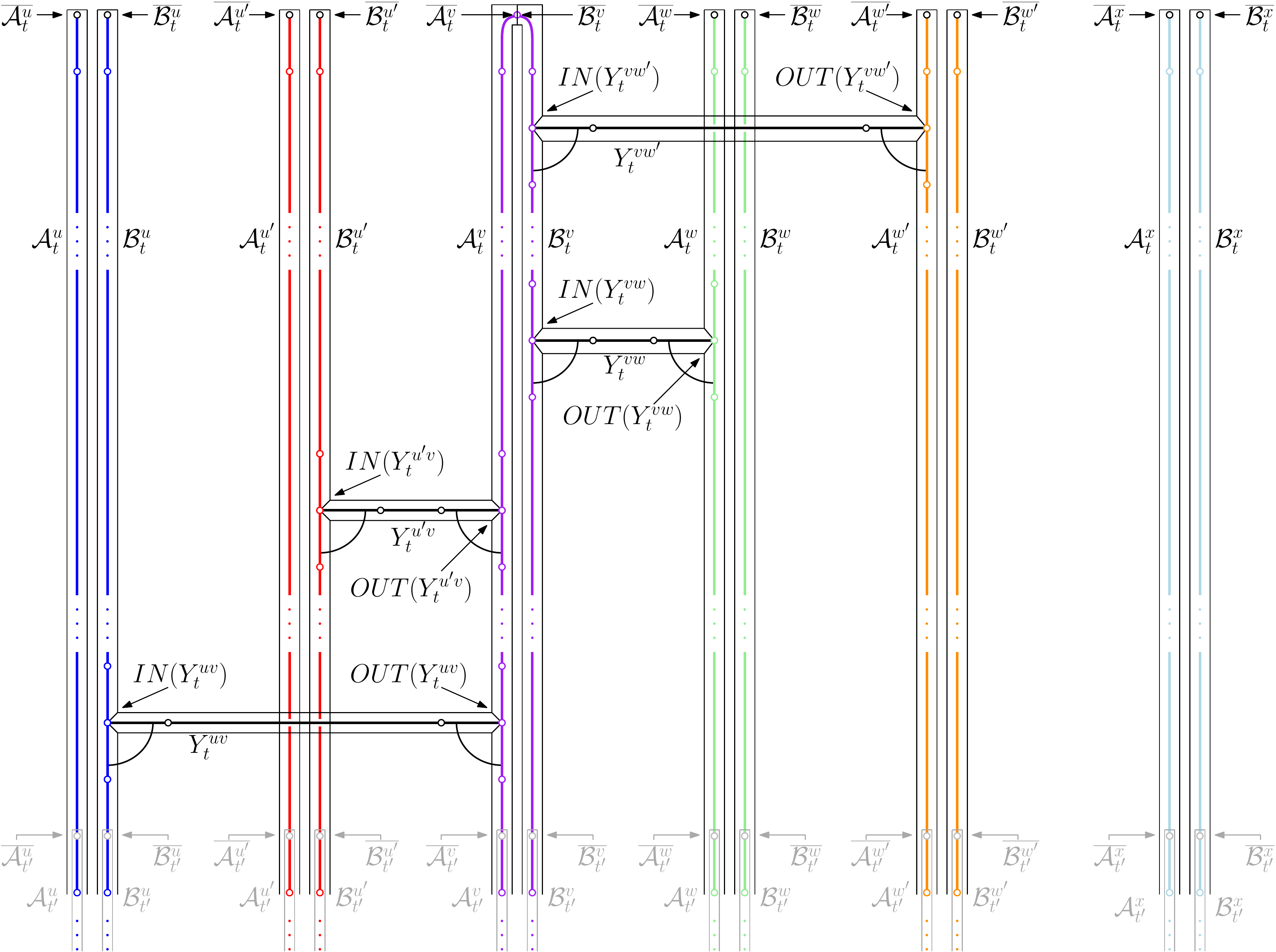}
    \caption{A more detailed description of how to construct the space $\aspace{t}$ from $\aspace{t'}$. In particular, the figure includes some additional labels indicating some of the components of the space that were not clearly marked in \cref{fig:reduction1_forget_bag}.}
    \label{fig:reduction1_forget_bag_fig_2}
\end{figure}


We use the same partitioning of the $2$-simplices contained in $Y_{t'}$ that we have constructed inductively. The newly added simplices to the space are partitioned into new and existing components in the way that is indicated by \cref{fig:reduction1_forget_bag_fig_2}. In particular, there are simplices added because of edges going to (resp. from) vertices in the bag from (resp. to) the vertex that is forgotten. The $2$-simplices of such an edge, $yv$ (resp. $vz$), make up a new component denoted $Y_t^{yv}$ (resp. $Y_t^{vz}$). The distinguished boundaries of the various components are updated as illustrated in the figures above. We do not discuss the first groups two of statements involved in the induction hypothesis this time around, because a) these ``proofs'' are essentially just observations and b) they are very similar to those we saw at the introduce bag. 

We will now prove the properties required to show the correctness of the reduction.

First, we prove the existence of cylinders in $\mathcal{A}_t^x$ with boundaries $IN(Y_t^{xy})$ and $\overline{\mathcal{A}_t^v}$ if at least one of the vertices $x$ or $y$ has been forgotten. The exact argument here depends on the vertex $x$. The only interesting case is if $x$ is in the bag at node $t'$. If the component $IN(Y_t^{xy})$ was added when we moved from bag $t'$ to $t$ then it is clear that such a cylinder exists by construction. If $IN(Y_t^{xy})$ was added earlier, then there is a cylinder from this boundary to $\overline{\mathcal{A}_{t'}^x}$ by the induction hypothesis. This cylinder can then be composed with the cylinder from $\overline{\mathcal{A}_{t'}^x}$ to $\overline{\mathcal{A}_{t}^x}$. An analogous argument can be used for the cylinders in $\mathcal{B}_{t}^x$.

Properties 2 and 3 have similar proofs, so we will only present the argument for property 3. Furthermore, we will only consider the case where the component is $\mathcal{A}_{t}^v$, where $v$ is the node that is being forgotten. The proofs for the other cases (in the case of different component and spaces corresponding to different vertices) the arguments are analogous.

By the same argument we saw at the introduce node, we can show that any manifold $M$ in $\mathcal{A}_{t}^v$ must intersect the newly added simplices. There are now two cases: either $M$  intersects $\mathcal{A}_{t'}^v$ or it does not. If $M$ intersects $\mathcal{A}_{t'}^v$ then we can make the argument that $M\cap \mathcal{A}_{t'}^v$ is a cylinder with one boundary component equal to some $IN(Y_{t}^{yv})$ for some previously forgotten node $y$ adjacent to $v$ and the other boundary component equal to $\overline{\mathcal{A}_{t'}^v}$. The latter is not a boundary of $\mathcal{A}_{t}^v$, so $M$ must contain some of the newly added $2$-simplices in $\mathcal{A}_{t}^v$. Because of the singularities, we can never add simplices so that our manifold $M$ have one of the $IN(Y_{t}^{u''v})$ as boundaries. Indeed, our only option is to move towards $\overline{\mathcal{A}_{t}^v}$. A similar argument works in the case when $M$ is entirely contained as a submanifold of the newly added simplices. If you assume $M$ to have some $IN(Y_{t}^{u''v})$ as one of its boundary components, then you are forced by the singularites to add $2$-simplices to $M$ until you end at $\overline{\mathcal{A}_{t}^v}$.
\end{proof}

\subsection{Proof of Correctness}\label{sec:proof_of_correctness}
Having done most of the work in the previous section, the actual proof that the reduction works is a simple lemma.

\begin{lemma}
The directed graph $D$ contains $\ell$ disjoint directed cycles if and only if $Y_D = Y_{root}$ contain $\ell$ disjoint manifolds. Furthermore, every manifold without a boundary contained in $Y_D$ is a torus.
\end{lemma}
\begin{proof}
We prove the forward direction first. Let $C_1,\dots, C_\ell$ be disjoint directed cycles in $D$. Then we know there are $\ell$ disjoint tori $T(C_1),\dots, T(C_\ell)$ in $Y_D$. These are obtained in the obvious way: Let $C_1$ be the cycle on the directed edges $(v_1v_2,v_2v_3,\dots, v_{p-1}v_p, v_{p}v_1)$. Then by Property 1 we know that $Y_D$ contains a cylinder in $\mathcal{B}^{v_1}$ attaching $\overline{\mathcal{B}^{v_1}}$ to $IN(Y^{v_1v_2})$ and a cylinder in $\mathcal{A}^{v_2}$ attaching $OUT(Y^{v_1v_2})$ to $\overline{\mathcal{A}^{v_2}}$. There is also a cylinder $Y^{v_1v_2}$ attaching $IN(Y^{v_1v_2})$ to $OUT(Y^{v_1v_2})$ by construction. By Properties 2 and 3, we know that the unions of these intersecting cylinders form a new cylinder (or potentially a torus if it's two boundary components are identified). The same is true for all edges in $C_1$. 

Since $C_1$ is a cycle, every vertex is unique, so there are by construction no intersection of these cylinders beyond those described in the above paragraph. Hence, if we let $T(C_1)$ be the union of all these cylinders, we have constructed our first torus. The other tori are constructed in the same way. Again, since the cycles are all disjoint, the tori constructed in this way can not intersect at any point. This concludes the forward direction.

For the other direction, we assume that $Y_D$ contains $\ell$ disjoint tori, $T_1,\dots, T_\ell$. Let us consider $T_1$. Since $T_1$ is a torus it must contain some $2$-simplex. This $2$-simplex must be contained in one of the three types of components of $Y_D$: $\mathcal{A}^v_2$, $\mathcal{B}^v_1$ or $Y^{v_1v_2}$. 

Let us assume that the $2$-simplex belongs to $Y^{v_1v_2}$ (the other cases have similar arguments). We know from construction that $Y^{v_1v_2}$ is a cylinder and that only the boundary components of this cylinder are attached to other components of the space. 
Since $T_1$ is a torus it should not have a boundary, which means that if one of the $2$-simplices of $Y^{v_1v_2}$ are in $T_1$ then all of $Y^{v_1v_2}$ must be contained in $T_1$. Moreover, the boundary component $IN(Y^{v_1v_2})$ intersects $\overline{\mathcal{A}{^v_2}}$. Since this cannot be a boundary in $T_1$, $T_1$ must itself intersect $\mathcal{A}^{v_2}$ in such a way that one of the boundaries of $T_1\cap \mathcal{A}^{v_2}$ is $IN(Y^{v_1v_2})$. By Property 3, $T_1$ intersects $\mathcal{A}^{v_2}$ as a cylinder with the other boundary equal to $\overline{\mathcal{A}^{v_2}}$. 

Again, we have that $\overline{\mathcal{A}^{v_2}}$ cannot be a boundary of $T_1$, so $T_1$ must intersect $\overline{\mathcal{B}^{v_2}}$ in such a way that one of the boundaries of $T_1\cap \overline{\mathcal{B}^{v_2}}$ is $\overline{\mathcal{A}^{v_2}}=\overline{\mathcal{B}^{v_2}}$ This means (by Property 2) that the other boundary is some $OUT(Y^{v_2v_3})$. This argument is then repeated until we eventually have that $T_1$ intersects some component with one boundary equal to $OUT(Y^{v_1v_2})$. The edge gadgets the tori contain, $v_1v_2, v_2v_3 \dots, v_{p-1}v_p, v_{p}v_1$ will then by definition be our first cycle $C(T_1)$. Note that these vertices cannot have been repeated, since we assumed that $T_1$ was a torus. We construct the other cycles from the tori in similar ways. Again, the cycles will be vertex disjoint since intersecting cycles would imply intersecting tori.

This argument also proves that the only connected submanifolds without a boundary contained in $Y_D$ are tori.

\end{proof}

As an immediate corollary we get the following results:
\begin{itemize}
    \item Sending an instance $(D,\ell)$ of the \DCP problem to the instance $(Y_D,\ell)$ of the \SP problem (where $Y_D$ is the input space and $\ell$ is the number of connected submanifolds) is a reduction.
    \item Sending an instance $(D,\ell)$ of the \DCP problem to the instance $(Y_D,\ell,\ell)$ of the \SoGSR problem (where $Y_D$ is the input space and $\ell$ is the number of connected submanifolds and the total genus) is a reduction.
    \item Sending an instance $(D,\ell)$ of the \DCP problem to the instance $(Y_D,(1: \ell))$ of the Subsurface Recognition problem (where $Y_D$ is the input space and $\ell$ is the number of connected submanifolds of genus $1$ we want to find) is a reduction.
\end{itemize}

Combine these results with the fact that the space $Y_D$ has a pathwidth that is linearly bounded by the pathwidth of $D$ and the ETH-lower bound for the \DCP problem parameterized by pathwidth, and we get \cref{thm:lower_bound} as an immediate consequence. In particular, this theorem also implies the same ETH-lower bounds for simplicial complexes of bounded treewidth, matching our algorithmic upper bounds for the \SP problem and the \SoGSR problem.







\section{Conclusion}

In this paper, we consider the parameterized complexity of several variants of the problem of finding surfaces in 2-dimensional simplicial complexes with respect to the treewidth of the Hasse diagram. We give ETH-optimal algorithms for the \SoGSR and \SP problems. We also give an ETH-based lower bound for Subsurface Recognition and an FPT algorithm for \CSRdot. Several questions surrounding subsurface recognition remain open, such as 
\begin{itemize}
    \item whether the algorithm presented in this paper for \CSR is ETH-optimal;
    \item whether or not the Subsurface Recognition Problem is W[1]-hard when parameterized by the treewidth of the Hasse diagram.
\end{itemize}
Future work could either attempt to find better parameterized algorithms or prove stronger lower bounds for these problems. 



\bibliography{main.bib}

\begin{thebibliography}{10}

\bibitem{ahlfors_sario}
L.V. Ahlfors and L.~Sario.
\newblock {\em Riemann Surfaces}.
\newblock Princeton mathematical series. Princeton University Press, 2015.
\newblock URL: \url{https://books.google.com/books?id=4C4PAAAAIAAJ}.

\bibitem{arnborg_treewidth_hardness}
Stefan Arnborg, Derek~G. Corneil, and Andrzej Proskurowski.
\newblock Complexity of finding embeddings in a k-tree.
\newblock {\em SIAM Journal on Algebraic Discrete Methods}, 8(2):277--284,
  1987.
\newblock \href {http://arxiv.org/abs/https://doi.org/10.1137/0608024}
  {\path{arXiv:https://doi.org/10.1137/0608024}}, \href
  {https://doi.org/10.1137/0608024} {\path{doi:10.1137/0608024}}.

\bibitem{babai_isomorphism}
L\'{a}szl\'{o} Babai.
\newblock Graph isomorphism in quasipolynomial time [extended abstract].
\newblock In {\em Proceedings of the Forty-Eighth Annual ACM Symposium on
  Theory of Computing}, STOC '16, page 684–697, New York, NY, USA, 2016.
  Association for Computing Machinery.
\newblock \href {https://doi.org/10.1145/2897518.2897542}
  {\path{doi:10.1145/2897518.2897542}}.

\bibitem{bagchi_tightness}
Bhaskar Bagchi, Basudeb Datta, Benjamin~A. Burton, Nitin Singh, and Jonathan
  Spreer.
\newblock {Efficient Algorithms to Decide Tightness}.
\newblock In S{\'a}ndor Fekete and Anna Lubiw, editors, {\em 32nd International
  Symposium on Computational Geometry (SoCG 2016)}, volume~51 of {\em Leibniz
  International Proceedings in Informatics (LIPIcs)}, pages 12:1--12:15,
  Dagstuhl, Germany, 2016. Schloss Dagstuhl--Leibniz-Zentrum fuer Informatik.
\newblock URL: \url{http://drops.dagstuhl.de/opus/volltexte/2016/5904}, \href
  {https://doi.org/10.4230/LIPIcs.SoCG.2016.12}
  {\path{doi:10.4230/LIPIcs.SoCG.2016.12}}.

\bibitem{vaagset2021mbc}
Nello Blaser, Morten Brun, Lars~M. Salbu, and Erlend~Raa V{\aa}gset.
\newblock The parameterized complexity of finding minimum bounded chains.
\newblock {\em CoRR}, 2021.
\newblock \href {http://arxiv.org/abs/2108.04563} {\path{arXiv:2108.04563}}.

\bibitem{blaser_hl}
Nello Blaser and Erlend~Raa V{\aa}gset.
\newblock Homology localization through the looking-glass of parameterized
  complexity theory, 2020.
\newblock \href {http://arxiv.org/abs/2011.14490} {\path{arXiv:2011.14490}}.

\bibitem{bodlaender2013ock}
Hans~L. Bodlaender, Pål~Grønås Drange, Markus~S. Dregi, Fedor~V. Fomin,
  Daniel Lokshtanov, and Michał Pilipczuk.
\newblock A $c^k n$ 5-approximation algorithm for treewidth.
\newblock {\em SIAM Journal on Computing}, 45(2):317--378, 2016.
\newblock \href {http://arxiv.org/abs/https://doi.org/10.1137/130947374}
  {\path{arXiv:https://doi.org/10.1137/130947374}}, \href
  {https://doi.org/10.1137/130947374} {\path{doi:10.1137/130947374}}.

\bibitem{burton-finding-two-sphere}
Benjamin Burton, Sergio Cabello, Stefan Kratsch, and William Pettersson.
\newblock {The Parameterized Complexity of Finding a 2-Sphere in a Simplicial
  Complex}.
\newblock In Heribert Vollmer and Brigitte Vall\'ee, editors, {\em 34th
  Symposium on Theoretical Aspects of Computer Science (STACS 2017)}, volume~66
  of {\em Leibniz International Proceedings in Informatics (LIPIcs)}, pages
  18:1--18:14, Dagstuhl, Germany, 2017. Schloss Dagstuhl--Leibniz-Zentrum fuer
  Informatik.
\newblock URL: \url{http://drops.dagstuhl.de/opus/volltexte/2017/7015}, \href
  {https://doi.org/10.4230/LIPIcs.STACS.2017.18}
  {\path{doi:10.4230/LIPIcs.STACS.2017.18}}.

\bibitem{burton_courcelle}
Benjamin Burton and Rodney Downey.
\newblock Courcelle's theorem for triangulations.
\newblock {\em Journal of Combinatorial Theory, Series A}, 146, 03 2014.
\newblock \href {https://doi.org/10.1016/j.jcta.2016.10.001}
  {\path{doi:10.1016/j.jcta.2016.10.001}}.

\bibitem{Burton_morse}
Benjamin~A. Burton, Thomas Lewiner, João Paixão, and Jonathan Spreer.
\newblock Parameterized complexity of discrete {M}orse theory.
\newblock {\em ACM Transactions on Mathematical Software}, 42(1):1–24, Mar
  2016.
\newblock URL: \url{http://dx.doi.org/10.1145/2738034}, \href
  {https://doi.org/10.1145/2738034} {\path{doi:10.1145/2738034}}.

\bibitem{burton_taut}
Benjamin~A. Burton and Jonathan Spreer.
\newblock The complexity of detecting taut angle structures on triangulations.
\newblock {\em Proceedings of the Twenty-Fourth Annual ACM-SIAM Symposium on
  Discrete Algorithms}, Jan 2013.
\newblock URL: \url{http://dx.doi.org/10.1137/1.9781611973105.13}, \href
  {https://doi.org/10.1137/1.9781611973105.13}
  {\path{doi:10.1137/1.9781611973105.13}}.

\bibitem{chernavsky_unrecognizability}
A.V. Chernavsky and V.P. Leksine.
\newblock Unrecognizability of manifolds.
\newblock {\em Annals of Pure and Applied Logic}, 141(3):325--335, 2006.
\newblock Papers presented at the Second St. Petersburg Days of Logic and
  Computability Conference on the occasion of the centennial of Andrey
  Andreevich Markov, Jr.
\newblock URL:
  \url{https://www.sciencedirect.com/science/article/pii/S0168007205001818},
  \href {https://doi.org/https://doi.org/10.1016/j.apal.2005.12.011}
  {\path{doi:https://doi.org/10.1016/j.apal.2005.12.011}}.

\bibitem{cook_complexity}
Stephen~A. Cook.
\newblock The complexity of theorem-proving procedures.
\newblock In {\em Proceedings of the Third Annual ACM Symposium on Theory of
  Computing}, STOC '71, page 151–158, New York, NY, USA, 1971. Association
  for Computing Machinery.
\newblock \href {https://doi.org/10.1145/800157.805047}
  {\path{doi:10.1145/800157.805047}}.

\bibitem{courcelle_msol}
Bruno Courcelle.
\newblock The monadic second-order logic of graphs. {I}. recognizable sets of
  finite graphs.
\newblock {\em Information and Computation}, 85(1):12 -- 75, 1990.
\newblock URL:
  \url{http://www.sciencedirect.com/science/article/pii/089054019090043H},
  \href {https://doi.org/https://doi.org/10.1016/0890-5401(90)90043-H}
  {\path{doi:https://doi.org/10.1016/0890-5401(90)90043-H}}.

\bibitem{Cygan_2015}
Marek Cygan, Fedor~V. Fomin, {\L}ukasz Kowalik, Daniel Lokshtanov, D{\'{a}}niel
  Marx, Marcin Pilipczuk, Micha{\l} Pilipczuk, and Saket Saurabh.
\newblock {\em Parameterized Algorithms}.
\newblock Springer International Publishing, 2015.
\newblock URL: \url{https://doi.org/10.1007%2F978-3-319-21275-3}, \href
  {https://doi.org/10.1007/978-3-319-21275-3}
  {\path{doi:10.1007/978-3-319-21275-3}}.

\bibitem{cycle_packing_paper}
Marek Cygan, Jesper Nederlof, Marcin Pilipczuk, Michal Pilipczuk, Johan M.~M.
  van Rooij, and Jakub~Onufry Wojtaszczyk.
\newblock Solving connectivity problems parameterized by treewidth in single
  exponential time.
\newblock In Rafail Ostrovsky, editor, {\em {IEEE} 52nd Annual Symposium on
  Foundations of Computer Science, {FOCS} 2011, Palm Springs, CA, USA, October
  22-25, 2011}, pages 150--159. {IEEE} Computer Society, 2011.
\newblock \href {https://doi.org/10.1109/FOCS.2011.23}
  {\path{doi:10.1109/FOCS.2011.23}}.

\bibitem{gallier_xu}
Jean Gallier and Dianna Xu.
\newblock {\em A Guide to the Classification Theorem for Compact Surfaces}.
\newblock Springer Berlin Heidelberg, 2013.
\newblock URL: \url{https://doi.org/10.1007%2F978-3-642-34364-3}, \href
  {https://doi.org/10.1007/978-3-642-34364-3}
  {\path{doi:10.1007/978-3-642-34364-3}}.

\bibitem{haken_normal}
Wolfgang Haken.
\newblock Theorie der normalfl{\"a}chen.
\newblock {\em Acta Mathematica}, 105(3):245--375, Sep 1961.
\newblock \href {https://doi.org/10.1007/BF02559591}
  {\path{doi:10.1007/BF02559591}}.

\bibitem{ivanov-hardness}
Sergei Ivanov.
\newblock computational complexity.
\newblock MathOverflow.
\newblock URL: \url{https://mathoverflow.net/q/118428}, \href
  {http://arxiv.org/abs/https://mathoverflow.net/q/118428}
  {\path{arXiv:https://mathoverflow.net/q/118428}}.

\bibitem{markov_insolubility}
A.~Markov.
\newblock The insolubility of the problem of homeomorphy.
\newblock {\em Dokl. Akad. Nauk USSR}, 12(2):218--220, 1958.

\bibitem{rubinstein_sphere}
Hyam Rubinstein.
\newblock The solution to the recognition problem for $\mathbb{S}^3$.
\newblock Lecture, 1992.

\bibitem{thompson_sphere}
Abigail Thompson.
\newblock Thin position and the recognition problem for $\mathbf{S}^3$.
\newblock {\em Mathematical Research Letters}, 1(5):613--630, 1994.

\bibitem{wildberger_zip}
N.J. Wildberger.
\newblock An algebraic {ZIP} proof of the classification | {A}lgebraic
  {T}opology | {NJ Wildberger}.
\newblock YouTube, November 2011.
\newblock \url{https://www.youtube.com/watch?v=-GJs7_NdLm8}.
\newblock URL: \url{https://www.youtube.com/watch?v=-GJs7_NdLm8}.

\end{thebibliography}


\newpage
\appendix

\section{Proofs of Lemma \ref{lem:cell_handle}, \ref{lem:cell_crosscap}, and \ref{lem:cell_boundary}}
\label{appendix:handle_crosscap_boundary}

\begin{proof}[Proof of Lemma \ref{lem:cell_handle}]
We will use our equivalence-preserving moves to show $(\overbar{aba^{-1}b^{-1}XY})=(\overbar{efe^{-1}f^{-1}YX})$.
\begin{align*}
    (\overbar{aba^{-1}b^{-1}XY})&=(\overbar{a^{-1}b^{-1}XYab})&\hfill\text{by (4)}\\
    &=(\overbar{a^{-1}cYaXc^{-1}})&\hfill\text{by Lemma \ref{lem:cell_multiple_boundaries}}\\
    &=(\overbar{Xc^{-1}a^{-1}cYa})&\hfill\text{by (4)}\\
    &=(\overbar{Xc^{-1}d^{-1}Ycd})&\hfill\text{by Lemma \ref{lem:cell_multiple_boundaries}}\\
    &=(\overbar{dXc^{-1}d^{-1}Yc})&\hfill\text{by (4)}\\
    &=(\overbar{dXeYd^{-1}e^{-1}})&\hfill\text{by Lemma \ref{lem:cell_multiple_boundaries}}\\
    &=(\overbar{e^{-1}dXeYd^{-1}})&\hfill\text{by (4)}\\
    &=(\overbar{e^{-1}f^{-1}YXef})&\hfill\text{by Lemma \ref{lem:cell_multiple_boundaries}}\\
    &=(\overbar{efe^{-1}f^{-1}YX})&\hfill\text{by (4)}
\end{align*}
\end{proof}

\begin{proof}[Proof of Lemma \ref{lem:cell_crosscap}]
We first prove the sub-lemma that $(\overbar{aaXY})=(\overbar{bYbX^{-1}})$.
\begin{align*}
    (\overbar{aaXY})&=(\overbar{YaaX})&\hfill\text{by (4)}\\
    &=(\overbar{Yab})+(\overbar{b^{-1}aX})&\hfill\text{by (2)}\\
    &=(\overbar{Yab})+(\overbar{X^{-1}a^{-1}b})&\hfill\text{by (3)}\\
    &=(\overbar{bYa})+(\overbar{a^{-1}bX^{-1}})&\hfill\text{by (4)}\\
    &=(\overbar{bYbX^{-1}})&\hfill\text{by (2)}
\end{align*}
We now use this sub-lemma to prove that $(\overbar{aaXY})=(\overbar{ddYX})$.
\begin{align*}
    (\overbar{aaXY})&=(\overbar{bYbX^{-1}})&\hfill\text{by the sub-lemma}\\
    &=(\overbar{XcY^{-1}c})&\hfill\text{by (3), where $c=b^{-1}$}\\
    &=(\overbar{cXcY^{-1}})&\hfill\text{by (4)}\\
    &=(\overbar{ddYX})&\hfill\text{by the sub-lemma}
\end{align*}
\end{proof}

\begin{proof}[Proof of Lemma \ref{lem:cell_boundary}]
Rearranging $XY$ to $YX$ follows immediately from move (6). As $a$ appears only in the boundary of $a$, we can replace $a$ with new edges edges $ef$ by move (1) without changing the boundary of any other faces. We can then replace $ef$ with a new edge $d$ by move (1).
\end{proof}

\section{Checking the Link Conditions of Simplices}
\label{sec:check_candidacy_proofs}

In this section, we prove that if $\Sigma$ is a candidate solution at $t$ and $\Tilde\Sigma$ is the corresponding candidate solution at $t$, then we can verify that the link of a simplex $\sigma\in X_t$ is admissible or complete if we only have access to $\Tilde\Sigma$. 

\subsection{Checking the Link Condition of Edges}
\label{sec:check_candidacy_edges_proofs}

 There is a simple and natural condition to verify that the link of an edge is admissible or complete. In a candidate solution $\Sigma$, the link of an edge $e$ is a set of vertices, one for each triangle incident to $e$. In the corresponding cell complex $\Tilde\Sigma$, the number of times $e$ appears in the boundary of a face equals the number of triangles incident to $e$ in $\Sigma$.

\begin{lemma}
\label{lem:cc_number_incident_faces}
    Let $\Sigma$ be a candidate solution at $t$, and let $\Tilde\Sigma$ be the corresponding cell complex at $t$. Let $e\in X_t$ be an edge. The number of times the real edge $e$ appears in the boundary of a face in $\Tilde\Sigma$ is the number of triangles in $\Sigma$ with $e$ in their boundary.
\end{lemma}
\begin{proof}
    This is vacuously true for the simplicial complex $\Sigma$. Moreover, the only real edges that have been removed from $\Tilde\Sigma$ are those in $\Sigma\setminus X_t$, so no appearance of $e$ has been removed from $\Tilde\Sigma$. Therefore, the number of times $e$ appears in faces of $\Tilde\Sigma$ is the number of triangles in $\Sigma$ with $e$ in their boundary.
\end{proof}

With this lemma in mind, the algorithm to check if the link of an edge is admissible or complete is simple: just count the number of times it appears in the cell complex. We can do this in $\OO(k)$ time by iterating through the edges in $\Tilde\Sigma$. The following lemmas are immediate.

\edgecompletelink*

\begin{lemma}
\label{lem:edge_link_admissible}
    Let $t$ be a node in the tree decomposition. Let $\Sigma$ be a candidate solution at $t$, and let $\Tilde\Sigma$ be the corresponding cell complex at $t$. Let $e\in S\cap X_t$ be an edge. The link of $e$ in $\Sigma$ is admissible if and only if
    \begin{itemize}
        \item $e\in B$ and the real edge $e$ appears at most once in the boundary of a face in $\Tilde{\Sigma}$, or
        \item $e\notin B$ and the real edge $e$ appears at most twice in the boundary of faces in $\Tilde{\Sigma}$.
    \end{itemize}
    Moreover, this condition can be checked on $\Tilde\Sigma$ in $\OO(k)$ time.
\end{lemma}

\subsection{Checking the Link Conditions of Vertices}
\label{sec:check_candidacy_vertices_proofs}

The key idea for checking the candidacy of vertices is that we can deduce information on the link of the vertex $v$ in a candidate solution $\Sigma$ based on the set of edges that enter $v$ in the corresponding cell complex $\Tilde\Sigma$. In particular, a path in the link of $v$ in $\Sigma$ exactly corresponds to a sequence of edges in the cell complex, and a cycle in the link of $v$ exactly corresponds to a \textit{cyclic} sequence of successors in $\Tilde\Sigma$.

\par 
We will prove the following two lemmas in this section. 

\begin{lemma}
\label{lem:vertex_admissible_link}
    Let $\Sigma$ be a candidate solution at $t$, and let $\Tilde\Sigma$ be the corresponding cell complex at $t$. Let $v\in X_t$ be a vertex. The link $\lk_{\Sigma}{v}$ is admissible if and only if one of the following conditions hold.
    \begin{enumerate}
        \item $v\in B$ and the edges entering $v$ in $\Tilde{\Sigma}$ form a cyclic sequence of successors, or
        \item the edges entering $v$ in $\Tilde{\Sigma}$ form a (possibly empty) collection of (non-cyclic) sequences of successors.
    \end{enumerate}
\end{lemma}

\vertexcompletelink*

We will use the criteria in these lemmas to determine link admissibility and completeness in our algorithm. In particular, if there is a cell complex corresponding to a candidate solution with a vertex that does not satisfy the conditions of Lemma \ref{lem:vertex_admissible_link}, we assume this algorithm is discarded.
\par 
We begin with a simple but helpful observation. 
\begin{lemma}
\label{lem:vertex_neighbors_forgotten}
    Let $\Sigma$ be a candidate solution at $t$, and let $\Tilde\Sigma$ be the corresponding cell complex at $t$. If a vertex $v\in X_t\cap\Sigma$, there is an edge incident to $v$ in $\Tilde{\Sigma}$.
\end{lemma}
\begin{proof}
    A vertex $v$ is first added to $\Sigma$ when a triangle is added; in this case, $v$ is incident to one of the real edges of the triangle and the lemma is true. Moreover, while there is a real edge incident to $v$, the lemma remains true. We only need to verify that $v$ is incident to an edge after all real edges have been removed. 
    \par
    If $v\notin B$, then each real edge $e$ incident to $v$ is not in $B$ either, so $e$ would have been incident to two faces when it was removed by Lemma \ref{lem:cc_number_incident_faces}. In particular, when the last real edge $a$ incident to $v$ was removed, it must have had the form $(aa^{-1}X)$ in the cell complex. (If $a$ did not have this form, then $v$ would have to be incident to another edge. As this edge is not a real edge or a boundary dummy edge, it must be an interior dummy edge. However, as mentioned in Section \ref{sec:remove_edges}, if a vertex is incident to an interior dummy edge, it can be incident to no other edges, including $a$.) Thus, removing $a$ adds an interior dummy edge incident to $v$ and the lemma is true. 
    \par 
    If $v\in B$, then we claim $v$ must have been incident to at least one edge in $\Sigma\cap B$. If $v$ was only incident to edges in not in $B$, each of these edges would have been incident to two faces when it was removed. Therefore, $v$ would be incident to an interior dummy edge, which cannot be the case as it would make the link of $v$ inadmissible. Therefore, $v$ was incident to an edge in $\Sigma\cap B$, and this edge was replaced with a boundary dummy edge when it was removed. If this boundary dummy edge was ever replaced with a merge boundary dummy edge, then this would have been when the other endpoint was forgotten, and one of this new edge's endpoint would still be $v$.
\end{proof}

Our goal in this section is to use sequences of successors in the cell complex $\Tilde\Sigma$ to deduce information on the links of vertices in the simplicial complex $\Sigma$. The following lemma shows what happens to the successors of a single edge when it is removed. 

\begin{lemma}
\label{lem:successors_after_removal}
Let $a$ be an edge in $\Tilde{\Sigma}$ with a pair of successors $b^{-1}\neq a^{-1}$ and $c^{-1}\neq a^{-1}$. After removing the edge $a$ with the rules in Section \ref{sec:remove_simplices}, the edges $b$ and $c^{-1}$ will be successors to one another. 
\end{lemma}
\begin{proof}
    We verify this in each of the cases for removing the edge $a$. We start with the rules in Section \ref{sec:remove_edges} for removing a forgotten edge. We do not need to verify cases 6 and 7 in this section; in case 6, $a^{-1}$ is a successor to $a$, and in case 7, $a$ has a single successor. We also do not need to verify the cases in Section \ref{sec:remove_vertices}. In case 1, $d$ is a successor to $d$. In case 2, both $d_1$ and $d_2$ have a single successor.
    \par 
    We can verify the remaining cases directly, namely
    \begin{enumerate}[font=\bfseries]
        \item $(\overbar{Xa})+(\overbar{a^{-1}Y})=(\overbar{b^{-1}X'a})+(\overbar{a^{-1}Y'^{-1}c})=(\overbar{b^{-1}X'Y'^{-1}c})$,
        \item $(\overbar{XaYa^{-1}})=(\overbar{Xab^{-1}Y'ca^{-1}})=(\overbar{X})(\overbar{b^{-1}Y'c})$,
        \item $(\overbar{XaYa})=(\overbar{c^{-1}X'ab^{-1}Y'a})=(\overbar{Y'^{-1}bc^{-1}X'})$,
        \item $(\overbar{Xa})(\overbar{Ya^{-1}})=(\overbar{c^{-1}X'a})(\overbar{Y'ba^{-1}})=(\overbar{Y'bc^{-1}X'})$,
        \item $(\overbar{Xa})(\overbar{Ya})=(\overbar{c^{-1}X'a})(\overbar{b^{-1}Y'a})=(\overbar{c^{-1}X'Y^{-1}b})$.
    \end{enumerate}
    In case 2 above, note that this holds even if $Y$ only contains a single edge $c$. In this case, $c$ and $c^{-1}$ are a pair of successors to $a$. After removing $a$, then $(c^{-1})^{-1}=c$ is a successor to $c$ in the boundary component $(\overbar{c})$.
\end{proof}

Recall that the goal of this section is to show that sequence of successors in $\Tilde\Sigma$ corresponds to a simple path in $\lk_{\Sigma}v$ and a cyclic sequence of successors in $\Tilde\Sigma$ corresponds to a simple cycle in $\lk_{\Sigma}v$. As a first step, Lemma \ref{lem:surface_face_link} shows that two consecutive edges $(w_1,v)$ and $(v,w_2)$ on a face in $\Tilde\Sigma$ correspond to a simple path or cycle in $\lk_{\Sigma}{v}$, and moreover, that these paths are pairwise disjoint except possibly at their endpoints. The proof of the main lemmas of this section will then show that a sequence of successors corresponds to the concatenation of these individual paths.
\par 
Before we can prove Lemma \ref{lem:surface_face_link}, we need to prove two other lemmas. Lemma \ref{lem:boundary_dummy_edge_path} proves that a boundary dummy edge replaces a path in the input boundary $B$. This lemma should not be surprising, as each boundary dummy edge either replaces an edge in $B$, or replaces two boundary dummy edges, each of which replaced a path in $B$. 

\begin{lemma}
\label{lem:boundary_dummy_edge_path}
    Let $\Sigma$ be a candidate solution at $t$, and let $\Tilde\Sigma$ be the corresponding cell complex at $t$. If $\{w_1,w_2\}\in\Tilde\Sigma$ is a boundary dummy edge, then a simple segment $(w_1,u_1,\ldots,u_{l-1},w_2)$ of $B$ is a subcomplex of $\Sigma$. Moreover, if $\{w_1,w_2\}$ and $\{w_3,w_4\}$ are distinct boundary dummy edges, then the corresponding paths $(w_1,\ldots,w_2)$ and $(w_3,\ldots,w_4)$ are disjoint except possibly at their endpoints.
\end{lemma}
\begin{proof}
    We can prove this induction on the length $l$ of the segment. Initially, a dummy edge $\{w_1,w_2\}$ replaces a real edge $\{w_1,w_2\}\in B$ when it is forgotten, so the statement is true for paths of length $l=1$.
    \par 
    Now assume the statement is true for paths of length less than $l$. The only other time a boundary dummy edge is added is when two dummy edges $\{w_1,v\}$ and $\{w_2,v\}$ sharing a common vertex $v$ are merged. The dummy edges replace paths $B_1=(w_1,\ldots,v)$ and $B_2=(v,\ldots,w_2)$ of $B$. As these paths are disjoint except at $v$, then $(w_1\ldots v\ldots w_2)$ is still a segment of $B$. Moreover, the path $(w_1,\ldots,w_2)$ is disjoint from the other paths corresponding to boundary dummy edges except possibly at $v$. If another path $(u_1,\ldots,u_2)$ intersected $(w_1,\ldots,v)$ at $v$, then $v$ would be incident to three edges in $B$, which cannot be the case as $B$ is a collection of simple cycles. 
\end{proof}

Let $A$ be a face in a cell complex $\Tilde\Sigma$. Lemma \ref{lem:face_triangles} tells us that $A$ corresponds to a set of triangles. 

\begin{lemma}
\label{lem:face_triangles}
    Let $\Sigma$ be a candidate solution at $t$, and let $\Tilde\Sigma$ be the corresponding cell complex at $t$. Each face $A$ in $\Tilde{\Sigma}$ corresponds to a set of triangles in $\Sigma$.
\end{lemma}
\begin{proof}
    We can see this by induction. If $A$ is a triangle in $K$, then this is obviously true. If $A\notin K$, then $A$ is the merge of two smaller faces that shared an edge, each of which corresponds to a set of triangles.
\end{proof}

For a vertex $w$, let $\lk_{A}{w}$ denote the link of $w$ in $\cl A$, the closure of the set of triangles corresponding to $A$. We would like to say that $\lk_{A}{w}$ is a collection of simple paths or a simple cycle. We prove a lemma in this direction.

\begin{lemma}
\label{lem:surface_face_link}
Let $A_1$ be a face in an annotated cell complex $\Tilde{\Sigma}$. Let $(w_1,w,w_2)$ be a segment on the boundary of $A_1$. There is a simple path $P_1$ in $\lk_{A_1}{w}$ such that
\begin{enumerate}
\item if $\{w_1,w\}$ is a real edge, then one endpoint of $P_1$ is $w_1$;
\item if $\{w_1,w\}$ is an interior dummy edge, the endpoints of $P_1$ are equal;
\item  if $\{w_1,w\}$ is a boundary dummy edge that replaces a segment $B_1=\{w_1,u_{l-1}\ldots,u_{1},w\}$ of $B$, then one endpoint of $P_1$ is $u_1$;
\item if $w$ appears in another segment $(w_3,w,w_4)$ on the boundary of a face $A_2$ in $\Tilde{\Sigma}$, then the path $P_2$ in $\lk_{A_2}{v}$ and $P_1$ are disjoint.
\end{enumerate}
\end{lemma}
\begin{proof}
We prove that if the lemma is true on an annotated cell complex $\Tilde{\Sigma}$, then it is true after removing an edge from $\Tilde{\Sigma}$. The condition of the lemma is defined facewise, so adding together two triangle-disjoint annotated cell complexes won't violate the condition.
\par
An annotated cell complex defined by our algorithm is initially a simplicial complex before any edges are removed. The base case of our proof is therefore a face $A_1$ that is a single triangle $\{w_1,w,w_2\}\in K$. The link of $w$ in $\{w_1,w,w_2\}$ is the simple path $(w_1,w_2)$. If $w$ appears in a distinct triangle $\{w_3,w,w_4\}$, the link of $w$ in $\{w_3,w,w_4\}$ is the simple path $(w_3,w_4)$. As the triangles are distinct, the paths $(w_1,w_2)$ and $(w_3,w_4)$ can share at most one endpoint.
\par
Now assume the lemma is true for an annotated cell complex $\Tilde{\Sigma}$. We will show the lemma is true after removing a real edge $\{w,u\}$ from two (possibly the same) faces. Let $(w_1,w,u)$ and $(u,w,w_2)$ be segments on the boundary of faces $A_1$ and $A_2$. By the lemma, there is a path $P_1$ in $\lk_{A_1}{w}$ with $u$ as an endpoint. Likewise, there is a path $P_2$ in $\lk_{A_2}{w}$ with $u$ as an endpoint. Merging the faces $A_1$ and $A_2$ at $\{w,v\}$ will create a face $A_3$ with $(w_1,w,w_2)$ on its boundary. The paths $P_1$ and $P_2$ are disjoint except possibly at their endpoints, so the concatenation of $P_1$ and $P_2$ is a simple path with endpoints $w_1$ and $w_2$. Merging $P_1$ and $P_2$ creates a new path with $u$ in its interior. The vertex $u$ does not appear in any other path $P_3$ in $\lk_{\Sigma}{w}$. If it did, it would have to be the endpoint of $P_3$ as $\Tilde{\Sigma}$ satisfied the condition of the lemma before removing $\{w,u\}$; however, such an annotated cell complex would be discarded by our algorithm as $\{w,u\}$ appears three times. The lemma is true for any other segment $(x,y,z)$ of the boundary of $A_1$ (say) because $x$ and $z$ were connected by a path $P\subset\lk_{A_1}{y}\subset\lk_{A}{y}$.
\par
We now verify this is true after replacing a real edge with an interior dummy edges. If the segment of the boundary with the interior dummy edges is $(w_1,w,w_1)$, then previously $w$ was on a segment of the boundary $(u,w,u)$ where $(w,u)$ is a real edge. So by assumption, there was a path $P_1$ in $\lk_{A}{w}$ with both endpoints equal to $u$.
\par
We now verify this is true for boundary dummy edges. We will prove this by induction on the length of the segment of $B$ in $\Sigma$. If $\{w_1,w\}$ replaces the real edge $\{w_1,w\}$, then this is true by condition 1 of this lemma. We now show this is true after merging two edges into a boundary dummy edge. Let $(w_1,w,u)$ and $(w,u,w_2)$ be segments of a face $A$. We replace the boundary dummy edges $\{w,u\}$ and $\{u,w_2\}$ with a merge boundary dummy edge $\{w,w_2\}$. This replacement does not change the triangles that compose $A$ so $\lk_{A}{w}$ is the same before and after the replacement. If $\{w,u\}$ replaced a segment $(w,u_1,\ldots,u)$ of $B$, then $\{w,w_2\}$ replaces some segment of the boundary $(w,u_1,\ldots,u,\ldots,w_2)$ and the lemma holds.
\end{proof}
\par \noindent

We are now ready to prove the main lemmas of this section.

\begin{proof}[Proof of Lemma \ref{lem:vertex_complete_link}]
    We begin with the case that $v\notin B$. In the first subcase, we claim no edges enter $v$ in $\Tilde\Sigma$ if and only if $\lk_{\Sigma}{v}$ is empty. If no edges enter $v$ in $\Tilde\Sigma$, then by Lemma \ref{lem:vertex_neighbors_forgotten}, $v\notin\Sigma$ and $\lk_{\Sigma}{v}$ is empty. Conversely, if $\lk_{\Sigma}{v}=\emptyset$, then obviously no edges in $\Tilde{\Sigma}$ enter $v$. 
    \par
    Now assume the edges entering $v$ form a cyclic sequence of successors $(\overbar{a_1\ldots a_k})$ with $a_{i}=\{v,w_{i}\}$. A pair of consecutive edges $a_{i}$ and $a_{i+1}$ in the sequence of successors corresponds to a segment $(w_{i},v,w_{i+1})$ on the boundary of a face $A_{i}$ in $\Tilde\Sigma$. If $k=2$ and $a_1$ and $a_2$ are interior dummy edges, then the lemma is immediately true by Lemma \ref{lem:surface_face_link}. So assume instead that the edges $a_{i}$ and $a_{i+1}$ are real edges. By \ref{lem:surface_face_link}, the vertices $w_{i}$ and $w_{i+1}$ are connected by a path $P_i$ in $\lk_{A_i}{v}$. Moreover, each pair of paths $P_i$ and $P_j$ are disjoint for $1\leq i< j\leq k$, except possibly at their endpoints. The path $P_i$ shares one endpoint each with $P_{i-1}$ and $P_{i+1}$, so the concatenation of the paths forms a simple cycle. This simple cycle is exactly $\lk_{\Sigma}{v}$, so the lemma is true.
    \par
    Now assume that $\lk_{\Sigma}{v}$ is a simple cycle. By Proposition \ref{prop:inner_vertex}, the edges entering $v$ in $\Sigma$ form a cyclic sequence of successors. The annotated cell complex $\Tilde{\Sigma}$ is obtained from $\Sigma$ by removing edges $e\notin K[X_t]$. We will show that each of these moves preserves the property of the lemma. Let $(a_1\ldots a_k)$ be the cyclic sequence of successors entering $v$. Removing an edge not incident to $v$ does not change $(\overbar{a_1\ldots a_k})$. If we do remove an edge $a_i$, by Lemma \ref{lem:successors_after_removal} the edges entering $v$ form a cyclic sequence of successors $(\overbar{a_1\ldots a_{i-1}a_{i+1}\ldots a_k})$. If there is a single edge $a$ entering $v$, then forgetting $a$ results in the creation of an interior dummy edge $d$. As discussed above, an interior dummy edge forms a cyclic sequence of successors $(\overbar{d})$.
    \par
    The proof is nearly identical if $v\in B$, except we use Lemma \ref{lem:surface_face_link} to show that $P_1$ and $P_k$ share a single endpoint with other paths and the other endpoints of these paths are $v$'s neighbors in $\beta.$ Likewise, the edges entering $v$ form a sequence of successors $(a_1\ldots a_k)$ in $\Sigma$, and forgetting an edge $a_i$ maintains this property. Forgetting the edges $a_1$ or $a_k$ creates a dummy edge $d$ entering $v$ that had the same successors as $a_1$ or $a_k$.
    \par 
    Finally, these conditions can be checked by building the sequence of successors of all those edges that enter $v$. We can define a graph where the vertices are the edges in $\Tilde\Sigma$ that enter $v$ and the edges of the graph connect two edges in $\Tilde\Sigma$ that are successors. For case 1i) of the lemma, we can check if the graph is empty. For case 1ii), we can check that the graph is a cycle. For case 2), we can check if the graph is a path. The graph can be built by iterating once over the $\OO(k)$ edges in $\Tilde\Sigma$, and the conditions can be checked in $\poly(k)$ time.
\end{proof}
\begin{proof}[Proof of Lemma \ref{lem:vertex_admissible_link}]
    The proof of this lemma is almost identical to the proof of Lemma \ref{lem:vertex_complete_link}.
\end{proof}

In the case that $v$ is about to be forgotten (i.e. $\Sigma$ is a candidate solution at $t'$, the child of the node $t$ that forget $v$), then we can say something stronger about the edges entering $v$. In particular, any real edge incident to $v$ must have already been forgotten, so there will only be dummy edges in $\Tilde\Sigma$ entering $v$. 

\vertexlinkcompleteforget*

\begin{proof}
    We make two observations. First, as $\Sigma$ is a candidate solution at $t$, then we know that $\lk_{\Sigma}{v}$ is complete. Lemma \ref{lem:vertex_complete_link} gives a necessary and sufficient condition for how the edges in $\Tilde\Sigma$ that enter $v$ must behave, depending on whether or not $v\in B$. Second, there are no edges $e$ incident to $v$ in $X_{t'}$; otherwise, $X_t$ would not be a simplicial complex, as it would contain $e$ but not contain the face $v$ of $e$. This contradicts the assumption that the tree decomposition is closed. Therefore, the only edges incident to $v$ in $\Tilde\Sigma$ will be dummy edges.
    \par 
    The vertex $v$ can either be incident to no edges, two interior dummy edges, or two boundary dummy edges; this is because interior dummy edges form a cyclic sequence of successors, so $v$ cannot be incident to any other edges if it is incident to interior dummy edges.
\end{proof}

\end{document}